\documentclass[12pt]{article}
\usepackage{latexsym}
\usepackage{amssymb}
\usepackage{graphicx}
\usepackage{enumerate}
\usepackage{amssymb}
\usepackage{amsmath}
\usepackage{color}
\usepackage{tikz-cd}
\usepackage{hyperref}
\hypersetup{hypertex=true,
	colorlinks=true,
	linkcolor=blue,
	anchorcolor=blue,
	citecolor=blue}
\usepackage{amsthm}
\usepackage{cleveref}
\crefformat{section}{\S#2#1#3} 
\crefformat{subsection}{\S#2#1#3}
\crefformat{subsubsection}{\S#2#1#3}
\def\Z{\mathbb{Z}}
\def\p{\partial}
\def\C{\mathbb{C}}

\def\bfL{{\mathbf{L}}}
\def\ccH{{\check{\mathbb{H}}}}

\def\Res{{\mathrm{Res}}}

\def\bfs{{\mathbf{s}}}
\def\sm{{\mathrm{reg}}}

\def\tE{{\widetilde{E}}}
\def\bfZ{{\mathbf{Z}}}

\def\bfX{\mathbf{X}}
\def\bfJ{{\mathbf{J}}}
\def\pr{{\mathrm{pr}}}
\def\bfI{{\mathbf{I}}}

\def\P{{\mathbb{P}}}
\def\top{{\mathrm{top}}}
\def\Im{\mathrm{Im}}
\def\vG{{\vec{\Gamma}}}
\def\bfx{{{\mathbf{x}}}}

\def\bfx{{\mathbf{x}}}
\def\bfu{{\mathbf{u}}}

\def\bfy{{\mathbf{y}}}
\def\Conf{{\mathrm{Conf}}}
\def\Flat{{\mathrm{Flat}}}
\def\bp{\bar{\p}}

\def\G{{\mathcal{G}}}
\def\cH{{\mathbb{H}}}
\def\id{{\mathrm{id}}}

\def\dbs{{{\bar{\partial}}^{\mathrm{sign}}}}

\def\a{\alpha}
\def\H{{\mathcal{H}}}

\def\bfw{{\mathbf{w}}}
\def\Re{\mathrm{Re}}

\def\bfi{{\mathbf{i}}}
\def\bft{{\mathbf{t}}}

\def\bfk{{\mathbf{k}}}

\def\bfl{{\mathbf{l}}}
\def\b{\beta}

\def\bfz{{\mathbf{z}}}

\def\bfi{{\mathbf{i}}}
\def\bfj{{\mathbf{j}}}
\def\ep{\epsilon}
\def\resch{{\widetilde{[0,\infty)^{|\vG_1|}}}}
\def\l{\lambda}

\def\cpsch{\widetilde{\C^{|\vG_1|}}}
\newcommand{\tmop}[1]{\ensuremath{\operatorname{#1}}}
\numberwithin{equation}{section}

\theoremstyle{plain}
\newtheorem{thm}{Theorem}[section]
\newtheorem*{theorem*}{Theorem}
\newtheorem{lem}[thm]{{{L}}emma}
\newtheorem{cor}[thm]{{C}orollary}
\newtheorem{prop}[thm]{{P}roposition}

\newtheorem{rem}[thm]{Remark}
\newtheorem{exm}[thm]{Example}
\newtheorem{defn}[thm]{{D}efinition}
\def\R{\mathbb{R}}
\def\Xint#1{\mathchoice
	{\XXint\displaystyle\textstyle{#1}}
	{\XXint\textstyle\scriptstyle{#1}}
	{\XXint\scriptstyle\scriptscriptstyle{#1}}
	{\XXint\scriptscriptstyle\scriptscriptstyle{#1}}
	\!\int}
\def\XXint#1#2#3{{\setbox0=\hbox{$#1{#2#3}{\int}$}
		\vcenter{\hbox{$#2#3$}}\kern-.5\wd0}}

\def\dashint{\Xint-}
\newcommand{\be}{\begin{equation}}
	\newcommand{\ee}{\end{equation}}
\newcommand{\ba}{\begin{aligned}}
	\newcommand{\ea}{\end{aligned}}

\begin{document}
	
	\title{Renormalizations in holomorphic field theories on K\"ahler manifolds}
	\author{Minghao Wang\footnote{School of Mathematical Sciences, East China Normal University, 200241, Shanghai, China, wmh18@tsinghua.org.cn}
		\and
		Junrong Yan\footnote{Department of Mathematics, Northeastern University, 02115, Boston, USA, j.yan@northeastern.edu} 
	}
    \date{}
	\maketitle
	\begin{abstract}
    The divergence of Feynman graph integrals is one of the central issues in the study of perturbative quantum field theories. A rigorous formulation of these integrals usually requires renormalization. In this paper, we prove that the Feynman graph integrals arising from holomorphic field theories on closed real-analytic K\"ahler manifolds are convergent with respect to heat-kernel renormalization, Cauchy principal value renormalization, and zeta-function renormalization. Moreover, these three renormalization procedures produce the same value. Our proof is based on the theory of wonderful compactifications in algebraic geometry, which provides a geometric understanding of these integrals. As a consequence, we establish a gauge anomaly formula for these graph integrals.
\end{abstract}
	\tableofcontents
	
	\section{Introduction}
    Quantum field theory originated in particle physics as a framework for describing the interactions of fundamental particles. It has since become a standard tool in physics. Beyond its remarkable success in physics, quantum field theory has also stimulated important developments in geometry, algebra, and analysis, with applications to knot invariants \cite{Witten:1988hf,kontsevich1994feynman}, 3-manifold invariants \cite{Reshetikhin:1991tc,bott2018integral,bott1999integral}, smooth structures on fiber bundles \cite{kontsevich1994feynman}, operad theory \cite{getzler1994operads,kontsevich1999operads}, deformation quantization of Poisson manifolds \cite{kontsevich2003deformation}, and many other topics. In this paper, we focus on the perturbative aspects of quantum field theory.

In perturbative quantum field theory, the quantities produced by the theory are often described in terms of Feynman weights or Feynman graph integrals; see \cref{Feynman weights for finite-dimensional field spaces} for a brief introduction. These integrals often take the form
\begin{equation}\label{divergence issue}
\int_{M^{|\vG_0|}}\check{W}(\vG,\Phi),
\end{equation}
where $M$ is a manifold, called the spacetime of the field theory; $\vG$ is a directed graph with $|\vG_{0}|$ vertices; $\Phi$ is a smooth differential form on $M^{|\vG_0|}$; and $\check{W}(\vG,\Phi)$ is a possibly singular differential form depending on the graph $\vG$ and on $\Phi$.

One of the central issues in perturbative quantum field theory is the possible divergence of the integral \eqref{divergence issue}, caused by singularities of the integrand. To resolve this, one introduces a procedure called renormalization. The basic philosophy can be described as follows:
\begin{enumerate}[(1)]
    \item Consider a continuous family of differential forms $\check{W}_{s}(\vG,\Phi)$ on $M^{|\vG_0|}$, 
    where $s$ is an element of a topological space, such that $\check{W}_{s_{0}}(\vG,\Phi)=\check{W}(\vG,\Phi)$, and $\check{W}_{s}(\vG,\Phi)$ has no singularities when $s\neq s_{0}$.
    \item Choose a decomposition of the integral
    \[
    \int_{M^{|\vG_0|}}\check{W}_{s}(\vG,\Phi)
    =
    W_{s}^{\mathrm{singular}}(\vG,\Phi)+W_{s}^{\mathrm{regular}}(\vG,\Phi),
    \]
    such that $\lim_{s\rightarrow s_{0}}W_{s}^{\mathrm{regular}}(\vG,\Phi)$ exists. Such a choice is called a renormalization scheme.
    \item Define $\lim_{s\rightarrow s_{0}}W_{s}^{\mathrm{regular}}(\vG,\Phi)$ to be the renormalized value of the Feynman graph integral.
\end{enumerate}

The renormalization theory of Feynman graph integrals has attracted considerable interest, including the classical BPHZ approach \cite{Bogoliubov:1957gp,Hepp:1966eg,Zimmermann:1969jj}, the Hopf-algebraic formulation of Connes-Kreimer \cite{Kreimer:1997dp,Connes:1998qv}, and Costello's approach to perturbative renormalization \cite{costellorenormalization}.

The literature contains many renormalization procedures, including momentum cutoffs, heat-kernel renormalization, dimensional regularization, and zeta-function renormalization, among others. Each procedure typically involves an additional choice of renormalization scheme, making it difficult to compare different approaches.

Despite the complexity of quantum field theory in general, for topological field theories such as Chern-Simons theory, which are often of particular interest in mathematics, the following result is well known and has motivated many applications in knot theory and 3-manifold topology; see Axelrod-Singer \cite{axelrod1994chern,Axelrod1991ChernSimonsPT} and Kontsevich \cite{kontsevich1994feynman}.

\begin{theorem*}[Axelrod-Singer and Kontsevich]
Feynman graph integrals in Chern-Simons theory are Lebesgue integrable. In particular, all reasonable renormalization procedures give the same value.
\end{theorem*}

In this paper, we study renormalization in holomorphic field theories, which have attracted considerable interest in recent years and arise, for example, as holomorphic twists of supersymmetric field theories \cite{Johansen:1994aw,Costello2011NotesOS}. In complex dimension one, holomorphic field theories provide an underlying field-theoretic framework for vertex algebras \cite{Li:2016gcb}.

\subsection*{Main results}

Unlike in the topological case, Feynman graph integrals in holomorphic field theories may fail to be Lebesgue integrable. However, recent studies \cite{Li:2011mi,li2021regularized,Costello2015QuantizationOO,Williams:2018ows,budzik2023feynman} suggest that these integrals are nevertheless convergent in the sense that
\begin{equation}\label{convergence property5}
\lim_{s\rightarrow s_{0}}\int_{M^{|\vG_0|}}\check{W}_{s}(\vG,\Phi)
=
\lim_{s\rightarrow s_{0}}W_{s}^{\mathrm{regular}}(\vG,\Phi).
\end{equation}
When the spacetime $M$ is an affine space, the convergence property \eqref{convergence property5} was proved for heat-kernel renormalization by the first named author of this paper \cite{wang2025feynman}. Later, we proved the convergence property \eqref{convergence property5} for real-analytic K\"ahler spacetimes in the setting of Cauchy principal value renormalization \cite{WYFeynman}. In this paper, we generalize these previous results by proving the following theorem.

\begin{thm}[\Cref{thm:CPV HK equivalence} and \Cref{thm:zeta fubini theorem}]\label{main result1}
For holomorphic field theories on real-analytic K\"ahler manifolds, the Feynman graph integrals are convergent under the following renormalization procedures:
\begin{enumerate}[(1)]
    \item heat-kernel renormalization;
    \item Cauchy principal value renormalization;
    \item zeta-function renormalization.
\end{enumerate}
Moreover, these three procedures give the same renormalized value, and the assignment
\[
\Phi \longmapsto \lim_{s\rightarrow s_{0}}\int_{M^{|\vG_0|}}\check{W}_{s}(\vG,\Phi)
\]
defines a current on $M^{|\vG_0|}$, which we denote by $W(\vG,-)$.
\end{thm}

Thus, although the naive Feynman graph integrals in holomorphic field theories may fail to be Lebesgue integrable, they nevertheless define canonical currents, independent of the renormalization procedure used to construct them. 

For an introduction to these three renormalization procedures, see \cref{Feynman graph integrands}.

One particularly important case of \Cref{main result1} is the convergence of heat-kernel renormalization. The idea can be described as follows. We rewrite the integral \eqref{divergence issue} as an integral over a larger space:
\[
\int_{(0,\infty)^{|\vG_{1}|}\times M^{|\vG_0|}}\widetilde W(\vG,\Phi),
\]
where $|\vG_{1}|$ is the number of edges of $\vG$. The convergence then follows from below.

\begin{thm}[\Cref{thm:schwinger smooth extension0}, \Cref{interability of Feynman graph integrand}, and \Cref{smooth extension to schwinger3}]\label{main result2}
There exists a partial compactification of $(0,\infty)^{|\vG_{1}|}$, denoted by $\resch$, such that the differential form
\[
\int_{M^{|\vG_0|}}\widetilde W(\vG,\Phi)
\]
extends smoothly to $\resch$. In particular, it is Lebesgue integrable over $\resch$. We call $\resch$ the partially compactified Schwinger parameter space.
\end{thm}

We refer the reader to \cref{Feynman graph integrands} for the relevant notions.

We point out that \Cref{main result2} extends the affine-space result of \cite{wang2025feynman} to real-analytic K\"ahler manifolds. The method is different: the proof of \Cref{main result2} uses the theory of wonderful compactifications \cite{de1995wonderful,macpherson1998making,hu2003compactification,li2009wonderful}, whereas the affine-space case relies on a more analytic argument.

\Cref{main result2} formulates the convergence problem in terms of a partial compactification of the integration domain. This is reminiscent of the Fulton-MacPherson compactification used in the topological case \cite{axelrod1994chern,kontsevich1994feynman}, with the important distinction that the space to be compactified is different. In fact, our strategy also applies to topological field theories, whereas the compactification strategy used in the topological case does not apply directly to the holomorphic case \cite{wang2024factorization}. 

In topological field theories, integrals over the boundaries of Fulton-MacPherson compactification spaces describe the gauge anomalies of the theories. In the holomorphic setting, a similar phenomenon occurs. In particular, we prove that the failure of the current $W(\vG,-)$ to be closed under the Dolbeault differential $\bar\partial$ is given by integrals over the boundaries of compactified Schwinger spaces:
\begin{prop}[\Cref{prop: anomaly IBP}]
    Let ${\vec{\Gamma}}$ be a holomorphic decorated directed graph, and let $\Phi \in \Omega^{*}(M^{|{\vec{\Gamma}}_0|})$. Then
    \begin{equation}\label{anomaly integral general}
    (\bar{\partial} W)({\vec{\Gamma}}, -)
    =
    \lim_{L\to\infty}
    (-1)^{|{\vec{\Gamma}}_1|}
    \int_{\partial \widetilde{[0, L]^{|{\vec{\Gamma}}_1|}}}
    \int_{M^{|\vG_0|}}\widetilde W(\vG,-).
    \end{equation}
\end{prop}

See \Cref{holomorphic gauge anomaly} for an explanation of this result. We also prove a useful characterization of the boundary integrals in \Cref{cor:kahler anomaly quotient graphs}.

The boundary integrals appearing on the right-hand side of \eqref{anomaly integral general} turn out to be significant in the study of higher-dimensional analogues of vertex algebras, even when $M$ is an affine space. In \cite{minghaoBrianZhengping2025}, these boundary integrals are used to construct homotopy chiral algebras on affine spaces. In \cite{wang2024factorization}, the vanishing of these boundary integrals allows one to construct factorization algebras from topological-holomorphic field theories in the sense of Costello-Gwilliam \cite{costello_gwilliam_2016,costello_gwilliam_2021}.

Finally, we point out potential applications of our results to mirror symmetry. In \cite{bershadsky1994kodaira}, a string-field-theoretic framework for higher-genus topological B-model invariants, which are conjecturally mirror to Gromov-Witten invariants, was introduced under the name Kodaira-Spencer gravity. This is a holomorphic field theory on Calabi-Yau manifolds. In \cite{costello2012quantum}, Costello and Li provided a mathematical framework for Kodaira–Spencer gravity, and its quantization problem was resolved for spacetime $M$ given by an affine space or an elliptic curve \cite{Costello2015QuantizationOO,li2011bcov}. Our theorems show that, once the quantization problem has been solved, the resulting B-model invariants are independent of the choice of renormalization scheme. We hope that \Cref{main result2} will provide new insight into the quantization of Kodaira-Spencer gravity.

\subsection*{Acknowledgments}
	We thank Keyou Zeng, Zhengping Gui, Si Li, Ezra Getzler, Maciej Szczesny, Brian Williams, and Kevin Costello for valuable discussions. Part of this work was completed while the first named author was a postdoctoral associate in the Department of Mathematics and Statistics at Boston University. He gratefully acknowledges the department for its support and stimulating research environment.
 	
	\section{Feynman graph integrals and main results}
	\label{section propagators}
	In this section, we first recall the motivation for and definition of Feynman
weights for finite-dimensional field spaces; see
\cref{Feynman weights for finite-dimensional field spaces}. We then
formulate their generalization to holomorphic field theories in terms of
Feynman graph integrals in \cref{Propagators in holomorphic field theories}
and \cref{Feynman graph integrands}. In \cref{Feynman graph integrands}, we also review three
renormalization methods in the context of holomorphic field theories: heat
kernel renormalization, Cauchy principal value renormalization, and zeta
function renormalization. We also state the main results of this paper there.

	\subsection{Feynman weights for finite-dimensional field spaces}\label{Feynman weights for finite-dimensional field spaces}
	
	In the physical formulation of perturbative quantum field theory, one often
	starts with a heuristic oscillatory integral of the form
	\begin{equation}
		\int_F e^{\frac{i}{\hbar}(S_0+I)}\,d\mu_F .
		\label{eq:Feynman-integral}
	\end{equation}
	Here $F$ is a linear space, possibly infinite-dimensional. The function
	$S_0$ is a nondegenerate quadratic polynomial, while $I$ is a polynomial whose terms have degree at
	least $3$. The symbol $d\mu_F$ denotes the heuristic Lebesgue measure on
	$F$, and $\hbar$ is a positive parameter.
	
	The integral in \eqref{eq:Feynman-integral} is usually only formal. One reason
	is that, when $F$ is infinite-dimensional, there is in general no Lebesgue
	measure on $F$. Nevertheless, motivated by the stationary phase approximation,
	one is led to a schematic Feynman graph expansion of the form
	\begin{equation}
		\int_F e^{\frac{i}{\hbar}(S_0+I)}\,d\mu_F
		\propto
		\exp\Big(
		\sum_{\Gamma\in G_{\mathrm{conn}}}
		\frac{
			i^{|\Gamma_1|+|\Gamma_0|}
			\hbar^{|\Gamma_1|-|\Gamma_0|}
		}{
			|\operatorname{Aut}(\Gamma)|
		}
		W(\Gamma)
		+
		O(\hbar^\infty)
		\Big).
		\label{eq:Feynman-graph-expansion}
	\end{equation}
	Here $G_{\mathrm{conn}}$ denotes the set of connected graphs. For a connected
	graph $\Gamma$, we denote by $\Gamma_0$ and $\Gamma_1$ the sets of vertices
	and edges, respectively, and by $\operatorname{Aut}(\Gamma)$ its automorphism
	group. We write $|\Gamma_0|$, $|\Gamma_1|$, and
	$|\operatorname{Aut}(\Gamma)|$ for their cardinalities. The number
	$W(\Gamma)$ is the Feynman weight assigned to $\Gamma$. We will describe this
	weight in detail below.
	
	The formula \eqref{eq:Feynman-graph-expansion} should be understood
	schematically. The exponential is a formal exponential, and the expression in
	the exponent is an asymptotic expansion rather than a convergent series. The
	general principle of perturbative quantum field theory is that, although the
	left-hand side of \eqref{eq:Feynman-graph-expansion} may be ill-defined, the
	graph expansion in the exponent may still be well-defined as a formal power
	series in $\hbar$. In that case, one studies this formal power series instead
	of the original ill-defined oscillatory integral.
	
	We now recall the definition of the Feynman weight $W(\Gamma)$, assuming that
	$F$ is a finite-dimensional complex vector space. For later use, we define the
	weight for directed graphs, although the resulting number $W(\Gamma)$ does not
	depend on the orientations of the edges.
	\begin{defn}
		A \textbf{directed graph} $\vec{\Gamma}$ consists of a finite ordered set of vertices
		$\vec{\Gamma}_0$, a finite ordered set of edges $\vec{\Gamma}_1$, and two maps
		\[
		t,h:\vec{\Gamma}_1\longrightarrow \vec{\Gamma}_0,
		\]
		assigning to each edge its tail and head. For a vertex $v\in\vec{\Gamma}_0$, we
		denote by $\deg(v)$ the number of half-edges incident to $v$.
	\end{defn}
	
	Let $n=\dim F$, and choose linear coordinates $x^1,\ldots,x^n$ on $F$.
	Assume that
	\[
	S_0=\frac{1}{2}\sum_{i,j=1}^{n} x^i A_{ij}x^j,
	\]
	where $A=(A_{ij})$ is a symmetric nondegenerate $n\times n$ matrix. Equivalently,
	$A$ determines an element of $\operatorname{Hom}(F\otimes F,\mathbb C)$, and
	its inverse determines an element
	\[
	A^{-1}\in F\otimes F.
	\]
	The tensor $A^{-1}$ is often called the propagator.
	
	For each $k\geq 3$, the $k$-th derivative of $I$ at $0$ defines an
	$S_k$-invariant multilinear map
	\[
	D^kI:F^{\otimes k}\longrightarrow \mathbb C .
	\]
	
	Given a directed graph $\vec{\Gamma}$, we place a copy of the propagator
	$A^{-1}\in F\otimes F$ on each edge $e\in\vec{\Gamma}_1$. This gives an element
	\[
	\otimes_{e\in \vec{\Gamma}_1} A^{-1}
	\in
	\otimes_{e\in \vec{\Gamma}_1}(F\otimes F).
	\]
	We also place $D^{\deg(v)}I$ at each vertex $v\in\vec{\Gamma}_0$. This gives an
	element
	\[
	\otimes_{v\in \vec{\Gamma}_0} D^{\deg(v)}I
	\in
	\operatorname{Hom}\Big(
	\otimes_{v\in \vec{\Gamma}_0} F^{\otimes \deg(v)},\mathbb C
	\Big).
	\]
	The incidence relation of the graph, together with the chosen orderings of
	$\vec{\Gamma}_0$ and $\vec{\Gamma}_1$, determines a natural reordering isomorphism
	\begin{equation}
		\otimes_{e\in \vec{\Gamma}_1}(F\otimes F)
		\cong
		\otimes_{v\in \vec{\Gamma}_0} F^{\otimes \deg(v)} .
		\label{eq:reordering-map}
	\end{equation}
	
	\begin{defn}\label{Feynman weight}
		The \textbf{Feynman weight} of $\vec{\Gamma}$, denoted $W(\vec{\Gamma})$, is the contraction
		of
		\[
		\otimes_{e\in \vec{\Gamma}_1}A^{-1}
		\]
		with
		\[
		\otimes_{v\in \vec{\Gamma}_0}D^{\deg(v)}I
		\]
		using the reordering isomorphism \eqref{eq:reordering-map}.
	\end{defn}
	
	\begin{rem}
		Since $A^{-1}$ and the tensors $D^kI$ are symmetric, the number $W(\vec{\Gamma})$
		is independent of the chosen orderings of $\vec{\Gamma}_0$ and $\vec{\Gamma}_1$. It is
		also independent of the orientations of the edges.
	\end{rem}
	
	\begin{exm}
		Let $\vec{\Gamma}$ be the graph
		\[
		\vec{\Gamma}=
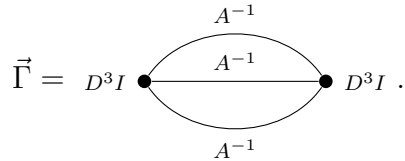
		\begin{tikzpicture}[baseline=-0.5ex, scale=1.0]
			\node[circle, fill=black, inner sep=1.8pt, label=left:{\scriptsize $D^3I$}] (v1) at (0,0) {};
			\node[circle, fill=black, inner sep=1.8pt, label=right:{\scriptsize $D^3I$}] (v2) at (2.4,0) {};
			
			\draw (v1) to[out=55,in=125] node[above] {\scriptsize $A^{-1}$} (v2);
			\draw (v1) -- node[above] {\scriptsize $A^{-1}$} (v2);
			\draw (v1) to[out=-55,in=-125] node[below] {\scriptsize $A^{-1}$} (v2);
		\end{tikzpicture}.
		\]
		Then
		\begin{align*}
			W(\vec{\Gamma})
			&=
			\sum_{i_{1,1}, i_{1,2}, i_{1,3}, i_{2,1}, i_{2,2}, i_{2,3}=1}^{n}
			(A^{-1})^{i_{1,1}i_{2,1}}
			(A^{-1})^{i_{1,2}i_{2,2}}
			(A^{-1})^{i_{1,3}i_{2,3}} \\
			&\quad\quad
			\times
			D^3I(\partial_{i_{1,1}},\partial_{i_{1,2}},\partial_{i_{1,3}})
			D^3I(\partial_{i_{2,1}},\partial_{i_{2,2}},\partial_{i_{2,3}}).
		\end{align*}
	\end{exm}
	When the space of fields $F$ is infinite-dimensional, the Feynman weight
	$W(\vec{\Gamma})$ can often be represented by an integral over a finite-dimensional
	manifold, typically a configuration space of points indexed by the vertices of
	$\vec{\Gamma}$. In such cases, we call $W(\vec{\Gamma})$ the Feynman graph integral
	associated with $\vec{\Gamma}$. We will define this notion in the context of
	holomorphic field theories in the later subsections.
	
	\subsection{Propagators in holomorphic field theories}\label{Propagators in holomorphic field theories}
	In this subsection, we recall the propagator used in holomorphic field theories and
	introduce its regularized version. We keep the notation and conventions
	of \cite[\S 2.1]{WYFeynman}.
	
	Let $M$ be a connected real analytic closed K\"ahler manifold, let $E\to M$ be a real analytic Hermitian
	holomorphic vector bundle, and set
	\[
	\tE:=E\otimes \Lambda^\bullet T_{0,1}^*M .
	\]
	
	The space of smooth sections $\Gamma(\widetilde{E})$ forms the Dolbeault complex of $E$, equipped with the Dolbeault operator $\bar{\partial}_{\widetilde{E}}$ and its formal adjoint $\bar{\partial}_{\widetilde{E}}^*$ with respect to the natural Hermitian inner product. The associated Laplacian is given by
	\[
	\Delta_{\widetilde{E}} = \bar{\partial}_{\widetilde{E}} \circ \bar{\partial}_{\widetilde{E}}^* + \bar{\partial}_{\widetilde{E}}^* \circ \bar{\partial}_{\widetilde{E}}.
	\]
	
	Assume that $E$ is equipped with a non-degenerate holomorphic pairing
	\[
	\omega:E\otimes E\longrightarrow K_M .
	\]
	As in \cite[Definitions 2.4 and 2.5]{WYFeynman}, the pairing $\omega$
	defines the harmonic projection kernel $\H$, the delta distribution
	$\delta$, and the propagator $P$ is given by the following proposition:
	
	\begin{prop}\label{smoothness of propagator}
		There exists a unique distributional section $P$ of $\widetilde{E}\boxtimes\widetilde{E}$, such that
		\begin{equation}\label{ordinary propagator}
			\left\{
			\begin{aligned}
				(\bar\partial_{\tE}\otimes\operatorname{id}+\operatorname{id}\otimes\bar\partial_{\tE})P
				&=\delta-\H,\\
				P&\in \mathrm{Im}(\bar\partial_{\tE}^*\otimes\operatorname{id}).
			\end{aligned}
			\right.
		\end{equation}
		Moreover, let $H_t\in
		C^{\infty}((0,\infty);\Gamma(\tE\boxtimes\tE))$ be the heat kernel associated
		with $\Delta_{\tE}$, normalized by
		\begin{equation}\label{heateq1}
			\left\{
			\begin{aligned}
				\partial_t H_t&=-(\Delta_{\tE}\otimes\operatorname{id})H_t\\
				\lim_{t\to0}H_t&=\delta
			\end{aligned}
			\right.,
		\end{equation}
		then the distributional section $P$ is given by
		\begin{equation}\label{existence of propagator}
			P
			=
			\int_0^{+\infty}
			dt\wedge(\bar\partial_{\tE}^*\otimes\operatorname{id})H_t,
		\end{equation}
		and $P|_{M \times M \setminus \triangle}$ is smooth, where $\triangle := \{ (p, p) \in M \times M : p \in M \}$.
	\end{prop}
	\begin{proof}
		See \cite[Proposition 2.7]{WYFeynman} and \cite[Proposition 2.9]{WYFeynman}.
	\end{proof}
	
	We also introduce the following useful refinement of the propagator.
	\begin{defn}
		\textbf{The propagator in Schwinger space} is defined by
		\begin{equation}\label{propagator t}
			P_t := -dt \wedge \big( \bar{\partial}_{\widetilde{E}}^* \otimes \mathrm{id} \big) H_t + H_t \in \Omega^*\Big( (0,+\infty); \Gamma(\widetilde{E} \boxtimes \widetilde{E}) \Big),
		\end{equation}
		where $\Omega^\bullet\left( (0,+\infty); \Gamma(\widetilde{E} \boxtimes \widetilde{E}) \right)$ denotes the space of $\Gamma(\widetilde{E} \boxtimes \widetilde{E})$-valued smooth forms on $(0,\infty)$, and $H_t$ is the heat kernel. The parameter $t$ is called \textbf{the Schwinger parameter}.
	\end{defn}
	
	The following proposition shows the connection between $P_t$ and $P$:
	
	\begin{prop}
		The propagator $P$ can be represented by
		\begin{equation}
			P
			=
			-\int_0^{+\infty}
			P_t.
		\end{equation}
		Moreover, we have
		\begin{equation}
			(\bar\partial_{\tE}\otimes\operatorname{id}+\operatorname{id}\otimes\bar\partial_{\tE}+d^{(0,\infty)})P_{t}=0,
		\end{equation}
		where $d^{(0,\infty)}$ is the de Rham differential on $\Omega^\bullet\left( (0,+\infty); \Gamma(\widetilde{E} \boxtimes \widetilde{E}) \right)$.
	\end{prop}
	\begin{proof}
		We notice that the second term in \eqref{propagator t} is a zero form with respect to Schwinger parameter $t$, and we use the convention that the integral of a non-top form is zero, so the first assertion follows from \eqref{existence of propagator}. The second assertion follows from \eqref{heateq1}.
	\end{proof}
	
	The main difficulty in constructing Feynman weights in field theories is that
	the propagator $P$ is not a smooth section of
	$\widetilde{E} \boxtimes \widetilde{E}$. Therefore, one cannot use the
	algebraic contraction to define Feynman weights as in Definition
	\ref{Feynman weight}. To define Feynman weights, one usually needs to study
	the singularities of $P$. In fact, these singularities are encoded in the
	asymptotic behavior of $P_t$ as $t \to 0$.
	
	In our setting, for sufficiently small $t>0$, the propagator $P_t$ admits a
	special structure near the diagonal, which we call a regular expression. This
	structure is nontrivial: it is precisely what allows us to define the Feynman
	weights as Cauchy principal values in \cite{WYFeynman}. We first
	introduce the notion of regular expression.
	
	\begin{defn}
		Let $U\subset\mathbb{C}^{n}$ be an open subset, let $E$ be a trivial
		holomorphic vector bundle on $\mathbb{C}^{n}$, and set
		$\widetilde{E}=E\otimes \Lambda^\bullet T_{0,1}^*\mathbb{C}^{n}$. The linear space of \textbf{regular expressions} $S_{\sm}(U\times U)\subset\Omega^{*}\Big((0,\infty);\Gamma\big((\widetilde{E}\boxtimes \widetilde{E})|_{U\times U}\big)\Big)$ is generated by expressions of the form
		\[
		a(\mathbf{y}, d\mathbf{y})\, b,
		\]
		where $a(\mathbf{y}, d\mathbf{y})$ is a polynomial in the variables
		\[
		\mathbf{y} = (y_1, y_2, \ldots, y_n), \quad \text{and} \quad d\mathbf{y} = (dy_1, dy_2, \ldots, dy_n),
		\]
		with
		\[
		y_i = \frac{\overline{z_i} - \overline{w_i}}{t}.
		\]
		Here, $d$ denotes the de Rham differential on $(0,\infty) \times U \times U$, so
		\[
		dy_i = d\left( \frac{\overline{z_i} - \overline{w_i}}{t} \right) = \frac{d(\overline{z_i} - \overline{w_i})}{t} - \frac{(\overline{z_i} - \overline{w_i})\, dt}{t^2}.
		\]
		The factor $b \in \Omega^\bullet\Big((0,\infty); \Gamma\big((\widetilde{E} \boxtimes \widetilde{E})|_{U \times U}\big)\Big)$ is assumed to extend smoothly to a differential form in $\Omega^\bullet\Big([0,+\infty); \Gamma\big((\widetilde{E} \boxtimes \widetilde{E})|_{U \times U}\big)\Big)$.
	\end{defn}
	
	It follows from \cite{WYFeynman} that regular expressions are preserved
	under pullbacks by holomorphic maps. Consequently, this notion can be
	defined globally on any complex manifolds. 
	
	In \cite{WYFeynman}, we prove the following result: 
	\begin{thm}[Theorem 4.3 in \cite{WYFeynman}]\label{thm: Propogator has regular expression}
		Let $\rho$ be the distance function on $M$. Then
		$e^{\frac{\rho^2}{2t}}P_t$ admits a regular expression. i.e., for any $p\in M$, there exists a holomorphic coordinate chart $U$
		containing $p$ such that $ (e^{\frac{\rho^2}{2t}}P_t)|_{U\times U}$ admits a regular expression.
        
        Moreover, any (higher) holomorphic derivatives of $P_t$ satisfy the same property.
	\end{thm}
	
	\subsection{Feynman graph integrals}\label{Feynman graph integrands}
	We briefly recall the graph-theoretic notation and the local functional used
	to define Feynman graph integrals. Our conventions are the
	same as in \cite[\S 2.2]{WYFeynman}.
	
	Let $E\to M$ be a holomorphic vector bundle. Locally, with respect to a
	holomorphic frame $e_1,\ldots,e_r$, every section $\varphi\in\Gamma(E)$
	can be written as
	\[
	\varphi=\sum_{i=1}^r \varphi^i e_i .
	\]
	For a multi-index $\bfj=(j_1,\ldots,j_n)\in\Z_{\ge0}^n$, where $n=\dim(M)$, we write
	\[
	\partial^{\bfj}
	:=
	\frac{\partial^{|\bfj|}}
	{\partial z_1^{j_1}\cdots \partial z_n^{j_n}},
	\quad 
	|\bfj|=j_1+\cdots+j_n .
	\]
	
	\begin{defn} A \textbf{holomorphic differential operator} from $E$ to the trivial holomorphic line $\mathcal{O}_{M}$ is a linear map
		\[
		D:\Gamma(E)\longrightarrow \Gamma(\mathcal{O}_M)
		\]
		which, in local holomorphic coordinates and a local holomorphic frame of
		$E$, is a finite sum of expressions of the form
		\begin{equation} c_{i,\bfj}\,\partial^{\bfj}(\varphi^i), \quad i\in\{1,2,\dots,r\},\quad \bfj=(j_1,\ldots,j_n)\in\Z_{\ge0}^n, \end{equation}
		where the coefficients $c_{i,\bfj}$ are
		holomorphic functions.
	\end{defn}
	
	\begin{defn}
		A \textbf{degree-$k$ holomorphic Lagrangian density} is a linear map
		\[
		I:\Gamma(E)^{\otimes k}\longrightarrow \Gamma(K_M)
		\]
		which, in local holomorphic coordinates and a local holomorphic frame of
		$E$, is a finite sum of expressions of the form
		\begin{equation}\label{defn-I}
			c_{i_1\cdots i_k,\bfj_1\cdots\bfj_k}\,
			dz_1\wedge\cdots\wedge dz_n\,
			\partial^{\bfj_1}(\varphi_1^{i_1})
			\cdots
			\partial^{\bfj_k}(\varphi_k^{i_k}),
		\end{equation}
		where the coefficients $c_{i_1\cdots i_k,\bfj_1\cdots\bfj_k}$ are
		holomorphic functions. 
		
		The density $I$ naturally extends to a linear map from $\Gamma(\widetilde{E}^{\otimes k})$ to $\Gamma(K_M\otimes \Lambda^\bullet T_{0,1}^*M)$ in a way compatible with \eqref{defn-I}, and we shall continue to denote it by $I$.
	\end{defn}
	
	\begin{rem}
		A coordinate-free version of this notion can be found in
		\cite[Definition 2.15]{Williams:2018ows}.
	\end{rem}
	\begin{defn}
		A \textbf{degree-$k$ Dolbeault holomorphic Lagrangian density} is a linear map
		\[
		I:\Gamma(\widetilde{E}^{\otimes k})\longrightarrow \Gamma(K_M\otimes \Lambda^\bullet T_{0,1}^*M)
		\]
		which, in local holomorphic coordinates and a local holomorphic frame of
		$E$, is a finite sum of expressions of the form
		\begin{equation}
			c_{i_1\cdots i_k,\bfj_1\cdots\bfj_k}\,
			dz_1\wedge\cdots\wedge dz_n\,
			\partial^{\bfj_1}(\varphi_1^{i_1})
			\cdots
			\partial^{\bfj_k}(\varphi_k^{i_k}),
		\end{equation}
		where the coefficients $c_{i_1\cdots i_k,\bfj_1\cdots\bfj_k}$ are
		Dolbeault differential forms. 
	\end{defn}

\begin{defn}\label{defn:directed-graph}
A \textbf{decorated directed graph} $\vG$ consists of finite ordered sets
of vertices and edges, denoted by $\vG_0$ and $\vG_1$, respectively, together
with two maps
\[
    t,h:\vG_1\longrightarrow \vG_0,
\]
assigning to each edge its tail and head. Each vertex $v\in\vG_0$ is
decorated by a Dolbeault holomorphic Lagrangian density $I_v$ of degree
$\deg(v)$.

The decorated directed graph $\vG$ is called a \textbf{holomorphic decorated directed graph}
if every vertex decoration $I_v$ is in fact a holomorphic Lagrangian density.

\end{defn}
	
	Given a decorated directed graph $\vG$, we write $$|\vG_0| \quad\text{and}\quad |\vG_1|$$ for the numbers of vertices and edges.
	For each edge $e\in\vG_1$, let $t_e$ denote the corresponding Schwinger
	parameter. We have
	\[
	\bigotimes_{e\in\vG_1}P_{t_e}
	\in
	\Omega^\bullet\left(
	(0,\infty)^{|\vG_1|};
	\Gamma\left(\boxtimes_{e\in\vG_1}(\widetilde E\boxtimes\widetilde E)\right)
	\right),
	\]
	and
	\[
	\bigotimes_{v\in\vG_0}I_v:
	\Gamma\Big(
	\prod_{v\in\vG_0}M^{\deg(v)};
	\boxtimes_{v\in\vG_0}\widetilde E^{\boxtimes\deg(v)}
	\Big)
	\longrightarrow
	\Gamma\left(
	\boxtimes_{v\in\vG_0}(K_M\otimes\Lambda^\bullet T_{0,1}^*M)
	\right).
	\]
	The directed graph determines a canonical reordering map
	\[
	\tau^\vG:(M\times M)^{|\vG_1|}
	\longrightarrow
	\prod_{v\in\vG_0}M^{\deg(v)},
	\]
	which assigns the two endpoints of each edge to the corresponding vertex
	factors. Hence
	\[
	\tau^\vG_*\Big(\bigotimes_{e\in\vG_1}P_{t_e}\Big)
	\in
	\Omega^\bullet\Big(
	(0,\infty)^{|\vG_1|};
	\Gamma\left(\boxtimes_{v\in\vG_0}\widetilde E^{\boxtimes\deg(v)}\right)
	\Big).
	\]
	
	For any $\Phi\in\Omega^\bullet(M^{|\vG_0|})$, set
	\begin{equation}\label{Feynman graph integrand}
	\widetilde W(\vG,\Phi)
	:=
	\Big(
	\big(\bigotimes_{v\in\vG_0}I_v\big)
	\circ \tau^\vG_*
	\Big)
	\Big(\bigotimes_{e\in\vG_1}P_{t_e}\Big)
	\wedge \Phi .
	\end{equation}
	This is a smooth form on
	\[
	(0,\infty)^{|\vG_1|}\times M^{|\vG_0|}.
	\]
	
	Formally, the Feynman weight of $\vec{\Gamma}$, which we also call the
	Feynman graph integral, is given by
	\begin{equation}\label{formal Feynman graph integral}
		\int_{(0,\infty)^{|\vG_1|}\times M^{|\vG_0|}}\widetilde W(\vG,\Phi).
	\end{equation}
	In general, this integral does not converge in the Lebesgue sense; see
	\cite[Example 1]{wang2025feynman}.
	
	In this subsection, we recall three different approaches to defining the
	integral in \eqref{formal Feynman graph integral}: heat kernel renormalization, Cauchy principal value
	renormalization, and zeta function
	renormalization. The main result of this paper is that
	\eqref{formal Feynman graph integral} converges under all three approaches,
	and that the resulting values agree.
	
	\subsubsection{Heat kernel renormalization}
	We first observe that
	\[
	\int_{M^{|\vG_0|}}\widetilde W(\vG,\Phi)
	\in \Omega^\bullet\big((0,\infty)^{|\vG_1|}\big)
	\]
	is a smooth form on $(0,\infty)^{|\vG_1|}$, since
	$M^{|\vG_0|}$ is compact and the integrand is smooth. The heat kernel
	renormalization of the Feynman graph integral is defined as follows.
	
	\begin{defn}\label{defn: heat formulation}
		Let $\vec{\Gamma}$ be a decorated directed graph and let
		$\Phi\in\Omega^\bullet(M^{|\vG_0|})$. If the limit
		\begin{equation}
			\label{heat kernel renormalization} \lim_{\substack{\epsilon\to 0\\ L\to \infty}}
			\int_{(\epsilon,L)^{|\vG_1|}}
			\int_{M^{|\vG_0|}}\widetilde W(\vG,\Phi)    
		\end{equation}
		exists, then we call it the \textbf{heat kernel renormalization of the
			Feynman graph integral}, and denote it by $W^{\mathrm{HK}}(\vG,\Phi)$.
	\end{defn}
	\begin{rem}
		For general quantum field theories, the limit in
		\eqref{heat kernel renormalization} may fail to exist because of
		singular behavior as $\epsilon\to 0$. In such cases, one usually studies
		the singular behavior as $\epsilon\to 0$ and subtracts the singular part.
		This requires a choice of renormalization scheme; see
		\cite{costellorenormalization} for details. Such an additional choice
		makes the construction less canonical. In holomorphic field theories,
		however, the convergence result proved in this paper shows that no such
		choice is needed.
	\end{rem}
	
	The existence of the heat kernel renormalization of Feynman graph integrals
	will be proved by using suitable partial compactifications of Schwinger
	parameter spaces. We now introduce the relevant notion.
	\begin{defn}\label{compactification}
		Given a directed graph $\vec{\Gamma}$, a \textbf{partial compactification}
		of the Schwinger parameter space is a manifold with corners, denoted by
		$\resch$, together with a smooth open embedding
		\[
		i:(0,\infty)^{|\vG_{1}|}\rightarrow\resch,
		\]
		and a proper smooth map
		\[
		\pi_\R^{\vG}:\resch\rightarrow [0,\infty)^{|\vG_{1}|},
		\]
		such that $\pi_\R^{\vG}\circ i$ is the canonical embedding from
		$(0,\infty)^{|\vG_{1}|}$ to $[0,\infty)^{|\vG_{1}|}$, and \[
		\resch=\partial \resch\cup i((0,\infty)^{|\vG_{1}|}).\]
	\end{defn}
	
	Definition
	\ref{compactification} is introduced in order to control the singular
	behavior of
	\[
	\int_{(\epsilon,L)^{|\vG_1|}}
	\int_{M^{|\vG_0|}}\widetilde W(\vG,\Phi)
	\]
	as $\epsilon\to 0$. The naive compactification
	$[0,\infty)^{|\vG_{1}|}$ is not sufficient for this purpose. Instead, we
	will prove the following result in \cref{sec: heat formulation},
	which is one of the main results of this paper.
	\begin{thm}\label{thm:schwinger smooth extension0}
		For any decorated directed graph $\vG$, there exists a partial
		compactification of the Schwinger parameter space $\resch$ such that, for
		every $\Phi\in\Omega^\bullet(M^{|\vG_0|})$, the form
		\[
		\int_{M^{|\vG_0|}}\widetilde W(\vG,\Phi)
		\]
		admits a smooth extension to the partially compactified Schwinger space
		$\resch$. We denote this extension by $\widehat W(\vG,\Phi)$. Moreover,
		the assignment
		\[
		\Phi\in\Omega^\bullet(M^{|\vG_0|})
		\longmapsto
		\widehat W(\vG,\Phi)\in\Omega^\bullet\big(\resch\big)
		\]
		defines a continuous linear map between topological vector spaces.
	\end{thm}
	The proof of \Cref{thm:schwinger smooth extension0} will be given in
	\Cref{sec: heat formulation}, more precisely, see
	\Cref{smooth extension to schwinger3}. 
	
	The following result can be derived from Theorem \ref{thm:schwinger smooth extension0}.
	
	\begin{cor}\label{interability of Feynman graph integrand}
		For every decorated directed graph $\vG$ and $\Phi\in\Omega^\bullet(M^{|\vG_0|})$, the differential form $\widehat W(\vG,\Phi)$ is Lebesgue integrable over $(0,\infty)^{|\vG_{1}|}$.
	\end{cor}
	\begin{proof}
Since only the top-degree component of $\widehat W(\vG,\Phi)$ enters the
integral, throughout the proof we discard the lower-degree components and write simply $\widehat W_{\top}(\vG,\Phi)$ for this
top-degree component. The estimates for the lower-degree components are more
subtle because of the possible contribution of harmonic forms. In the
top-degree component, however, harmonic forms do not contribute, since
$\bar\partial^*$ annihilates them.
    
		We notice the following decomposition:
		\[
		(0,\infty)^{|\vG_1|}
		=
		\bigcup_{S\subseteq\vG_1}
		(0,1]^{|\vG_1\setminus S|}
		\times[1,\infty)^{|S|}.
		\]
		It suffices to prove that  $\widehat W_{\top}(\vG,\Phi)$ is integrable over $(0,1]^{|\vG_1\setminus S|}\times[1,\infty)^{|S|}$ for any $S\subseteq\vG_1$.
		
		We fix $S\subseteq\vG_1$. Let $\vG'$ be the subgraph whose vertex set $\vG'_{0}=\vG_{0}$ and edge set $\vG'_{1}=\vG_1\setminus S$. Over $(0,1]^{|\vG'_1|}\times[1,\infty)^{|S|}$, we have
		\begin{equation}
			\label{eq:hat-W-decomposition}
			\begin{aligned}
				\widehat W(\vG,\Phi)
				&=\int_{M^{|\vG_0|}}
				\Big(
				\big(\bigotimes_{v\in\vG_0} I_v\big)
				\circ \tau^\vG_*
				\Big)
				\Big(
				\bigotimes_{e\in\vG_1} P_{t_e}
				\Big)
				\wedge \Phi                                      \\
				&=\pm\int_{M^{|\vG_0|}}
				\Big(
				\big(\bigotimes_{v\in\vG_0} I_v\big)
				\circ \tau^\vG_*
				\Big)
				\Big(
				\big(\bigotimes_{e'\in\vG_1'} P_{t_{e'}}\big)
				\otimes
				\big(\bigotimes_{e\in S} P_{t_e}\big)
				\Big)
				\wedge \Phi .
			\end{aligned}
		\end{equation}
		We notice that $I_{v}$ is a Dolbeault holomorphic Lagrangian density, and $P_{t_e}$ is smooth over $[1,\infty)\times M\times M$, so the last line in \eqref{eq:hat-W-decomposition} can be written as a finite sum of 
		\begin{equation}
			\widehat W\big(\vG',\Psi(\bft')\big), \quad
			\bft'=(t_e)_{e\in S}\in [1,\infty)^{|S|},  
		\end{equation}
		where the directed graph $\vG'$ is equipped with some decoration, and \[\Psi(\bft')\in \Omega^\bullet\Big([1,\infty)^{|S|}\times M^{|\vG'_0|})\Big)\]
		is locally of the form
		$
		\prod_{e\in S}\Big( (D_{e} P_{t_e})(\bfz_{h(e)},\bfz_{t(e)})\Big)\wedge \Phi, \quad (\bfz_{v})_{v\in \vG'_{0}}\in U^{|\vG'_{0}|},
		$
		where $D_{e}$ is a holomorphic differential operator from $E\boxtimes E\big|_{U\times U}$ to $\mathcal{O}_{M\times M}|_{U\times U}$.
		
		Since the $dt$-component contains no harmonic part, as the harmonic part is
annihilated by $\bar\partial^*$, we recall the following estimate; see
\cite[Proposition 2.37]{berline2003heat}. For every $k\ge0$, there exist
constants $C_k,c_k>0$ such that
\[
        \big\|
        \iota_{\frac{\partial}{\partial t}}P_{t_e}
        \big\|_{C^k(M\times M)}
        \le
        C_ke^{-c_kt_e},
        \quad
        t_e\in[1,\infty).
\]
		Consequently, for every $k\ge0$, there exist constants $C_k,c_k>0$ such
		that
		\[
		\|\Psi_{\top}(\bft')\|_{C^k(M^{|\vG_0|})}
		\le
		C_ke^{-c_k\sum_{e\in S}t_{e}},
		\]
		where $\Psi_{\top}(\bft')$ is the top form part of $\Psi$ as a differential form on Schwinger parameter space $[1,\infty)^{|S|}$ valued on $\Omega^\bullet\big([1,\infty)^{|S|}\times M^{|\vG'_0|}\big)$.
        
		By Theorem \ref{thm:schwinger smooth extension0}, there exists a partial compactification $\widetilde{[0,\infty)^{|\vG_1'|}} $, such that the linear map 
		\[\Psi\in \Omega^\bullet( M^{|\vG_0|}\big) \longmapsto
		\widehat W\big(\vG',\Psi\big)\in \Omega^\bullet\big(\widetilde{[0,\infty)^{|\vG_1'|}}\big)
		\]
		is continuous. Therefore, we have
		\[
		\widehat W\big(\vG',\Psi(\bft')\big)\in \Omega^\bullet\Big(\widetilde{[0,1]^{|\vG_1'|}}\times[1,\infty)^{|S|}\Big),
		\]
		where $\widetilde{[0,1]^{|\vG_1'|}}=(\pi_\R^{\vG})^{-1}([0,1]^{|\vG_1'|})$ and $\pi_\R^{\vG}:\widetilde{[0,\infty)^{|\vG_1'|}}\to [0,\infty)^{|\vG_1'|}$ is the canonical proper smooth map defined in Definition \ref{compactification}. Moreover, since $\widetilde{[0,1]^{|\vG_1'|}}$ is compact, after we choose a Riemannian structure on $\widetilde{[0,1]^{|\vG_1'|}}$, there exist $r\in\mathbb N$ and $C,C_{r},c_{r}>0$ such that 
		\begin{equation*}
			\|\widehat W_{\top}\big(\vG',\Psi(\bft')\big)\|_{C^{0}(\widetilde{[0,1]^{|\vG_1'|}})}\leq C\|\Psi_{\top}(\bft')\|_{C^r(M^{|\vG_0'|})}\le
			 C_re^{-c_r\sum_{e\in S}t_{e}}.   
		\end{equation*}
		
		Since $\widehat W(\vG,\Phi)$ is a finite sum of terms like $W\big(\vG',\Psi(\bft')\big)$, for some $C',c'>0$, 
		\begin{equation}\label{exponential decay}
			\|\widehat W_{\top}(\vG,\Phi)\|_{C^{0}(\widetilde{[0,1]^{|\vG_1'|}})}\le
			C'e^{-c'\sum_{e\in S}t_{e}},\quad \bft'=(t_e)_{e\in S}\in [1,\infty)^{|S|}.
		\end{equation}
		Therefore, $\widehat W(\vG,\Phi)$ is Lebesgue integrable over $\widetilde{[0,1]^{|\vG_1'|}}\times[1,\infty)^{|S|}$.
		
		Finally, We notice that $(0,1]^{|\vG_1'|}$ differs from $\widetilde{[0,1]^{|\vG_1'|}}$ by a zero measure set, so $\widehat W(\vG,\Phi)$ is Lebesgue integrable over $(0,1]^{|\vG_1\setminus S|}\times[1,\infty)^{|S|}$. 
	\end{proof}
	
	\begin{cor}\label{thm:graph integral current}
		Let $\vec{\Gamma}$ be a decorated directed graph and let
		$\Phi\in\Omega^\bullet(M^{|\vG_0|})$. Then the heat kernel renormalization of the
		Feynman graph integral $W^{\mathrm{HK}}(\vG,\Phi)$ exists. Moreover, the linear map
        \[
        \Phi\in\Omega^\bullet(M^{|\vG_0|})\longmapsto W^{\mathrm{HK}}(\vG,\Phi)
        \]
        is continuous.
	\end{cor}
	\begin{proof}
		We notice that 
		\[ \int_{(\epsilon,L)^{|\vG_1|}}
		\int_{M^{|\vG_0|}}\widetilde W(\vG,\Phi)=\int_{(0,\infty)^{|\vG_1|}}
		\mathbf{1}_{(\epsilon,L)^{|\vG_1|}}\int_{M^{|\vG_0|}}\widetilde W(\vG,\Phi), \]
		where $\mathbf{1}_{(\epsilon,L)^{|\vG_1|}}$ is the characteristic function.
		Then by Corollary \ref{interability of Feynman graph integrand} and dominated convergence theorem, we have 
        \[
        W^{\mathrm{HK}}(\vG,\Phi)=\lim_{\substack{\epsilon\to 0\\ L\to \infty}}
			\int_{(\epsilon,L)^{|\vG_1|}}
			\int_{M^{|\vG_0|}}\widetilde W(\vG,\Phi)= \int_{(0,\infty)^{|\vG_1|}}\widehat W(\vG,\Phi).
        \]

        Let $\{\epsilon_i\}_{i=1}^{\infty}$ and $\{L_i\}_{i=1}^{\infty}$ be two sequences of positive numbers such that
\[
    \lim_{i\to\infty}\epsilon_i=0,
    \quad
    \lim_{i\to\infty}L_i=\infty.
\]
For each $i$, define a continuous linear map
\[
    T_i:\Omega^\bullet(M^{|\vG_0|})\longrightarrow \mathbb{C}
\]
by
\[
    T_i(\Phi)
    =
    \int_{(\epsilon_i,L_i)^{|\vG_1|}}
    \int_{M^{|\vG_0|}}\widetilde W(\vG,\Phi).
\]
Then the sequence $\{T_i\}_{i=1}^{\infty}$ converges pointwise to
$W^{\mathrm{HK}}(\vG,-)$. By the Banach-Steinhaus theorem for Fréchet spaces,
$\{T_i\}_{i=1}^{\infty}$ is equicontinuous. Since the pointwise limit of an equicontinuous family of linear maps is continuous, it follows that
$W^{\mathrm{HK}}(\vG,-)$ is a continuous linear map.
	\end{proof}

    \subsubsection{Cauchy principal value
	renormalization}

The heat kernel renormalization provides a rigorous formulation of the
Feynman graph integral as an iterated Lebesgue integral
\[
    W^{\mathrm{HK}}(\vG,\Phi)
    =
    \int_{(0,\infty)^{|\vG_1|}}
    \int_{M^{|\vG_0|}}\widetilde W(\vG,\Phi).
\]
One might expect that the iterated integral with the opposite order of
integration,
\begin{equation}\label{Heuristic iterated integral}
    \int_{M^{|\vG_0|}}
    \int_{(0,\infty)^{|\vG_1|}}
    \widetilde W(\vG,\Phi),
\end{equation}
also exists as a Lebesgue integral. However, this is not true; see
\cite[Example 1]{wang2025feynman}. Instead, the Cauchy principal
value renormalization provides a rigorous formulation of
\eqref{Heuristic iterated integral}.

\begin{defn}
		Let $\vG$ be a directed graph. The configuration space of $M$ with respect to $\vG$ is defined by
		\[
		\Conf_{\vG_{0}}(M) = \big\{ (p_1, \dots, p_{|\vG_{0}|}) \in M^{|\vG_{0}|}:\, p_i \neq p_j \text{ for } i \neq j \big\}.
		\]
	\end{defn}

\begin{lem}
    Let $\vec{\Gamma}$ be a decorated directed graph and let
    $\Phi\in\Omega^\bullet(M^{|\vG_0|})$. The Lebesgue integral
    \be\label{eq: check W}
        \left.
        \check{W}(\vG,\Phi):=\int_{(0,\infty)^{|\vG_1|}}
        \widetilde W(\vG,\Phi)
        \right|_{\Conf_{\vG_0}(M)}
    \ee
    converges and defines a smooth differential form on
    $\Conf_{\vG_0}(M)$.
\end{lem}

\begin{proof}
    This follows from Proposition~\ref{smoothness of propagator}.
\end{proof}

\begin{defn}
		 A \textbf{fake distance function} $\tilde{\rho}:M\times M\rightarrow\mathbb{R}$ is a non-negative function satisfying:
		\begin{enumerate}[(1)]
			\item $\tilde{\rho}^{2}$ is smooth and $\tilde{\rho}^{-1}(0)=\triangle$, where $\triangle:=\{(p,p)\in M\times M : p\in M\}$.
			\item There exists an open neighborhood $U\subset M\times M$ of $\triangle$ such that \[\tilde{\rho}^{2}|_{U}=\rho^{2}|_{U},\]
			where $\rho$ is the distance function on $M$.
		\end{enumerate}
	\end{defn}
Let $\tilde{\rho}$ be a fake distance function on $M$. For $\epsilon>0$, we define
\[
    \Conf^{\epsilon}_{\vG_{0}}(M;\tilde{\rho})
    =
    \big\{
    (p_1,\dots,p_{|\vG_{0}|})\in M^{|\vG_{0}|}
    :
    \prod_{i<j\in \vec{\Gamma}_{0}}
    \tilde{\rho}^{2}(p_i,p_j)>\epsilon
    \big\}.
\]

\begin{defn}
    Let $\vec{\Gamma}$ be a decorated directed graph and let
    $\Phi\in\Omega^\bullet(M^{|\vG_0|})$. If the limit
    \begin{equation}\label{def:W CPV}
        \lim_{\epsilon\to 0}
        \int_{\Conf^{\epsilon}_{\vG_{0}}(M;\tilde{\rho})}
        \int_{(0,\infty)^{|\vG_1|}}
        \widetilde W(\vG,\Phi)
    \end{equation}
    exists and is independent of the choice of the fake distance function
    $\tilde{\rho}$, then we call it the \textbf{Cauchy principal value
    renormalization of the Feynman graph integral} and denote it by
    $W^{\mathrm{CPV}}(\vG,\Phi)$.
\end{defn}

In \cite[Theorem 2.21]{WYFeynman}, the following result has been proved:

\begin{thm}
    Let $\vec{\Gamma}$ be a decorated directed graph and let
		$\Phi\in\Omega^\bullet(M^{|\vG_0|})$. Then the Cauchy principal value renormalization of the
		Feynman graph integral $W^{\mathrm{CPV}}(\vG,\Phi)$ exists. Moreover, it defines a generalized Cauchy principal value in the sense of Definition \ref{defn:global generalized cpv}:
        \[
        W^{\mathrm{CPV}}(\vG,\Phi)
		=\dashint_{\widetilde{\Conf}_{\vG_0}(M)}
		\int_{(0,\infty)^{|\vG_1|}}
        \widetilde W(\vG,\Phi),
        \]
        where $\widetilde{\Conf}_{\vG_0}(M)$ is the Fulton-MacPherson compactification of $\Conf_{\vG_{0}}(M)$, see \cite{17c23791-d52f-3aa0-a5a9-5abc4d54669e}.
\end{thm}
	
\subsubsection{Zeta function renormalization}
	
	Zeta function renormalization can viewed as a rigorous formulation of the formal double integral \eqref{formal Feynman graph integral} by analytical continuation. Let us recall the relevant notions.

    \begin{defn}
        For any $s\in \mathbb{C}$, \textbf{The zeta regularized propagator in Schwinger space} is defined by
		\begin{equation}\label{zeta propagator t}
			P^s_{t}=\frac{t^{s}}{\Gamma(s+1)}P_t  \in \Omega^\bullet\Big( (0,+\infty); \Gamma(\widetilde{E} \boxtimes \widetilde{E}) \Big).
		\end{equation}
    \end{defn}

    The following lemma is critical in zeta function regularization:
    \begin{lem}\label{lemma in zeta}
         When $\Re(s)>\dim(M)+2$, $\iota_{\frac{\partial}{\partial t}}P_{t}^{s}$ can be extended to a continuous section on $[0,\infty)\times M\times M$. 
        Moreover, For any non-negative integer $k$ and $\Re(s)>\dim(M)+k+2$, there exist constants $c_{k},C_{k}>0$, such that
        \[
        \|\iota_{\frac{\partial}{\partial t}}P_{t}^{s}\|_{C^{k}(M\times M)}\leq C_ke^{-c_kt}, \quad t\in[0,\infty).
        \]
    \end{lem}
    \begin{proof}
        when $t>0$, $\iota_{\frac{\partial}{\partial t}}P_{t}^{s}$ is smooth since the heat kernel is smooth. To prove the continuity at $t=0$, we notice the following estimate (see \cite[Theorem 3.5]{ludewig2019strong}): For any non-negative integer $k$, there exists a constant $C'_{k}$, such that
        \[
        \|\iota_{\frac{\partial}{\partial t}}P_{t}\|_{C^{k}(M\times M)}\leq C'_{k} t^{-\dim(M)-k-2}\|e^{-\frac{\rho^{2}}{t}}\|_{C^{0}(M\times M)}\leq C'_{k}t^{-\dim(M)-k-2}, \quad t\in (0,1],
        \]
        therefore, for $\Re(s)>\dim(M)+k+2$, we have 
        \be\label{eq: zeta heat kernel small time}
        \|\iota_{\frac{\partial}{\partial t}}P_{t}^{s}\|_{C^{k}(M\times M)}\leq Ct^{\Re(s)-n-2}, \quad t\in (0,1].
        \ee
        This shows that for $\Re(s)>\dim(M)+k+2$, the $k$-th derivative of $\iota_{\frac{\partial}{\partial t}}P_{t}^{s}$ can be extended to a continuous section on $[0,\infty)\times M \times M$, and equals $0$ on $\{0\}\times M\times M$.

        For $\Re(s)>\dim(M)+k+2$, combining the continuity result with \cite[Proposition 2.37]{berline2003heat}, there exists a constant $C_{k},c_{k}>0$, such that
        \be\label{eq: zeta heat kernel large time}
        \|
        \iota_{\frac{\partial}{\partial t}}P^s_{t}
        \|_{C^k(M\times M)}
        \le
        C_{k}e^{-c_{k}t},
        \quad
        t\in[1,\infty).
        \ee
        The lemma then follows from \eqref{eq: zeta heat kernel small time} and \eqref{eq: zeta heat kernel large time}. 
    \end{proof}

Let $\vec{\Gamma}$ be a decorated directed graph and let
    $\Phi\in\Omega^\bullet(M^{|\vG_0|})$. Similar to \eqref{Feynman graph integrand}, for $\mathbf{s}=(s_{e})_{e\in \vG_{1}}\in \mathbb{C}^{|\vG_{1}|}$, we define
\[
	\widetilde W(\vG,\Phi,\mathbf{s})
	:=
	\big(
	(\otimes_{v\in\vG_0}I_v)
	\circ \tau^\vG_*
	\big)
	\big(\otimes_{e\in\vG_1}P_{t_e}^{s_{e}}\big)
	\wedge \Phi .
	\]
    The following result is a direct consequence of \Cref{lemma in zeta}:
    \begin{cor}\label{zeta at large real part}
        Let $\vec{\Gamma}$ be a decorated directed graph and let
    $\Phi\in\Omega^\bullet(M^{|\vG_0|})$. Let $\mathbf{s}=(s_{e})_{e\in \vG_{1}}\in \mathbb{C}^{|\vG_{1}|}$ and $\Re(s_{e})\gg0$ for any $e\in \vG_{1}$. The Lebesgue integral
    \be\label{eq: Psi ZF}
        W^{\mathrm{ZF}}(\vG,\Phi,\bfs):=\int_{(0,\infty)^{|\vG_1|}\times M^{|\vG_0|}}\widetilde W(\vG,\Phi,\mathbf{s})
    \ee
    converges. Moreover, the integral depends on $\mathbf{s}$ holomorphically.
    \end{cor}

   \begin{defn}
Let $\vG$ be a decorated directed graph, and let
$\Phi\in\Omega^\bullet(M^{|\vG_0|})$. Set
\[
        \cH:=\{s\in\mathbb C:\Re(s)\geq0\}.
\]
Suppose that $W^{\mathrm{ZF}}(\vG,\Phi,\bfs)$, initially defined by
\eqref{eq: Psi ZF} for $\Re(s_e)\gg0$, admits an analytic continuation
to the interior of $\cH^{|\vG_1|}$ that extends continuously to
$\cH^{|\vG_1|}$. We denote this extension again by
$W^{\mathrm{ZF}}(\vG,\Phi,\bfs)$. The \textbf{zeta-function renormalization of
the Feynman graph integral} is defined by
\[
        W^{\mathrm{ZF}}(\vG,\Phi)
        :=
        W^{\mathrm{ZF}}(\vG,\Phi,\mathbf{0}).
\]
\end{defn}
    \begin{rem}
As in heat kernel renormalization, for a general quantum field theory one
cannot expect $W^{\mathrm{ZF}}(\vG,\Phi,\mathbf{s})$ to admit an analytic continuation to
$\mathbb{H}^{|\vG_1|}$. In fact, its analytic continuation may have poles at
$\mathbf{s}=\mathbf{0}$; see \cite{Dang2017RenormalizationOF} for details. In
such cases, one usually subtracts the singular part at
$\mathbf{s}=\mathbf{0}$. This subtraction requires an additional choice, as in
heat kernel renormalization, and hence makes the construction less canonical
for general field theories.
\end{rem}

	The second main result of the present paper is the following Fubini-type theorem.
	
	\begin{thm}\label{thm:CPV HK equivalence}
		Let $\vec{\Gamma}$ be a decorated directed graph and let
    $\Phi\in\Omega^\bullet(M^{|\vG_0|})$. The zeta function renormalization $ W^{\mathrm{ZF}}(\vG,\Phi)$ exists, and we have 
		\[
		W^{\mathrm{CPV}}(\vG,\Phi)
		=W^{\mathrm{ZF}}(\vG,\Phi)
		=W^{\mathrm{HK}}(\vG,\Phi).
		\]
	\end{thm}
	
	The proof is given in \Cref{sec:zeta fubini theorem}; more precisely, see
\Cref{thm:zeta fubini theorem}. The strategy is to introduce
zeta-regularized Feynman graph integrals
$W^{\mathrm{HK}}(\vG,\Phi,\mathbf{s})$ and
$W^{\mathrm{CPV}}(\vG,\Phi,\mathbf{s})$ for heat kernel renormalization and
Cauchy principal value renormalization, respectively, which are defined on
$\mathbb{H}^{|\vG_1|}$. When $\Re(s_e)\gg 0$ for all $e\in\vG_1$, the equality
follows from the ordinary Fubini theorem. Since three zeta-regularized
expressions above depend holomorphically on the parameter $\mathbf{s}$, the equality extends to
the full parameter space $\mathbb{H}^{|\vG_1|}$.
Evaluating at $\mathbf{s}=\mathbf{0}$ gives \Cref{thm:CPV HK equivalence}.

	Let's use the notation 
	\[
	W(\vG,\Phi)=W^{\mathrm{HK}}(\vG,\Phi)=W^{\mathrm{ZF}}(\vG,\Phi)=W^{\mathrm{CPV}}(\vG,\Phi).
	\]

	By \Cref{thm:graph integral current}, $W(\vG,\cdot)$ is a well-defined
	current on $M^{|\vG_0|}$. It therefore makes sense to study its holomorphicity in the sense of currents, under the assumption that the decoration of $\vG$ is holomorphic.
	More precisely, we will study whether
	\[
	\bar\p W(\vG,\Phi):=\pm W(\vG,\bar\p\Phi)=0,
	\quad
	\forall\,\Phi\in\Omega^\bullet(M^{|\vG_0|}).
	\]
	
	We can now state the final main result of this paper.
	\begin{thm}[\Cref{prop:large-time-boundary} and \Cref{cor:kahler anomaly quotient graphs}]\label{thm:intro anomaly formula}
		Let $\vec{\Gamma}$ be a holomorphic decorated directed graph. Then $\bar\partial W(\vG,-)$ is a finite sum of boundary integrals over the partially
		compactified Schwinger space in \Cref{thm:schwinger smooth extension0}. 
	\end{thm}
	
	\section{Heat Kernel Formulation}\label{sec: heat formulation}
In this section, we prove that Feynman graph integrals in holomorphic field
theories admit heat kernel renormalization in the sense of
\Cref{defn: heat formulation}. The main idea is to pass from the real
Schwinger space $\resch$ to its complexification, denoted by
$\cpsch$. Correspondingly, we complexify the form
$\widetilde{W}(\vG,\Phi)$ to a form on
\[
M^{|\vG_0|}\times\cpsch.
\]
The key observation is that the complexified form admits an canonical extension to
a suitable blow-up of $M^{|\vG_0|}\times\cpsch$. The extended form enjoys
good properties. As a consequence, its fiber integral (equivalently, its pushforward)
to $\cpsch$ is a smooth differential form. This  will enable us to establish the existence of
the heat kernel renormalization. Finally, we also introduce a zeta-regularized version of the heat kernel
formulation.

	\subsection{A theorem about smoothness of pushforward map}\label{subsec: smooth pushforward}

    In this subsection, we introduce the notion of flatness of a smooth form along a
closed subset. We then establish a criterion showing that, if a form on $M$ is
flat along the critical locus of a submersion $M\to N$, then its pushforward,
initially defined as a current on $N$, is represented by a smooth form; see
\Cref{pushforward is smooth}. This criterion will be used to prove the heat kernel formulation.
	
	First, we introduce a concept as follows.
    \begin{defn}\label{defn:flat along a closed set}
		Let $X$ be a smooth manifold and let $A\subset X$ be a closed subset. A
		smooth function $f\in C^\infty(X)$ is said to be \textbf{flat along $A$} if,
		for every $x\in A$, $f$ as well as any derivative of $f$ at $x$ vanishes.
		
		The same definition applies to sections of a smooth vector bundle
		$E\to X$: a section $s\in\Gamma(X,E)$ is said to be \textbf{flat along
			$A$} if $s$ as well as any derivative of $s$ vanishes at every point of $A$. 
		
		We denote by $\Flat(X,A,E)$ the space of sections of
		$E\to X$ that are flat along $A$. Similarly, $\Flat_{c}(X,A,E)$ is the space of compactly supported sections in $\Flat(X,A,E)$. When the bundle $E$ is clear from the
		context, we abbreviate $\Flat(X,A,E)$ to $\Flat(X,A)$.
	\end{defn}

    The following lemma gives an important property of flat smooth functions.
    \begin{lem}\label{flat property}
        Let $A\subset \mathbb R^n$ be a closed subset. Suppose that
$f\in C^\infty(\mathbb R^n)$ is flat along $A$.
Then for every integer $k\geq 0$, there exists a constant $C_k>0$,
depending only on $n$ and $k$, such that
\[
|f(\mathbf{x})|\leq C_k\|f\|_{C^k(\mathbb R^n)}
\operatorname{dist}(\mathbf{x},A)^k,\quad \mathbf{x}\in\mathbb R^n ,
\]
where $\operatorname{dist}(\mathbf{x},A)$ is the Euclidean distance.
    \end{lem}
    \begin{proof}
        Let $\mathbf{x}\in \mathbb R^n$. Since $A$ is closed and nonempty, we may choose
$\mathbf{a}\in A$ such that
\[
|\mathbf{x}-\mathbf{a}|=\operatorname{dist}(\mathbf{x},A).
\]
If $\mathbf{x}\in A$, the estimate is immediate since $f$ is flat along $A$.
Otherwise, applying Taylor's formula to $f$ at $\mathbf{a}$, and using the fact
that all derivatives of $f$ of order $<k$ vanish on $A$, we obtain
\[
|f(\mathbf{x})|\leq C_k \|f\|_{C^k(\mathbb R^n)}
\operatorname{dist}(\mathbf{x},A)^k.
\]

    \end{proof}
	
	\newcommand{\Crit}{\operatorname{Crit}}
	Let $M$ and $N$ be smooth manifolds. For a  smooth map $f:M\to N$ , we define
	\[
	\Crit(f)
	:=
	\big\{x\in M:\, df_x:T_xM\to T_{f(x)}N
	\text{ is not surjective}\big\}.
	\]    
	\begin{thm}\label{pushforward is smooth}
		Let $M,N$ be two real analytic manifolds, and let $f:M \to N$ be an analytic surjective map. 
		Take a differential form $\a \in \Flat_{c}\big(M,\Crit(f)\big)$.  Then the pushforward current defined by
		\[
		(f_*\alpha)(\phi)
		=
		\int_M \alpha\wedge f^*\phi,
		\quad \phi\in \Omega_c^\bullet(N),
		\]
		is in fact smooth; that is, it is represented by an element of $\Omega^\bullet(N)$.
		
		\def\Flat{{\mathrm{Flat}}}
		Moreover, the map $\Flat_{c}\big(M,\Crit(f)\big)\to \Omega^\bullet(N)$, $\a\mapsto f_*\a$ is a continuous linear map under the $C^\infty$-topology. 
		
	\end{thm}
	
	\begin{proof}
		\def\bfx{\mathbf{x}}
		\def\bfy{\mathbf{y}}
		\def\bff{\mathbf{f}}
		\def\bfxi{\boldsymbol{\xi}}
		It suffices to verify the theorem locally.  
		Let $\bfx=(x_1,\dots,x_m)$ and $\bfy=(y_1,\dots,y_n)$ be coordinates on $\R^m$ and $\R^n$ respectively.  
		Let $\bff:\R^m\to\R^n$ be analytic:
		\[
		\bfx \mapsto \bff(\bfx)=\big(f_1(\bfx),\dots,f_n(\bfx)\big),
		\]
		and let $\a\in \Omega^\bullet_c(\R^m).$  
		Then we can regard $\bff_*\a$ as a distribution defined by 
		\[
		\bff_*\a(\phi)=\int_{\R^m}\a\wedge \bff^*\phi, 
		\quad \forall \phi\in\Omega^\bullet_c(\R^n).
		\]
		
		Moreover, we can rewrite this as
		\[
		\begin{aligned}
			\bff_*\a(\phi) 
			&= \int_{\R^n}\int_{\R^m}\a(\bfx)\wedge \phi(
			\bfy
			)\,\delta(\bfy-\bff(\bfx))\, d^n\bfy \\
			&= \frac{1}{(2\pi)^n}\int_{\R^n}\int_{\R^n}\int_{\R^m}\a(\bfx)\wedge \phi(\bfy)\, 
			e^{i\bfxi\cdot\big(\bfy-\bff(\bfx)\big)}\, d^n\bfy\, d^n\bfxi,
		\end{aligned}
		\]
		where $\bfxi=(\xi_1,\dots,\xi_n)$, $d^n\bfy:=dy_1\wedge\cdots\wedge dy_n$, and 
		$d^n\bfxi:=d\xi_1\wedge\cdots\wedge d\xi_n$.
		
		By the Fourier-transform characterization of smoothness (see \cite[Section 8.1]{hormander1983analysis}), to show that $\bff_*\a$ is smooth, it suffices to prove that 
		\[
		\big|\bff_*(\a)\big(e^{i\bfk\cdot\bfy}\phi\big)\big|
		\;\leq\; \frac{C_N}{(1+|\bfk|^2)^N}
		\]
		for any $N>0$, where $\bfk=(k_1,\dots,k_n)\in\R^n.$
		
		Note that
		\[
		\begin{aligned}
			&\ \ \ \ \bff_*(\alpha)\!\left(\phi e^{i \bfk \cdot \bfy}\right) 
			= \frac{1}{(2 \pi)^n} \int_{\R^n}\int_{\R^n}\int_{\R^m} 
			\a(\bfx)\wedge\phi(\bfy)\, e^{i\bfxi\cdot(\bfy-\bff(\bfx))} e^{i \bfk \cdot \bfy} 
			\, d^n \bfy\, d^n \bfxi \\
			&= \frac{1}{(2 \pi)^n} \int_{\R^n}\int_{\R^n}\int_{\R^m} 
			\a(\bfx)\wedge\phi(\bfy)\, e^{i \bfxi \cdot(\bfy-\bff(\bfx))} 
			\frac{\Big(1-\sum_{j=1}^n \frac{\p^2}{\p y_j^2}\Big)^N }{(1+|\bfk|^2)^N}
			e^{i \bfk \cdot\bfy} 
			\, d^n \bfy\, d^n \bfxi.
		\end{aligned}
		\]
		
		Let 
		\[
		L \;=\; 1-\sum_{j=1}^n \frac{\partial^2}{\partial y_j^2}.
		\]  
		We can integrate by parts in the $\bfy$-variables:
		\[
		\bff_*(\alpha)\!\left(\phi e^{i \bfk \cdot \bfy}\right) 
		= \frac{1}{(2 \pi)^n (1+|\bfk|^2)^N} 
		\int_{\R^n}\int_{\R^n}\int_{\R^m} 
		\a(
		\bfx
		)\wedge L^N\!\Big(\phi(\bfy) e^{i \bfxi \cdot(\bfy-\bff(\bfx))}\Big) 
		e^{i \bfk \cdot\bfy}\, d^n \bfy\, d^n \bfxi.
		\]
		
		Now $L^N\!\Big(\phi(\bfy) e^{i \bfxi \cdot(\bfy-\bff(\bfx))}\Big)$ can be written as a finite sum of terms of the form
		\[
		\psi(\bfy)\, p(\bfxi)\, e^{i \bfxi \cdot(\bfy-\bff(\bfx))},
		\]
		where $\psi$ is a compactly supported smooth form in $\bfy$, and $p$ is a polynomial in $\bfxi$ of degree $\leq 2N$.
		
		As a result, it suffices to prove that
		\be\label{integral we need to estimate}
		\left|\int_{\R^n}\int_{\R^n}\int_{\R^m} 
		\a(\bfx)\wedge \psi(\bfy)\, p(\bfxi)\, 
		e^{i \bfxi \cdot\big(\bfy-\bff(\bfx)\big)} e^{i \bfk \cdot\bfy}
		\, d^n \bfy\, d^n \bfxi\right|<\infty,
		\ee
		where $\psi$ and $p$ are as described above.

		Let $A(\bfx)$ be the Jacobian matrix of $\bff$: 
		\[
		A(\bfx)=
		\begin{pmatrix}
			\frac{\p f_1}{\p x_1}(\bfx) & \cdots & \frac{\p f_1}{\p x_m}(\bfx) \\
			\vdots & \ddots & \vdots \\
			\frac{\p f_n}{\p x_1}(\bfx) & \cdots & \frac{\p f_n}{\p x_m}(\bfx)
		\end{pmatrix}.
		\]
		
		Let 
		\[
		\p_{\bfx}:=\Bigg(\frac{\p}{\p x_1},\dots,\frac{\p}{\p x_m}\Bigg)^T.
		\]
		Define
		\be\label{L  prime on e}
		L' \;=\; 1 + i\,\bfxi\cdot \big(A(\bfx)A^T(\bfx)\big)^{-1}A(\bfx)\p_{\bfx}.
		\ee
		Then outside $\mathrm{Crit}(\bff)$, the operator $L'$ is well-defined. Moreover, outside $\mathrm{Crit}(\bff)$, one can verify that
		\[
		L'e^{i \bfxi \cdot\big(\bfy-\bff(\bfx)\big)}
		= (1+|\bfxi|^2)\,e^{i \bfxi \cdot\big(\bfy-\bff(\bfx)\big)}.
		\]
		
		\def\ad{\mathrm{ad}}
		\def\ad{\mathrm{ad}}
		Given any square matrix $M$, let $\ad(M)$ denote the adjugate matrix of $M$.  
		If $M(\bfx)$ is a smooth family of invertible matrices, then
		\[
		M(\bfx)^{-1} \;=\; \det(M(\bfx))^{-1}\,\ad(M(\bfx)),
		\]
		and
		\[
		\p_{x_i}\big(M(\bfx)^{-1}\big) 
		= -M(\bfx)^{-1}\,\big(\p_{x_i}M(\bfx)\big)\,M(\bfx)^{-1}.
		\]
		Thus one can write $L'$ in the form
		\be\label{L prime}
		1+\sum_{k=1}^m h_{1,k}(\bfxi,\bfx)\det\!\big(A(\bfx)A^T(\bfx)\big)^{-1}\frac{\p}{\p x_k},
		\ee
		where $h_{1,k}(\bfxi,\bfx)$, $k=1,\dots,m$, are smooth in $\bfx$ and linear in $\bfxi$.  
		Similarly, the formal adjoint $L^{\prime,*}$ of $L'$ can be expressed as
		\be\label{L prime star}
        1
		+ g_1(\bfxi,\bfx)\det\!\big(A(\bfx)A^T(\bfx)\big)^{-2}
		+ \sum_{k=1}^m g_{2,k}(\bfxi,\bfx)\det\!\big(A(\bfx)A^T(\bfx)\big)^{-1}\frac{\p}{\p x_k},
		\ee
		where $ g_1(\bfxi,\bfx)$ and $g_{2,k}(\bfxi,\bfx)$ are smooth in $\bfx$ and linear in $\bfxi$.
		
		Here, by the formal adjoint of $L'$, we mean that for any two compactly supported forms $\b_1,\b_2$ on $\R^m$ supported away from $\mathrm{Crit}(\bff)$, one has
		\[
		\int_{\R^m}\b_1\wedge L'\b_2 
		= \int_{\R^m} L^{\prime,*}\b_1 \wedge \b_2.
		\]
		
		\def\dist{\mathrm{dist}}
		Now, by \L ojasiewicz's inequality, on the support of $\a$ there exist constants $C>1$ and $\theta\in(0,1)$ such that
		\be\label{Lojasiewicz}
		\left|\det\!\big(A(\bfx)A^T(\bfx)\big)^{-1}\right|
		\;\leq\; \frac{C}{\dist(\bfx,\mathrm{Crit}(\bff))^\theta}.
		\ee
		
		Consequently, by flatness of $\a$ along $\Crit(\bff)$, \Cref{flat property}, \eqref{L prime star}, and \eqref{Lojasiewicz}, for any $N'>0$ there exists $C_{N'}>0$ such that
		\be\label{L prime star alpha}
		\big|(L^{\prime,*})^{N'} \a(\bfx)\big|
		\;\leq\; C_{N'}(1+|\bfxi|^2)^{\frac{N'}{2}}\chi(\bfx),
		\ee
		where
		\[
		\chi(\bfx)=
		\begin{cases}
			1, & \text{if } \bfx\in \mathrm{supp}(\a),\\
			0, & \text{otherwise.}
		\end{cases}
		\]
		
		If $N'$ is chosen sufficiently large, then
		\[\ba
		&\;\;\left|\int_{\R^n}\int_{\R^n}\int_{\R^m} 
		\a(\bfx)\wedge \psi(\bfy)\, p(\bfxi)\, 
		e^{i \bfxi \cdot(\bfy-\bff(\bfx))} e^{i \bfk \cdot\bfy}
		\, d^n \bfy\, d^n \bfxi\right| \\[6pt]
		&= \left|\int_{\R^n}\int_{\R^n}\int_{\R^m} 
		\a(\bfx)\wedge \psi(\bfy)\, p(\bfxi)\, 
		\frac{(L')^{N'}}{(1+|\bfxi|^2)^{N'}}\Big(e^{i \bfxi \cdot(\bfy-\bff(\bfx))}\Big) e^{i \bfk \cdot\bfy}
		\, d^n \bfy\, d^n \bfxi\right| \\[6pt]
		&= \left|\int_{\R^n}\int_{\R^n}\int_{\R^m} 
		\frac{(L^{\prime,*})^{N'}}{(1+|\bfxi|^2)^{N'}}\big(\a(\bfx)\big)\wedge \psi(\bfy)\, p(\bfxi)\, 
		e^{i \bfxi \cdot(\bfy-\bff(\bfx))} e^{i \bfk \cdot\bfy}
		\, d^n \bfy\, d^n \bfxi\right| \\[6pt]
		&\leq C \int_{\R^n}\int_{\mathrm{supp}(\phi)}\int_{\mathrm{supp}(\a)} 
		\frac{1}{(1+|\bfxi|^2)^{\frac{N'}{2}-N}}
		\, d^m\bfx\, d^n \bfy\, d^n \bfxi \;<\;\infty,
		\ea\] 
		where the second equality follows from flatness of $\a$ along $\Crit(\bff)$, \eqref{L prime}, \eqref{L prime star}, and \eqref{Lojasiewicz} (justifying integration by parts), and the first inequality in the last line follows from \eqref{L prime star alpha} together with the fact that $p(\bfxi)$ is a polynomial of degree $2N$. 
		
		Thus, \eqref{integral we need to estimate} is verified. The continuity of $f_{*}$ follows by expanding the notation introduced in the preceding construction and keeping track of the corresponding estimates. 
	\end{proof}
	
	The following corollary follows from \Cref{pushforward is smooth} and partition of unity.
    \begin{cor}\label{pushforward is smooth2}
      Let $M,N$ be two real analytic manifolds, and let $f:M \to N$ be a proper analytic surjective map. 
		Take a differential form $\a \in \Flat\big(M,\Crit(f)\big)$.  Then the pushforward current defined by
		\[
		(f_*\alpha)(\phi)
		=
		\int_M \alpha\wedge f^*\phi,
		\quad \phi\in \Omega_c^\bullet(N),
		\]
		is smooth.\def\Flat{{\mathrm{Flat}}}
		Moreover, the map $\Flat\big(M,\Crit(f)\big)\to \Omega^\bullet(N)$, $\a\mapsto f_*\a$ is a continuous linear map under the $C^\infty$-topology.   
    \end{cor}
	The following simple fact is often used in verifying the conditions in
	\Cref{pushforward is smooth} and \Cref{pushforward is smooth2}.

	\begin{prop}\label{prop:exp-minus-f-over-g-smooth}
		Let $X$ be a compact smooth manifold. Let $f,g\in C^\infty(X)$, with
		$g\ge0$. Assume that $f$ is strictly positive on $\{g=0\}$. Then for any
		$k>0$, the function
		\[
		g^{-k}e^{-f/g}
		\]
		defined on $\{g>0\}$ extends to a smooth function on $X$, and is flat along
		$\{g=0\}$.
	\end{prop}
	
	\begin{proof}
		Since $f>0$ on the compact set $\{g=0\}$, there is a neighborhood $U$ of
		$\{g=0\}$ and a constant $c>0$ such that $f\ge c$ on $U$. Hence, on
		$U\cap\{g>0\}$, $e^{-f/g}\le e^{-c/g}$, and $e^{-c/g}$ tends to zero
		faster than any power of $g$. Every derivative of $g^{-k}e^{-f/g}$ is a finite
		sum of terms of the form
		\[
		e^{-f/g}\frac{A}{g^{N}},
		\]
		with $A$ smooth; hence all derivatives tend to zero as $g\to0$. Therefore
		the zero extension is smooth, and all its derivatives vanish on $\{g=0\}$.
	\end{proof}

	\def\sq{\mathrm{sq}}
	\def\Bl{\mathrm{Bl}}
	\subsection{Complexified Schwinger parameter}\label{subsec: complexfied schwinger}

    In this subsection, we complexify the time parameter and denote the resulting
complex time parameter by $\tau$. We prove that there exists an integer
$N>0$ such that $\tau^N$ times the complexified propagator extends to a
smooth form on a suitable blow-up of $M\times M\times\mathbb C$ and is flat
along the desired closed subset; see \Cref{propogator is meromorphic}.

	Let $\sq\colon \mathbb C\to [0,\infty)$ be the map defined by
\[
\sq(\tau)=t,\quad t=|\tau|^2.
\]
We also regard $\tau$ as the holomorphic linear coordinate on $\mathbb C$,
and call it the complexified Schwinger parameter. For a closed complex submanifold $Z\subset M$, we denote by $\mathrm{Bl}_{Z}M$ the blow-up of $M$ along $Z$.

	\begin{prop}\label{propogator is meromorphic}
		Let $R_t \in \Omega^\bullet\!\big((0,\infty);\Gamma(\tE \boxtimes \tE)\big)$, such that $e^{\frac{\rho^2}{2t}}R_t$ admits a regular expression, where $\rho$ is the distance function on $M$.  
		Then there exists $N\in\mathbb{N}$, such that
		\[
		\tau^N\sq^* R_t \in \Omega^\bullet\!\big(\C\setminus\{0\};\Gamma(\tE \boxtimes \tE)\big)
		\]
		extends to a smooth differential form on 
		\[
		\mathrm{Bl}_{\triangle \times \{0\}}(M \times M \times \C),
		\]
		where \[
		\triangle := \{(q,q) \in M \times M : q \in M\} \subset M \times M
		\]
		is the diagonal. Moreover, $\tau^N\sq^*R_t$ is flat along the closure of 
        \[(M\times M)\setminus\triangle\times\{0\}\subset \mathrm{Bl}_{\triangle \times \{0\}}(M \times M \times \C).
        \]
		 In particular, by \Cref{thm: Propogator has regular expression}, these statements hold for the propagator in Schwinger parameter space $P_{t}$ and its (higher) holomorphic derivatives.
	\end{prop}
	\begin{proof}
		The proposition is local, so we may assume that $M=\mathbb{C}^n$, and $\sq^* R_t$ has the following form:
		\[
		\sq^* R_t
		=
		e^{-\frac{\rho^2}{2|\tau|^2}}a\left(\frac{\bar{\bfz}-\bar{\bfw}}{|\tau|^2},d\Big(\frac{\bar{\bfz}-\bar{\bfw}}{|\tau|^2}\Big)\right)b,\quad (\mathbf{z},\mathbf{w})\in\mathbb{C}^n \times \mathbb{C}^n, 
		\]
		where $a$ is a polynomial and $b \in \Omega^\bullet\Big(\mathbb{C}; \Gamma\big((\widetilde{E} \boxtimes \widetilde{E})\big)\Big)$.
		
		We only need to prove that there exists a neighborhood 
		$W \subset \C^n \times \C^n \times \C$ of $(0,0,0)$ such that our claim holds in 
		$p^{-1}(W) \subset \Bl_{\triangle \times \{0\}}(\C^n \times \C^n \times \C)$, where \[p:\Bl_{\triangle \times \{0\}}(\C^n \times \C^n \times \C)\longmapsto\C^n \times \C^n \times \C
        \] is the canonical blowup map.

		\def\P{\mathbb{P}}
		The blowup admits the following description:
		\[
		\begin{aligned}
			\Bl_{\triangle \times \{0\}}\big(\C^n \times \C^n \times \C\big)
			\cong
			\Big\{
			&(\bfz,\bfw,\tau,[\lambda_1 : \cdots : \lambda_n : \lambda])
			\in (\C^n \times \C^n\times \C) \times \P^{n}
			\;:\; \\
			&\bfz - \bfw = k \cdot \boldsymbol{\lambda},
			\quad
			\tau = k \lambda
			\text{ for some } k \in \C
			\Big\},
		\end{aligned}
		\]
		where
		\[
		\boldsymbol{\lambda} := (\lambda_1,\dots,\lambda_n),
		\quad
		\bfz = (z_1,\cdots,z_n),
		\quad
		\bfw = (w_1,\cdots,w_n) \in \C^n.
		\]

		Let $\{U_i\}_{i=1}^{n+1}$ be an open cover of 
		$\Bl_{\triangle \times \{0\}}(\C^n \times \C^n \times \C)$, where each $U_i$ corresponds to the subset $\{\lambda_i \neq 0\},\, 1 \leq i \leq n$, and $U_{n+1}$ corresponds to $\{\lambda \neq 0\}$.

		On the chart $U_1$, we may use the canonical affine coordinates
		\[
		(z_1 - w_1, \lambda_2,\dots,\lambda_n,\lambda, \bfw), 
		\]
        such that $\lambda_{1}=1$ and 
        \[
        \begin{cases}
z_{i}-w_{i}=\lambda_{i}(z_{1}-w_{1}),\quad i\neq 1,\\
\tau=\lambda(z_{1}-w_{1}).
\end{cases}
        \]
		
		Then there exists a neighborhood 
		$W \subset \C^n \times \C^n \times \C$ of $(0,0,0)$, such that on $U_1\cap p^{-1}(W)$, the distance function has the following form (see \cite[Corollary A.2]{WYFeynman}):
		\begin{equation}\label{chart U1}
			\begin{aligned}
		\rho^2(\bfz,\bar\bfz,\bfw,\bar\bfw)
				&= |z_1-w_1|^2 ( 
				1+ H
				),
			\end{aligned}
		\end{equation}
		for some smooth functions $\{H\}$ with $H|_{(\bfz,\bfw)=(0,0)} = 0$.
		
		By shrinking $W$ if necessary, we may assume that
		\be\label{chart U1 eq2}
		F(\bfz,\bar\bfz,\bfw,\bar\bfw):=\frac{\rho^2(\bfz,\bar\bfz,\bfw,\bar\bfw)}{|z_1-w_1|^2} \ge \tfrac{1}{2}.
		\ee
		
		Then we notice that on $U_1\cap p^{-1}(W)$,
         \[
        \begin{cases}
        e^{-\frac{\rho^2}{2|\tau|^2}}=e^{-\frac{F}{2|\lambda|^2}},\\
        \tau\cdot\frac{\bar{z}_1-\bar{w}_1}{|\tau|^2}=\frac{\lambda}{|\l|^2},\\ \tau\cdot\frac{\bar{z}_i-\bar{w}_i}{|\tau|^2}=\frac{\lambda\bar\l_i}{|\l|^2},\quad i\neq 1,\\
        \tau^2\cdot d\Big(\frac{\bar{z}_1-\bar{w}_1}{|\tau|^2}\Big)=\frac{-\lambda^{2}|\lambda|^{2}d(z_{1}-w_{1})-(z_{1}-w_{1})\lambda^{2}d(|\lambda|^{2})}{|\l|^4},\\ \tau^{2}\cdot d\Big(\frac{\bar{z}_i-\bar{w}_i}{|\tau|^2}\Big)=\frac{-\lambda^{2}|\lambda|^{2}\bar\l_i d(z_{1}-w_{1})+(z_{1}-w_{1})\lambda^{2}|\lambda|^{2}d\bar\l_i-(z_{1}-w_{1})\lambda^{2}\bar\l_i d(|\lambda|^{2})}{|\l|^4},\quad i\neq 1.
        \end{cases}
        \]
		Therefore, let $N\in \mathbb{N}$ be a number that is larger than $2$ times the polynomial degree of $a$, by \Cref{prop:exp-minus-f-over-g-smooth}, 
        \[
        \tau^{N}e^{-\frac{\rho^2}{2|\tau|^2}}a\bigg(\frac{\bar{\bfz}-\bar{\bfw}}{|\tau|^2},d\Big(\frac{\bar{\bfz}-\bar{\bfw}}{|\tau|^2}\Big)\bigg)
        \]
		extends to a smooth differential form on $U_1\cap p^{-1}(W)$, which is flat along
        \[
        \{(z_1 - w_1, \lambda_2,\dots,\lambda_n,\lambda, \bfw)\in U_{1}\cap p^{-1}(W):\lambda=0\},
        \]
        which is the closure of \[
        \big((M\times M)\setminus\triangle\times\{0\}\big)\cap U_{1}\cap p^{-1}(W)\subset U_{1}\cap p^{-1}(W).
        \]
		
		This proves our claim on $U_1\cap p^{-1}(W)$. Similarly, we can prove our claim on $U_i\cap p^{-1}(W)$, where $i\neq n+1$.

		On the chart $U_{n+1}$, we may use the canonical affine coordinates
		\[
		(\lambda_1,\dots,\lambda_n,\tau,w_1,\dots,w_n).
		\]
        such that $\lambda=1$ and 
        \[
z_{i}-w_{i}=\lambda_{i}\tau,\quad 1\leq i\leq n.
        \]
		
		Then on $U_{n+1}\cap p^{-1}(W)$,
		\begin{equation}\label{chart Un}
			\rho^2
			= |\tau|^2 G,
		\end{equation}
		for some smooth functions $G$.
		
		Therefore, on $U_{n+1}\cap p^{-1}(W)$, we have
		\[
		\begin{cases}
        e^{-\frac{\rho^2}{2|\tau|^2}}=e^{-\frac{1}{2}G},\\
        \tau\cdot\frac{\bar{\bfz}-\bar{\bfw}}{|\tau|^2}
		=
		\bar{\boldsymbol{\lambda}},\\
        \tau^{2}\cdot d\Big(\frac{\bar{\bfz}-\bar{\bfw}}{|\tau|^2}\Big)
		=
		-\bar{\boldsymbol{\lambda}}d\tau+\tau d\bar{\boldsymbol{\lambda}}.
        \end{cases}
		\]
		Therefore,
		\[
		\tau^{N}e^{-\frac{\rho^2(\bfz,\bfw)}{2|\tau|^2}}
		a\bigg(
		\frac{\bar{\bfz}-\bar{\bfw}}{|\tau|^2},
		d\Big(\frac{\bar{\bfz}-\bar{\bfw}}{|\tau|^2}\Big)
		\bigg)
		\]
		extends to a smooth differential form on $U_{n+1}\cap p^{-1}(W)$. Finally, we notice that the closure of \[
        \big((M\times M)\setminus\triangle\times\{0\}\big)\cap W\subset  p^{-1}(W)
        \]
		has no intersection with $U_{n+1}$, our claim follows.

	\end{proof}

	\subsection{Review on wonderful compactification of a building set}\label{subsec: wonderful compactification}
    This subsection prepares for the proof of an analogous result for the
complexification of $\widetilde W(\vG,\Phi)$, corresponding to the result for
$P_t$ established in \Cref{propogator is meromorphic}. We review the
wonderful partial compactification associated with a building set on a connected
complex manifold. 
	Related concepts, definitions, and notation can be found in \cite{de1995wonderful,macpherson1998making,hu2003compactification,li2009wonderful}.
	
	\begin{defn}[Simple arrangements]
		A \textbf{simple arrangement} of closed connected complex submanifolds of a connected complex manifold $X$ is a finite set $\mathcal{S}=\left\{S_i\right\}$ of closed connected complex submanifolds $S_i$ satisfying the following conditions:
		\begin{enumerate}[(1)]
			\item $S_i$ and $S_j$ intersect cleanly (i.e., their intersection is a smooth submanifold and the tangent bundles satisfy $T\left(S_i \cap S_j\right)=T\left(S_i\right)\big|_{\left(S_i \cap S_j\right)} \cap T\left(S_j\right)\big|_{\left(S_i \cap S_j\right)})$;
			\item  $S_i \cap S_j$ is either equal to some $S_k$ or is empty.
		\end{enumerate}
		
	\end{defn}

	\begin{defn}[Building sets]
		Let $\mathcal{S}$ be a simple arrangement of closed connected complex submanifolds of $X$. A subset $\mathcal{G} \subseteq \mathcal{S}$ is called a \textbf{building set} of $\mathcal{S}$ if, for all $S \in \mathcal{S}$, the minimal elements in $\{G \in \mathcal{G}$ : $G \supseteq S\}$ intersect transversally and their intersection is $S$. In this case, these minimal elements are called \textbf{the $\mathcal{G}$-factors} of $S$.
		
		A finite set $\mathcal{G}$ of closed connected complex submanifolds of $X$ is called a \textbf{building set} if the set of all possible finite intersections of collections of submanifolds from $\mathcal{G}$ forms a simple arrangement $\mathcal{S}$ and if $\mathcal{G}$ is a building set of $\mathcal{S}$. In this situation, $\mathcal{S}$ is called \textbf{the arrangement induced by $\mathcal{G}$}.
	\end{defn}
	\begin{rem}
It is easy to see that any simple arrangement $\mathcal{S}$ of $X$ is
itself a building set.
\end{rem}
	
	\begin{defn}[Wonderful compactification]\label{Wonderful compactification}
		Let $\mathcal{G}$ be a nonempty building set and $X^{\circ}=X \backslash \bigcup_{G \in \mathcal{G}} G$. The closure of the image of the natural locally closed embedding
		
		$$
		X^{\circ} \hookrightarrow \prod_{G \in \mathcal{G}} \mathrm{Bl}_G X
		$$
		is called the wonderful (partial) compactification of the building set $\mathcal{G}$ and is denoted by $X_{\mathcal{G}}$. 
	\end{defn}
	
	\begin{defn}[$\mathcal{G}$-nest]
		A subset $\mathcal{T} \subseteq \mathcal{G}$ is called $\mathcal{G}$-nested (or a $\mathcal{G}$ nest) if it satisfies:\\
		There is a flag of elements in $\mathcal{S}$ : $S_1 \subseteq S_2 \subseteq \ldots \subseteq S_{\ell}$ such that
		
		$$
		\mathcal{T}=\bigcup_{i=1}^{\ell}\left\{A: A \text { is a } \mathcal{G} \text {-factor of } S_i\right\} .
		$$
		(We say $\mathcal{T}$ is induced by the flag $S_1 \subseteq S_2 \subseteq \cdots \subseteq S_{\ell}$.)
	\end{defn}

\begin{defn}\label{defn: dominant transform}
Let $Z\subset X$ be a closed connected complex submanifold. We denote by
$\mathrm{Bl}_{Z}X$ the blow-up of $X$ along $Z$, and by
\[
p:\mathrm{Bl}_{Z}X\longrightarrow X
\]
the canonical blow-up map. Let $G\subset X$ be a closed connected complex submanifold. The
\textbf{dominant transform} $\widetilde{G}$ of $G$ under the blow-up of
$X$ along $Z$ is defined as follows. 
\begin{enumerate}[(1)]
    \item If $G\subset Z$, then
\[
\widetilde{G}=p^{-1}(G).
\]
\item If $G\not\subset Z$, then $\widetilde{G}$ is the closure of
$p^{-1}(G\setminus Z)$ in $\mathrm{Bl}_{Z}X$.
\end{enumerate}
 For a sequence of blow-ups, we still denote the iterated dominant transform by $\widetilde{G}$.
\end{defn}

    The following theorem was proved in \cite{li2009wonderful} for algebraic
varieties. The same argument applies essentially unchanged to complex manifolds.
	\begin{thm}[Theorem 1.2 and Theorem 1.3 in \cite{li2009wonderful}]\label{properties of wonderful compactification}
		Let $X$ be a connected complex manifold and let $\mathcal{G}=\left\{G_1, \ldots, G_N\right\}$ be a nonempty building set of closed connected complex submanifolds of $X$. Then the wonderful compactification $X_{\mathcal{G}}$ is a complex manifold. Moreover, for each $G \in \mathcal{G}$ there is a smooth divisor $D_G \subset X_{\mathcal{G}}$ such that:
		\begin{enumerate}[(1)]
			\item the union of these divisors is $X_{\mathcal{G}} \backslash X^{\circ}$;
			\item  any set of these divisors meets transversally. An intersection of divisors $D_{T_1} \cap \cdots \cap D_{T_r}$ is nonempty exactly when $\left\{T_1, \ldots, T_r\right\}$ form a $\mathcal{G}$-nest.
		\end{enumerate}
		Furthermore, suppose that $\mathcal{G}=\{G_1,\ldots,G_N\}$ is ordered so that
$\{G_1,\ldots,G_i\}$ is a building set for every $1\leq i\leq N$. Then
\[
   X_{\mathcal{G}}
   =
   \Bl_{\widetilde{G}_{N}}\cdots
   \Bl_{\widetilde{G}_{2}}\Bl_{G_{1}}X,
\]
where each blow-up is along a closed smooth complex submanifold.
	\end{thm}
    \begin{prop}\label{linear order G}
        Let $X$ be a connected complex manifold and let $\mathcal{G}$ be a nonempty building set. If $\mathcal{G}=\left\{G_1, \ldots, G_N\right\}$ is indexed in an order compatible with inclusion relations, i.e., $i\leq j$ if $G_{i}\subseteq G_{j}$, then $\left\{G_1, \ldots, G_i\right\}$ is a building set for any $1\leq i\leq N$.
    \end{prop}
\begin{proof}
    Let \[
    S=G_{j_{1}}\cap\cdots G_{j_{k}},
    \]
    where $1\leq j_{l}\leq i$ for $1\leq l\leq i$.
    For any $G_{j_{l}}$, there exists a minimal element $G'_{j_{l}}\in \mathcal{G}$ containing $S$, and $G'_{j_{l}}\subset G_{j_{l}}$. By our assumption, $G'_{j_{l}}\in \left\{G_1, \ldots, G_i\right\}$. Therefore,
    \[
    S=G'_{j_{1}}\cap\cdots G'_{j_{k}}.
    \]
    Since $\mathcal{G}$ is a building set, $\{G'_{j_{1}},\dots G'_{j_{k}}\}$ includes all the $\mathcal{G}$-factors of $S$. Thus, $\left\{G_1, \ldots, G_i\right\}$ is a building set.
\end{proof}
    Given a building set $\mathcal{G}$, \Cref{properties of wonderful compactification}
gives a canonical proper map
\[
p_{\mathcal{G}}:X_{\mathcal{G}}
=\Bl_{\widetilde{G}_{N}}\cdots\Bl_{\widetilde{G}_{2}}\Bl_{G_{1}}X
\longrightarrow X.
\]
\begin{prop}\label{pull back divisor decomposition}
Let $X$ be a connected complex manifold, and let $\mathcal{G}$ be a
nonempty building set of complex submanifolds of $X$. Let
\[
p_{\mathcal{G}}\colon X_{\mathcal{G}}\longrightarrow X
\]
be the canonical proper map. For $G\in\mathcal{G}$, denote by
$\mathcal{I}_{G}\subseteq\mathcal{O}_{X}$ the ideal sheaf of $G$.
Then the total transform (pull back)
\[
p_{\mathcal{G}}^{-1}\mathcal{I}_{G}
\cdot\mathcal{O}_{X_{\mathcal{G}}}
=
\prod_{\substack{G'\in\mathcal{G}\\G'\subseteq G}}
\mathcal{I}_{\widetilde{G'}},
\]
where $\mathcal{I}_{\widetilde{G'}}$ denotes the ideal sheaf of the
boundary divisor $\widetilde{G'}\subseteq X_{\mathcal{G}}$.
In particular, the effective divisor
\[
\sum_{\substack{G'\in\mathcal{G}\\G'\subseteq G}}
\widetilde{G'}
\]
is a simple normal crossing divisor.
\end{prop}

 \begin{proof}
Choose an ordering
\[
\mathcal{G}=\{G_{1},\ldots,G_{N}\}
\]
compatible with inclusion.
By \Cref{properties of wonderful compactification}
and \Cref{linear order G}, the wonderful compactification is obtained
as the iterated blow-up
\[
X_{\mathcal{G}}
=
\Bl_{\widetilde{G}_{N}}
\cdots
\Bl_{\widetilde{G}_{2}}
\Bl_{G_{1}}X.
\]

By
\cite[Proposition~2.8]{li2009wonderful}, after each blow-up the
dominant transforms of the elements of $\mathcal{G}$ again form a
building set. Since the center of each blow-up is a minimal element
of the current building set,
\cite[Definition--Lemma~2.6(i)]{li2009wonderful} implies that the
current dominant transform of $G$ either contains the blow-up center
or meets it transversally.

Recall that, for the blow-up
\[
\pi\colon \Bl_{Z}Y\longrightarrow Y
\]
of a complex manifold along a smooth center $Z$, and for a smooth
submanifold $V\subseteq Y$, one has
\[
\pi^{-1}\mathcal{I}_{V}\cdot\mathcal{O}_{\Bl_{Z}Y}
=
\begin{cases}
\mathcal{I}_{\widetilde{V}}\mathcal{I}_{E},
    & Z\subsetneq V,\\[4pt]
\mathcal{I}_{E},
    & Z=V,\\[4pt]
\mathcal{I}_{\widetilde{V}},
    & Z \text{ intersects } V \text{ transversally},
\end{cases}
\]
where $E$ is the exceptional divisor and $\widetilde{V}$ is the
dominant transform of $V$.

Applying this formula successively along the above sequence of
blow-ups, an exceptional factor is introduced precisely for those
centers $G'\in\mathcal{G}$ satisfying $G'\subseteq G$. Hence
\[
p_{\mathcal{G}}^{-1}\mathcal{I}_{G}
\cdot\mathcal{O}_{X_{\mathcal{G}}}
=
\prod_{\substack{G'\in\mathcal{G}\\G'\subseteq G}}
\mathcal{I}_{\widetilde{G'}}.
\]

Finally, by \Cref{properties of wonderful compactification}, any collection
of boundary divisors of $X_{\mathcal{G}}$ meets transversely.
Therefore,
\[
\sum_{\substack{G'\in\mathcal{G}\\G'\subseteq G}}
\widetilde{G'}
\]
is a simple normal crossing divisor.
\end{proof}
    \subsection{Partial compactification of Schwinger spaces}\label{subsec: partial compactification}
	In this subsection, we construct a partial compactification of Schwinger
 spaces using wonderful compactifications. 

Given a directed graph $\vG$, we consider complex Schwinger parameter space $\mathbb{C}^{|\vG_{1}|}$.

For each subset $S \subset \vG_1$, consider the coordinate subspace  
	\begin{equation}\label{definition of Z}
	\mathcal{Z}_S := \{(\tau_1,\dots,\tau_{|\vG_1|}) \in \C^{|\vG_1|} : \tau_e = 0 \ \text{for all } e \in S\}.
	\end{equation}
    We notice that 
    \[
    \{\mathcal{Z}_S\subset \mathbb{C}^{|\vG_{1}|}:S\subset\vG_{1}\}
	\]
    is a simple arrangement.
\begin{defn}
    Given a directed graph $\vG$, \textbf{the partially compactified complex Schwinger parameter space} $\cpsch$ is the wonderful compactification of $\mathbb{C}^{|\vG_{1}|}$, where the building set is
    \[
    \{\mathcal{Z}_S\subset \mathbb{C}^{|\vG_{1}|}:S\subset\vG_{1}\}.
    \]
    There is a natural blow-up map $\cpsch \to \C^{|\vG_1|}$, which we denote by 
		\[
		\pi^{\vG} = (\tau_1,\ldots,\tau_{|\vG_1|}).
		\]
\end{defn}
We consider the following natural $(\mathbb{C}^{*})^{|\vG_{1}|}$ action on $\mathbb{C}^{|\vG_{1}|}$:
\[
(\lambda_{e})_{e\in\vG_{1}}\times(\tau_{e})_{e\in\vG_{1}}\in(\mathbb{C}^{*})^{|\vG_{1}|}\times\mathbb{C}^{|\vG_{1}|}\longrightarrow (\lambda_{e}\tau_{e})_{e\in\vG_{1}}\in\mathbb{C}^{|\vG_{1}|}.
\]
This action can be lifted to $\cpsch$.
\begin{lem}\label{cstar action}
    Given a directed graph $\vG$, there is a unique $(\mathbb{C}^{*})^{|\vG_{1}|}$ action on $\cpsch$, such that $\pi^{\vG}$ is a $(\mathbb{C}^{*})^{|\vG_{1}|}$ equivariant map. 
\end{lem}\label{Cstar action}
\begin{proof}
    We notice that $\mathcal{Z}_S$ is invariant under the action of $(\mathbb{C}^{*})^{|\vG_{1}|}$ for any $S\subset\vG_{1}$, so the $(\mathbb{C}^{*})^{|\vG_{1}|}$ action can be lifted to
    \[
    \prod_{S \subset \vG_1} \Bl_{\mathcal{Z}_S}(\C^{|\vG_1|})
    \]
    and the natural embeding
    \[
    (\C^*)^{|\vG_1|} \longrightarrow
    \prod_{S \subset \vG_1} \Bl_{\mathcal{Z}_S}(\C^{|\vG_1|})
    \]
    is equivariant. Thus, its closure $\cpsch$ has a unique $(\mathbb{C}^{*})^{|\vG_{1}|}$ action.
\end{proof}
\begin{rem}
    By \Cref{cstar action}, the partially compactified complex Schwinger parameter space $\cpsch$ is a smooth toric variety. Many constructions in this subsection are inspired by toric geometry.
\end{rem}
Let us consider the map
\begin{equation}\label{square map}
	\begin{array}{rcl}
		\prod_{e\in\vG_1} \sq_e:\C^{|\vG_1|} &\longrightarrow& [0,\infty)^{|\vG_1|}, \\[6pt]
		(\tau_1,\ldots,\tau_{|\vG_1|}) &\longmapsto& (|\tau_1|^2,\ldots,|\tau_{|\vG_1|}|^2).
	\end{array}
\end{equation}
    We notice that 
    \[
    (\prod_{e\in\vG_1} \sq_e)\circ \pi^{\vG}:\cpsch\longrightarrow[0,\infty)^{|\vG_1|}
    \]is invariant under the action of $(S^{1})^{|\vG_{1}|}\subset(\mathbb{C}^{*})^{|\vG_{1}|}$, so it descends to a map
    \[
    \pi_\R^{\vG}:\cpsch/(S^{1})^{|\vG_{1}|}\longrightarrow [0,\infty)^{|\vG_1|}.
    \]
\begin{defn}
    Given a directed graph $\vG$, \textbf{the partially compactified Schwinger parameter space} $\resch$ is defined by the quotient space $\cpsch/(S^{1})^{|\vG_{1}|}$.
    
    The natural proper map $\pi_\R^{\vG}$ is denoted by 
		\[
		\pi_\R^{\vG} = (t_1,\ldots,t_{|\vG_1|}).
		\]
\end{defn}
\begin{rem}
    The partially compactified Schwinger parameter space $\resch$ can also be constructed by the real version of wonderful compactification, see \cite{wang2025feynman,10.1155/S1073792803209077,Ammann2019ACO}. In this paper, the comparison between $\cpsch$ and $\resch$ is necessary, so we do not use this approach.
\end{rem}
Inspired by \eqref{square map}, we use the same notation $\prod_{e\in\vG_1} \sq_e$ to denote the quotient map from $\cpsch$ to $\resch$:
\[
\prod_{e\in\vG_1} \sq_e:\cpsch\longrightarrow\resch. 
\]
\begin{thm}\label{real wonderful compactification}
    Given a directed graph $\vG$, $\resch$ is a smooth manifold with corners. Moreover, it is a partial compactification of Schwinger parameter space in the sense of \Cref{compactification}. 
\end{thm}

\subsubsection{Proof of \Cref{real wonderful compactification}}
  Following \cite{de1995wonderful}, we introduce a collection of
$(\mathbb{C}^{*})^{|\vG_{1}|}$-equivariant coordinate charts on $\cpsch$.
\begin{defn}\label{defn: nested sequence}
Given a directed graph $\vG$, a \textbf{marked nested sequence of edges}
is a sequence
\[
\vG_1=S_0 \supseteq S_1 \supsetneq S_2 \supsetneq \cdots
\supsetneq S_m \supsetneq S_{m+1}=\varnothing,
\]
together with a choice of an element
\[
e_k\in S_k\setminus S_{k+1},\quad \text{for each }k=1,\dots,m.
\]
We call $m$ the length of this sequence. We use $\mathsf{MNS}(\vG)$ to denote the set of marked nested sequences of edges, and an element with length $m$ is denoted by $\mathfrak{s}=(S_{i},e_{i})_{i=1}^{m}\in \mathsf{MNS}(\vG)$.
\end{defn}

Given a marked nested sequence of edges
$\mathfrak{s}=(S_i,e_i)_{i=1}^{m}\in \mathsf{MNS}(\vG)$, we define a map
\[
\kappa_{\mathfrak{s}}\colon \mathbb{C}^{m}\times(\mathbb{C}^{*})^{|\vG_{1}\setminus\{e_{1},\dots,e_{m}\}|}
\longrightarrow \mathbb{C}^{|\vG_{1}|}
\]
as follows:
\begin{equation}\label{coordinates of complex schwinger space}
\kappa_{\mathfrak{s}}^{*}\tau_e=
\begin{cases}
\gamma_{e},& e\in \Gamma_{1}\setminus S_{1}\\
\gamma_e\prod_{k=1}^{i}\gamma_{e_k},
& e\in S_i\setminus S_{i+1}\text{ and } e\neq e_i, \text{ for }1\leq i\leq m,\\
\prod_{k=1}^{i}\gamma_{e_k},
& e=e_i,\text{ for }1\leq i\leq m.
\end{cases}
\end{equation}
Here, $\{\gamma_e\}_{e\in\vG_1}$ are the coordinates on the domain, and
$\{\tau_e\}_{e\in\vG_1}$ are the coordinates on the codomain. To simplify notations, we use \[
C_{\mathfrak{s}}:=\mathbb{C}^{m}\times(\mathbb{C}^{*})^{|\vG_{1}\setminus\{e_{1},\dots,e_{m}\}|}
\]
to denote the domain of $\kappa_{\mathfrak{s}}$.

We consider the following $(\mathbb{C}^{*})^{|\vG_{1}|}$ action on $C_{\mathfrak{s}}$:
for $(\lambda_{e})_{e\in\vG_{1}}\in (\mathbb{C}^{*})^{|\vG_{1}|}$,
\[
(\lambda_{e})_{e\in\vG_{1}}\cdot(\gamma_{e})_{e\in\vG_{1}}= (\gamma'_{e})_{e\in\vG_{1}}, 
\]
where
\begin{equation}\label{local cstar action}
\gamma_{e}'=
\begin{cases}
\lambda_{e}\gamma_{e},& e\in \Gamma_{1}\setminus S_{1}\\
\frac{\lambda_{e}}{\lambda_{e_{i}}}\gamma_{e},
& e\in S_i\setminus S_{i+1}\text{ and } e\neq e_i, \text{ for }1\leq i\leq m,\\
\frac{\lambda_{e_{i}}}{\lambda_{e_{i-1}}}\gamma_{e},
& e=e_i,\text{ for }1\leq i\leq m, \text{ where }\lambda_{e_0}=1.
\end{cases}
\end{equation}
This action makes $\kappa_{\mathfrak{s}}$ an equivarant map.
\begin{prop}\label{prop: kappa s}
    Given a marked nested sequence of edges
$\mathfrak{s}=(S_i,e_i)_{i=1}^{m}\in \mathsf{MNS}(\vG)$, the map $\kappa_{\mathfrak{s}}\colon C_{\mathfrak{s}}
\longrightarrow \mathbb{C}^{|\vG_{1}|}$ lifts to an equivarant open embedding into $\cpsch$. We still use $\kappa_{\mathfrak{s}}$ to denote it. Moreover, 
\[
\cpsch =\bigcup_{\mathfrak{s}\in\mathsf{MNS}(\vG)}\kappa_{\mathfrak{s}}(C_{\mathfrak{s}}).
\]
\end{prop}
\begin{proof}
This is proved in \cite[Section 3.1]{de1995wonderful}.
\end{proof}
\begin{prop}\label{chart on real}
    Given a marked nested sequence of edges
$\mathfrak{s}=(S_i,e_i)_{i=1}^{m}\in \mathsf{MNS}(\vG)$, the quotient space $C_{\mathfrak{s}}/(S^{1})^{|\vG_{1}|}$ is homeomorphic to \[
C_{\mathfrak{s}}^{\mathbb{R}}:=[0,\infty)^{m}\times(0,\infty)^{|\vG_{1}\setminus\{e_{1},\dots,e_{m}\}|}.
\]
Moreover, the quotient map can be represented by
\[
	\begin{array}{rcl}
		(\prod_{e\in\vG_1} \sq_e)\circ\kappa_{\mathfrak{s}}: C_{\mathfrak{s}} &\longrightarrow& C_{\mathfrak{s}}^{\mathbb{R}}, \\[6pt](\tau_1,\ldots,\tau_{|\vG_1|}) &\longmapsto& (|\tau_1|^2,\ldots,|\tau_{|\vG_1|}|^2).
	\end{array}
\]

\end{prop}
\begin{proof}
    By \eqref{local cstar action}, $(\prod_{e\in\vG_1} \sq_e)\circ\kappa_{\mathfrak{s}}$ is invariant under the action of $(S^{1})^{|\vG_{1}|}\subset(\mathbb{C}^{*})^{|\vG_{1}|}$. We observe that the group homomorphism
    \[
    (\lambda_{e})_{e\in\vG_1}\in(\C^*)^{|\vG_1|}\longrightarrow(\lambda'_{e})_{e\in\vG_1}\in(\C^*)^{|\vG_1|},
    \]
    where
    \[
    \lambda_{e}'=
\begin{cases}
\lambda_{e},& e\in \Gamma_{1}\setminus S_{1}\\
\frac{\lambda_{e}}{\lambda_{e_{i}}},
& e\in S_i\setminus S_{i+1}\text{ and } e\neq e_i, \text{ for }1\leq i\leq m,\\
\frac{\lambda_{e_{i}}}{\lambda_{e_{i-1}}},
& e=e_i,\text{ for }1\leq i\leq m, \text{ where }\lambda_{e_0}=1,
\end{cases}
    \]
    is a group isomorphism. This isomorphism conjugates the action of $(\mathbb{C}^{*})^{|\vG_{1}|}$ given by \eqref{local cstar action} to the standard action
    \[
    (\lambda_{e})_{e\in\vG_{1}}\cdot(\gamma_{e})_{e\in\vG_{1}}= (\lambda_{e}\gamma_{e})_{e\in\vG_{1}}.
    \]
    Then our claim follows from basic topology.
\end{proof}
\Cref{chart on real} shows that $\{C_{\mathfrak{s}}^{\mathbb{R}}\}_{\mathfrak{s}\in\mathsf{MNS}(\vG)}$ is an open covering of $\resch$. One can check that the transition maps are smooth. This shows that $\resch$ is a smooth manifolds with corners.

Let $i$ be the following map
\[
i:
(0,\infty)^{|\vG_1|}\cong(\C^*)^{|\vG_1|}/(S^{1})^{|\vG_{1}|}\longrightarrow \cpsch/(S^{1})^{|\vG_{1}|}\cong\resch,
\]
then $\resch$, $i$, $\pi_\R^{\vG}$ defines a partial compactification defined as \Cref{compactification}.
\subsubsection{Differential forms on compactified Schwinger parameter spaces}\label{subsubsec: criterion}

Let $\vG$ be a directed graph. Given a smooth differential form $\alpha$
on $(0,\infty)^{|\vec{\Gamma}_1|}$, one may ask when $\alpha$ extends to
a smooth differential form on the partially compactified Schwinger parameter
space $\resch$. We will prove a criterion in terms of $\cpsch$. More specifically,
we will show that if there exists $N\in\mathbb{N}$ such that
\[
\Big(\prod_{e\in\vG_1}\tau_e^N\Big)
\Big(\prod_{e\in\vG_1}\sq_e\Big)^*\alpha
\]
extends to a smooth differential form on $\cpsch$, then $\alpha$ extends
to a smooth differential form on $\resch$. We will also give a criterion for
a differential form $\beta$ on $\cpsch$ to descend to a differential form
on $\resch$.

\begin{lem}\label{descent smooth function}
Let
\[
   g\in C^{\infty}\bigl((0,\infty)^{m+n}\bigr)
\]
and \[
\prod_{i=1}^{m+n}\sq:(\gamma_{i})_{i=1}^{m+n}\in (\C^{*})^{m+n}\longrightarrow(|\gamma_{i}|^{2})_{i=1}^{m+n}\in (0,\infty)^{m+n}.
\]
Suppose there exists $N\in\mathbb{N}$, such that
\[
    \Big(\prod_{i=1}^{m+n}\gamma_{i}^{N}\Big)\Big(\prod_{i=1}^{m+n}\sq\Big)^{*}g
\]
can be extended to a smooth function $f$ on $\mathbb{C}^{m}\times(\C^{*})^n$, then $g$ can be extended to a smooth function on $[0,\infty)^{m}\times(0,\infty)^n$.
\end{lem}
\begin{proof}
This is a special case of Schwarz-Poénaru theorem, where the compact Lie group is $G=(S^1)^{m+n}$, see \cite[Theorem 4.3 and 5.3 in Chapter XII]{GolubitskyStewartSchaeffer1988} for details. We present a proof for $m=1$, $n=0$. The general case can be proved similarly.

In the case of $m=1$, $n=0$, we have
\begin{equation}\label{S1 action simple formula}
f(e^{i\theta}\gamma)=e^{iN\theta}f(\gamma).
\end{equation}
We observe that the $l$-th derivative of $g$ is 
\[
g^{(l)}(|\gamma|^2)=\frac{1}{\gamma^{l+N}}\frac{\partial^{l}}{\partial \bar\gamma^{l}}f|_{\C^{*}}(\gamma).
\]
By \eqref{S1 action simple formula}, we have
\begin{equation}\label{s1 action formula}
\gamma\frac{\partial}{\partial \gamma} f-\bar{\gamma}\frac{\partial}{\partial \bar{\gamma}} f=N f.
\end{equation}
Therefore, each nonzero monomial in the Taylor expansion of $f$ at $0$ has the following form:
\[
\frac{1}{(l+N)!l!}\frac{\partial^{l+N}}{\partial \gamma^{l+N}}\frac{\partial^{l}}{\partial \bar\gamma^{l}}f(0)\gamma^{l+N}\bar{\gamma}^{l},
\]
where $l$ is non-negative. This shows that 
\[
\lim_{\gamma\to0}g^{(l)}(|\gamma|^2)=\lim_{\gamma\to0}\frac{1}{\gamma^{l+N}}\frac{\partial^{l}}{\partial \bar\gamma^{l}}f(\gamma)=\frac{1}{(l+N)!}\frac{\partial^{l+N}}{\partial \gamma^{l+N}}\frac{\partial^{l}}{\partial \bar\gamma^{l}}f(0).
\]
Thus
\[
    \lim_{t\to0^+} g^{(l)}(t)
\]
exists for every $l\geq 0$. Hence $g$ extends to a smooth function on
$[0,\infty)$.
\end{proof}
\begin{prop}
Let
\[
   \alpha\in \Omega^{*}\bigl((0,\infty)^{m+n}\bigr)
\]
and \[
\prod_{i=1}^{m+n}\sq:(\gamma_{i})_{i=1}^{m+n}\in (\C^{*})^{m+n}\longrightarrow(|\gamma_{i}|^{2})_{i=1}^{m+n}\in (0,\infty)^{m+n}.
\]
Suppose there exists $N\in\mathbb{N}$, such that
\[
    \Big(\prod_{i=1}^{m+n}\gamma_{i}^{N}\Big)\Big(\prod_{i=1}^{m+n}\sq\Big)^{*}\alpha
\]
can be extended to a smooth differential form on $\mathbb{C}^{m}\times(\C^{*})^n$, then $\alpha$ can be extended to a smooth differential form on $[0,\infty)^{m}\times(0,\infty)^n$.
\end{prop}
\begin{proof}
We first assume that $\alpha$ has the form
\begin{equation}\label{a simple term in differential form}
\alpha=g\,dt_{i_{1}}\wedge dt_{i_{2}}\wedge\cdots\wedge dt_{i_{k}},
\end{equation}
where $g\in C^{\infty}\bigl((0,\infty)^{m+n}\bigr)$, $i_{1}<i_{2}<\cdots<i_{k}$, and $0\leq k\leq m+n$. It suffices to prove that $g$ can be extended to a smooth function on $[0,\infty)^m\times(0,\infty)^n.$

Since
\[
   \Bigl(\prod_{i=1}^{m+n}\sq\Bigr)^*dt_j
   =
   \bar\gamma_j\,d\gamma_j+\gamma_j\,d\bar\gamma_j,
\]
the coefficient of $d\bar{\gamma}_{i_{1}}\wedge\cdots\wedge d\bar{\gamma}_{i_{k}}$ in
\[
   \Bigl(\prod_{i=1}^{m+n}\gamma_i^N\Bigr)
   \Bigl(\prod_{i=1}^{m+n}\sq\Bigr)^*\alpha
\]
is
\[
   \Bigl(\prod_{i=1}^{m+n}\gamma_i^N\Bigr)
   \Bigl(\prod_{j=1}^{k}\gamma_{i_{j}}\Bigr)
   \Bigl(\prod_{i=1}^{m+n}\sq\Bigr)^*g.
\]
Hence this function extends smoothly to $\C^m\times(\C^*)^n$. Multiplying by
\[
   \prod_{j\notin\{i_1,\dots,i_k\}}\gamma_j,
\] we find that $\Bigl(\prod_{i=1}^{m+n}\gamma_i^{N+1}\Bigr)\Bigl(\prod_{i=1}^{m+n}\sq\Bigr)^*g$
extends smoothly to $\C^m\times(\C^*)^n$. By \Cref{descent smooth function}, $g$
extends smoothly to $[0,\infty)^m\times(0,\infty)^n$. Since any differential form is a finite sum of such terms, our claim follows.
\end{proof}

By \Cref{chart on real}, we can get the following corollary.
\begin{cor}\label{descent of differential forms}
    Let $\vG$ be a directed graph, and $\alpha\in \Omega^{*}\bigl((0,\infty)^{|\vG_{1}|}\bigr)$.
Suppose there exists $N\in\mathbb{N}$, such that
\[
    \Big(\prod_{e\in\vG_1}\tau_e^N\Big)
\Big(\prod_{e\in\vG_1}\sq_e\Big)^*\alpha
\]
can be extended to a smooth differential form on $\cpsch$, then $\alpha$ can be extended to a smooth differential form on $\resch$.
\end{cor}

	\subsection{Smooth extension of Feynman graph integrands}\label{subsec: smooth extension}

In this subsection, we prove \Cref{thm:schwinger smooth extension0}. 
The strategy can be described as follows:
\begin{enumerate}[(1)]
    \item We construct a suitable wonderful compactification by iterated 
    blow-ups of
    \[
        M^{|\vG_0|}\times\C^{|\vG_1|},
    \]
    denoted by
    \[
        \bigl(M^{|\vG_0|}\times\C^{|\vG_1|}\bigr)_{\vG}.
    \]

    \item We construct a holomorphic map
    \[
        p^{\vG}\colon
        \bigl(M^{|\vG_0|}\times\C^{|\vG_1|}\bigr)_{\vG}
        \longrightarrow \cpsch,
    \]
    which lifts the canonical projection
    \[
        M^{|\vG_0|}\times\C^{|\vG_1|}
        \longrightarrow
        \C^{|\vG_1|}.
    \]

    \item We prove that there exists a sufficiently large integer $N>0$
    such that
    \[
        \left(\prod_e \tau_e^N\right)\widetilde W(\vG,\Phi)
    \]
    extends to a smooth differential form on
    \[
        \bigl(M^{|\vG_0|}\times\C^{|\vG_1|}\bigr)_{\vG}.
    \]
    Moreover, this extension is flat along the critical locus of
    $p^{\vG}$.

    \item By \Cref{pushforward is smooth2}, the pushforward of
    \[
        \left(\prod_e \tau_e^N\right)\widetilde W(\vG,\Phi)
    \]
    is a smooth differential form. Then
    \Cref{thm:schwinger smooth extension0} follows from
    \Cref{descent of differential forms}.
\end{enumerate}

Assume that $\vec{\Gamma}$ is a decorated directed graph. For each subset
$S\subseteq \vG_1$, define
\[
    \triangle_S
    :=
    \left\{
        (p_1,\ldots,p_{|\vG_0|})\in M^{|\vG_0|}
        :
        p_{t(e)}=p_{h(e)}
        \text{ for all } e\in S
    \right\}.
\]
For subsets $S\subseteq S'\subseteq \vG_1$, define
\[
    G_{S\subseteq S'}
    :=
    \triangle_S\times\mathcal{Z}_{S'}\subset M^{|\vG_0|}\times\C^{|\vG_1|},
\]
where $\mathcal{Z}_{S'}$ is defined in
\eqref{definition of Z}.

Let
\[
    \G^{\vG}
    :=
    \left\{
        G_{S\subseteq S'}
        :
        S\subseteq S'\subseteq\vG_1
    \right\}.
\]
For any
\[
    G_{S_1\subseteq S_1'},
    G_{S_2\subseteq S_2'}
    \in\G^{\vG},
\]
we have
\begin{align*}
    G_{S_1\subseteq S_1'}
    \cap
    G_{S_2\subseteq S_2'}
    &=
    \bigl(\triangle_{S_1}\cap\triangle_{S_2}\bigr)
    \times
    \bigl(\mathcal{Z}_{S_1'}\cap\mathcal{Z}_{S_2'}\bigr) \\
    &=
    \triangle_{S_1\cup S_2}
    \times
    \mathcal{Z}_{S_1'\cup S_2'} \\
    &=
    G_{(S_1\cup S_2)\subseteq(S_1'\cup S_2')}.
\end{align*}
Therefore, $\G^{\vG}$ is a simple arrangement and hence a building set.

Let $(M^{|\vG_0|}\times\C^{|\vG_1|})_\vG$ be the wonderful compactification of $M^{|\vG_0|}\times\C^{|\vG_1|}$ with respect to $\G^{\vG}$. 

\begin{lem}\label{smooth extension of integrand}

    For any $e\in \vG$, there is a canonical holomorphic map
    \[
    p_{e}:(M^{|\vG_0|}\times\C^{|\vG_1|})_\vG\longrightarrow \mathrm{Bl}_{\triangle\times\{0\}}(M\times M\times \mathbb{C}).
    \]
\end{lem}
\begin{proof}
    By \Cref{Wonderful compactification}, for any $e\in\vG_{1}$, we have a natrual holomorphic map
    \[
    (M^{|\vG_0|}\times\C^{|\vG_1|})_\vG\longrightarrow \mathrm{Bl}_{\triangle_{\{e\}}\times \mathcal{Z}_{\{e\}}}(M^{|\vG_0|}\times\C^{|\vG_1|}).
    \]
    If $e$ is not a self-loop, we have\[
    \mathrm{Bl}_{\triangle_{\{e\}}\times \mathcal{Z}_{\{e\}}}(M^{|\vG_0|}\times\C^{|\vG_1|})\cong M^{|\vG_0|-2}\times\C^{|\vG_1|-1}\times \mathrm{Bl}_{\triangle\times\{0\}}(M\times M\times \mathbb{C}),
    \]
    so we have a holomorphic map \[
    p_{e}:(M^{|\vG_0|}\times\C^{|\vG_1|})_\vG\longrightarrow \mathrm{Bl}_{\triangle\times\{0\}}(M\times M\times \mathbb{C})
    \]
    by projection to the factor $\mathrm{Bl}_{\triangle\times\{0\}}(M\times M\times \mathbb{C})$. If $e$ is a self-loop, we have
    \[
    \mathrm{Bl}_{\triangle_{\{e\}}\times \mathcal{Z}_{\{e\}}}(M^{|\vG_0|}\times\C^{|\vG_1|})\cong (M^{|\vG_0|-1}\times\C^{|\vG_1|-1})\times(M\times\mathbb{C}),
    \]
    where the second factor corresponds to $h(e)=t(e)\in\vG_{0}$ and $e\in\vG_{1}$. We notice that\[
    M\times \mathbb{C}\cong \widetilde{\triangle\times\mathbb{C}}\subset\mathrm{Bl}_{\triangle\times\{0\}}(M\times M\times \mathbb{C}),
    \]
    where $\widetilde{\triangle\times\mathbb{C}}$ is the dominant transform of $\triangle\times\mathbb{C}$. Therefore, we have a natural holomorphic map\[
    p_{e}:(M^{|\vG_0|}\times\C^{|\vG_1|})_\vG\longrightarrow \mathrm{Bl}_{\triangle\times\{0\}}(M\times M\times \mathbb{C}).
    \]
\end{proof}

Now, we can extend the integrand $\widetilde W(\vG,\Phi)$ defined by \eqref{Feynman graph integrand} to $(M^{|\vG_0|}\times\C^{|\vG_1|})_\vG$ in the following sense:

\begin{prop}\label{smooth extension of integrand1}
    Let $\vec{\Gamma}$ be a decorated directed graph and let
		$\Phi\in\Omega^\bullet(M^{|\vG_0|})$. Then there exists $N\in\mathbb{N}$, such that 
        \[
        \Big(\prod_{e\in\vG_1}\tau_e^N\Big)
\Big(\prod_{e\in\vG_1}\sq_e\Big)^*\widetilde W(\vG,\Phi)
        \]
        can be extended to a smooth differential form on $(M^{|\vG_0|}\times\C^{|\vG_1|})_\vG$. We still use the same notation to denote the smooth extension.
\end{prop}
\begin{proof}
    By \Cref{smooth extension of integrand} and \Cref{propogator is meromorphic}, there exists $N\in \mathbb{N}$, such that $p_{e}^{*}\tau^N\sq^* P_t$ can be extended to a smooth differential form on $\mathrm{Bl}_{\triangle\times\{0\}}(M\times M\times \mathbb{C})$. Similar property holds for any (higher) holomorphic derivatives of $P_{t}$. Our claim follows from the description of 
\[
\Big(\prod_{e\in\vG_1}\sq_e\Big)^*\widetilde W(\vG,\Phi)
\]
in local coordinates.
\end{proof}
The following proposition gives another characterization of $(M^{|\vG_0|}\times\C^{|\vG_1|})_\vG$.

\begin{prop}\label{another characterization}
    The space
    \[
    \bigl(M^{|\vG_0|}\times \C^{|\vG_1|}\bigr)_{\vG}
    \]
    is biholomorphic to an iterated blow-up of
    $M^{|\vG_0|}\times\cpsch$ along smooth centers.

    More precisely, for $S\subseteq S'\subseteq \vG_1$, define
    \[
    G'_{S\subseteq S'}
    :=
    \triangle_S\times D_{S'}
    \subset
    M^{|\vG_0|}\times\cpsch,
    \]
    where $D_{S'}\subset\cpsch$ is the smooth divisor corresponding
    to $\mathcal{Z}_{S'}$ under the construction in
    \Cref{properties of wonderful compactification}. Order the collection
    \[
    \left\{
        G'_{S\subseteq S'}
        :
        S\subseteq S'\subseteq \vG_1,\ 
        \triangle_S\neq M^{|\vG_0|}
    \right\}
    =
    \{G'_1,\dots,G'_{N'}\}
    \]
    compatibly with inclusion. Then
    \[
    \bigl(M^{|\vG_0|}\times \C^{|\vG_1|}\bigr)_{\vG}
    \cong
    \Bl_{\widetilde{G}'_{N'}}\cdots\Bl_{G'_1}
    \bigl(M^{|\vG_0|}\times\cpsch\bigr).
    \]
\end{prop}

\begin{proof}
    Let
    \[
    \mathcal{G}_1
    :=
    \left\{
        M^{|\vG_0|}\times\mathcal{Z}_{S'}
        :
        S'\subseteq \vG_1
    \right\}
    =
    \{G_1,\dots,G_N\},
    \]
    where the elements are ordered compatibly with inclusion. Since
    $\mathcal{G}_1$ is a simple arrangement, it is, in particular, a
    building set. Its wonderful compactification is
    \[
    \Bl_{\widetilde{G}_N}\cdots\Bl_{G_1}
    \bigl(M^{|\vG_0|}\times\C^{|\vG_1|}\bigr)
    \cong
    M^{|\vG_0|}\times\cpsch.
    \]

    For every $S\subseteq S'\subseteq\vG_1$, the dominant transform of
    \[
    G_{S\subseteq S'}
    =
    \triangle_S\times\mathcal{Z}_{S'}
    \]
    under these first $N$ blow-ups is
    \[
    G'_{S\subseteq S'}
    =
    \triangle_S\times D_{S'}.
    \]
    Consequently, the dominant transforms of the elements of
    $\mathcal{G}_{\vG}\setminus\mathcal{G}_1$ are precisely
    $G'_1,\dots,G'_{N'}$.

    Write
    \[
    \mathcal{G}_{\vG}\setminus\mathcal{G}_1
    =
    \{G_{N+1},\dots,G_{N+N'}\},
    \]
    where $G_{N+i}$ is chosen so that its dominant transform after the
    first $N$ blow-ups is $G'_i$. We choose this ordering compatibly
    with inclusion.

    We claim that, for every $1\leq j\leq N+N'$, the collection
    \[
    \mathcal{G}^{(j)}
    :=
    \{G_1,\dots,G_j\}
    \]
    is a building set. For $1\leq j\leq N$, this follows from
    \Cref{linear order G}. Now suppose that $N<j\leq N+N'$, and consider
    a nonempty intersection
    \[
    Y
    =
    G_{k_1}\cap\cdots\cap G_{k_p},
    \qquad
    1\leq k_1,\dots,k_p\leq j.
    \]

    If $G_{k_l}\in\mathcal{G}_1$ for every $l$, then, since
    $\mathcal{G}_1$ is a simple arrangement,
    \[
    Y\in\mathcal{G}_1\subseteq\mathcal{G}^{(j)}.
    \]
    Otherwise, $G_{k_l}\notin\mathcal{G}_1$ for at least one $l$. In
    that case,
    \[
    Y\subseteq G_{k_l}.
    \]
    Moreover, $Y\notin\mathcal{G}_1$: indeed, every element of
    $\mathcal{G}_1$ has first factor $M^{|\vG_0|}$, whereas
    $G_{k_l}\notin\mathcal{G}_1$ has a proper diagonal as its first
    factor. Hence
    \[
    Y\in\mathcal{G}_{\vG}\setminus\mathcal{G}_1.
    \]
    Since the ordering on
    $\mathcal{G}_{\vG}\setminus\mathcal{G}_1$ is compatible with
    inclusion and $Y\subseteq G_{k_l}$, the element $Y$ occurs no later
    than $G_{k_l}$. Therefore,
    \[
    Y\in\mathcal{G}^{(j)}.
    \]

    Thus $\mathcal{G}^{(j)}$ is closed under nonempty intersections and
    is therefore a simple arrangement. In particular,
    $\mathcal{G}^{(j)}$ is a building set, proving the claim.

    We may therefore apply
    \Cref{properties of wonderful compactification} successively in the
    order
    \[
    G_1,\dots,G_N,G_{N+1},\dots,G_{N+N'}.
    \]
    It follows that
    \[
    \begin{aligned}
        \bigl(M^{|\vG_0|}\times\C^{|\vG_1|}\bigr)_{\vG}
        &\cong
        \Bl_{\widetilde{G}_{N+N'}}
        \cdots
        \Bl_{\widetilde{G}_{N+1}}
        \left(
            \Bl_{\widetilde{G}_N}
            \cdots
            \Bl_{G_1}
            \bigl(M^{|\vG_0|}\times\C^{|\vG_1|}\bigr)
        \right)
        \\
        &\cong
        \Bl_{\widetilde{G}_{N+N'}}
        \cdots
        \Bl_{\widetilde{G}_{N+1}}
        \bigl(M^{|\vG_0|}\times\cpsch\bigr)
        \\
        &\cong
        \Bl_{\widetilde{G}'_{N'}}
        \cdots
        \Bl_{G'_1}
        \bigl(M^{|\vG_0|}\times\cpsch\bigr),
    \end{aligned}
    \]
    where the last biholomorphism follows from the identification of
    $\widetilde{G}_{N+i}$ with $G'_i$ for
    $1\leq i\leq N'$.
\end{proof}

\begin{cor}\label{construction of morphism to schwinger}
    There is a unique holomorphic map
    \[
    p^{\vG}:
    \bigl(M^{|\vG_0|}\times \C^{|\vG_1|}\bigr)_{\vG}
    \longrightarrow
    \cpsch
    \]
    lifting the projection
    \[
    M^{|\vG_0|}\times \C^{|\vG_1|}
    \longrightarrow
    \C^{|\vG_1|}.
    \]
\end{cor}

\begin{proof}
The map $p^{\vG}$ is the composition
\[
\bigl(M^{|\vG_0|}\times \C^{|\vG_1|}\bigr)_{\vG}
\longrightarrow
M^{|\vG_0|}\times\cpsch
\longrightarrow
\cpsch,
\]
where the first map is given by \Cref{another characterization} and the second map is the projection onto the second factor. Its uniqueness is immediate.
\end{proof}

To study the critical points of $p^{\vG}$, we will show that $p^{\vG}$ is a semistable morphism, which is an important notion in algebraic geometry. We refer to \cite{AbramovichKaru2000,AdiprasitoLiuTemkin2019,EnokizonoHashizume2025} for further background on this notion.

\subsubsection{Semistable morphisms}

We first observe that this map is semistable in the following sense.

\begin{defn}
    Let $X$ and $Y$ be complex manifolds with simple normal crossing divisors
    $D$ and $D'$, respectively. Let $m=\dim_\C(X)$ and $n=\dim_\C(Y).$ A holomorphic map $f:X\to Y$ is called a
    semistable morphism if, locally on $X$ and $Y$, it can be written as
    \[
    f(z_{1},\dots,z_{m})=(w_{1},w_{2},\dots,w_{n}), \quad
    w_{i}=\prod_{j=l_{i-1}+1}^{l_{i}}z_{j},
    \]
    for some integers
    \[
    0=l_{0}<l_{1}<\cdots<l_{n}\leq m,
    \]
    where $z_{1},\dots,z_{m}$ and $w_{1},\dots,w_{n}$ are local coordinates
    on $X$ and $Y$, respectively, such that
    \[
    D'=\{w_{1}w_{2}\cdots w_{k}=0\}
    \]
    and
    \[
    D=\{z_{1}z_{2}\cdots z_{l}=0\}
    \]
    for some $k$ and some $l\geq l_{k}$. Moreover,
    \[
    l_{i+1}=l_i+1,\quad k\leq i\leq n-1.
    \]
\end{defn}
The critical points of semistable morphisms are described as follows.
\begin{prop}\label{critical points of semistable}
    Let $X$ and $Y$ be complex manifolds with simple normal crossing divisors
    $D$ and $D'$, respectively. Let $f:X\to Y$ be a semistable morphism. Assume the following conditions are satisfied:
    \begin{enumerate}[(1)]
        \item $D$ and $D'$ have finitely many irreducible components.
        \item $f^{*}(D')=D$.
    \end{enumerate}
    Let $\{D'_{i}\}_{i=1}^{p}$ be the set of irreducible components of $D$, and 
    \[f^{*}(D'_{i})=\sum_{j=l_{i-1}+1}^{l_{i}}D_{j}, \quad \text{for }0=l_{0}<l_{1}<\cdots<l_{p},\]
    where $D_{j}$ is a irreducible component of $D$. Then the critical points of $f$ is given by
    \[
    \bigcup_{k=1}^{p}\Big(\bigcup_{l_{k-1}+1\leq i<j\leq l_{k}}(D_{i}\cap D_{j})\Big).
    \]
\end{prop}
\begin{proof}
    We only need to verify this locally. Let $p\in M$, by the semistable assumption, there exists coordinate charts $U\ni p$ and $V\ni f(p)$, such that $f$ has the following form:
    \[
    f(z_{1},\dots,z_{m})=(w_{1},w_{2},\dots,w_{n}), \quad
    w_{i}=\prod_{j=l'_{i-1}+1}^{l'_{i}}z_{j},\quad0=l'_{0}<l'_{1}<\cdots<l'_{n}\leq m,
    \]
    where \[
    D'\cap V=\{w_{1}w_{2}\cdots w_{k}=0\},\quad D\cap U=\{z_{1}z_{2}\cdots z_{l'_{k}}=0\},
    \]
    and 
    \[l'_{i+1}=l'_i+1,\quad\text{for } k\leq i\leq n-1.\]
    A direct computation of the Jacobian matrix shows that the critical
locus is
\[
        \bigcup_{k'=1}^{k}
        \Big(
        \bigcup_{l'_{k'-1}+1\leq i<j\leq l'_{k'}}
        \{z_i=z_j=0\}
        \Big).
\]
    Our claim follows.
\end{proof}

\begin{lem}\label{presemistable}
    Let
    \[
    p^{\vG}:
    \bigl(M^{|\vG_0|}\times \C^{|\vG_1|}\bigr)_{\vG}
    \longrightarrow
    \cpsch
    \]
    be the map defined in \Cref{construction of morphism to schwinger}.
    Then, for every $S\subseteq \vG_1$,
    \[
    (p^{\vG})^{*}(D_S)
    =
    \sum_{G_{S'\subset S}\in\mathcal{G}^{\vG}}
    \widetilde{G}_{S'\subset S}.
    \]
    In particular, $(p^{\vG})^{*}(D_S)$ is a simple normal crossing
    divisor.
\end{lem}

\begin{proof}
    Consider the commutative diagram
    \[
    \begin{tikzcd}
        \bigl(M^{|\vG_0|}\times \C^{|\vG_1|}\bigr)_{\vG}
        \arrow[r, "p^{\vG}"]
        \arrow[d, "p_{\mathcal{G}_{\vG}}"']
        &
        \cpsch
        \arrow[d, "\pi^{\vG}"]
        \\
        M^{|\vG_0|}\times \C^{|\vG_1|}
        \arrow[r, "\operatorname{pr}_2"']
        &
        \C^{|\vG_1|},
    \end{tikzcd}
    \]
    where $p_{\mathcal{G}_{\vG}}$ and $\pi^{\vG}$ are the canonical
    proper maps associated with the corresponding wonderful
    compactifications, and $\operatorname{pr}_2$ is the projection onto
    the second factor.

    Let $\mathcal{I}_{\mathcal{Z}_S}$ denote the ideal sheaf of
    $\mathcal{Z}_S$. By the commutativity of the diagram,
    \[
    (p^{\vG})^{*}(\pi^{\vG})^{*}\mathcal{I}_{\mathcal{Z}_S}
    =
    p_{\mathcal{G}_{\vG}}^{*}
    \operatorname{pr}_2^{*}\mathcal{I}_{\mathcal{Z}_S}.
    \]
    By \Cref{pull back divisor decomposition}, this gives
    \[
    (p^{\vG})^{*}
    \left(
        \sum_{S''\supseteq S}D_{S''}
    \right)
    =
    \sum_{S''\supseteq S}
    \sum_{G_{S'\subset S''}\in\mathcal{G}^{\vG}}
    \widetilde{G}_{S'\subset S''}.
    \tag{1}
    \]

    Set
    \[
    E_S
    :=
    \sum_{G_{S'\subset S}\in\mathcal{G}^{\vG}}
    \widetilde{G}_{S'\subset S}.
    \]
    We prove
    \[
    (p^{\vG})^{*}(D_S)=E_S
    \]
    by descending induction on $|S|$. The assertion is immediate for
    $S=\vG_1$. Suppose that it holds for every $S''\supsetneq S$.
    Subtracting the corresponding identities from $(1)$, we obtain
    \[
    \begin{aligned}
        (p^{\vG})^{*}(D_S)
        &=
        (p^{\vG})^{*}
        \left(
            \sum_{S''\supseteq S}D_{S''}
        \right)
        -
        \sum_{S''\supsetneq S}(p^{\vG})^{*}(D_{S''})
        \\
        &=
        \sum_{S''\supseteq S}E_{S''}
        -
        \sum_{S''\supsetneq S}E_{S''}
        \\
        &=E_S.
    \end{aligned}
    \]

    Finally, the divisors
    $\widetilde{G}_{S'\subset S}$ are boundary divisors of the
    wonderful compactification. Since its boundary is a simple normal
    crossing divisor, their reduced sum is also a simple normal crossing
    divisor.
\end{proof}

We are now ready to prove that $p^{\vG}$ is a semistable morphism.
To establish a stronger version of
\Cref{thm:schwinger smooth extension0}, we also consider the map
\[
p_1^{\vG}:
\bigl(M^{|\vG_0|}\times \C^{|\vG_1|}\bigr)_{\vG}
\longrightarrow
M\times\cpsch
\]
defined by
\[
p_1^{\vG}
:=
\bigl(\operatorname{pr}_1\circ p_{\mathcal{G}_{\vG}},\,p^{\vG}\bigr),
\]
where
\[
p_{\mathcal{G}_{\vG}}:
\bigl(M^{|\vG_0|}\times \C^{|\vG_1|}\bigr)_{\vG}
\longrightarrow
M^{|\vG_0|}\times\cpsch
\]
is the canonical proper map, and $\operatorname{pr}_1$ denotes the
projection onto the first $M$-factor.

\begin{prop}
    The maps $p^{\vG}$ and $p_{1}^{\vG}$ are semistable morphisms
    with respect to the natural boundary divisors induced by the
    wonderful compactifications.
\end{prop}

\begin{proof}
    We first prove that $p^{\vG}$ is semistable. Since semistability is
    a local property, it suffices to verify it over an open cover of
    $\cpsch$.

    Let
    \[
    \mathfrak{s}=(S_i,e_i)_{i=1}^{m}\in\mathsf{MNS}(\vG)
    \]
    be a marked nested sequence of edges. Recall that
    \[
    C_{\mathfrak{s}}
    :=
    \C^{m}\times
    (\C^{*})^{|\vG_1\setminus\{e_1,\dots,e_m\}|}
    \]
    is an open chart of $\cpsch$. It is enough to prove that
    \[
    p^{\vG}\big|_{(p^{\vG})^{-1}(C_{\mathfrak{s}})}
    :
    (p^{\vG})^{-1}(C_{\mathfrak{s}})
    \longrightarrow
    C_{\mathfrak{s}}
    \]
    is semistable.

    Let $\gamma_i:=\gamma_{e_i}$ for $1\leq i\leq m$. Then
    \[
    D_{S_i}\cap C_{\mathfrak{s}}
    =
    \{\gamma_i=0\}\times
    (\C^{*})^{|\vG_1\setminus\{e_1,\dots,e_m\}|}.
    \]
    For $S'\subseteq S_i$, set
    \[
    G''_{S'\subseteq S_i}
    :=
    \triangle_{S'}\times\{\gamma_i=0\}
    \subset
    M^{|\vG_0|}\times\C^m.
    \]
    Order the collection
    \[
    \left\{
        G''_{S'\subseteq S_i}
        :S'\subseteq S_i,\ 1\leq i\leq m
    \right\}
    =
    \{G''_1,\dots,G''_N\}
    \]
    compatibly with inclusion. By
    \Cref{another characterization}, we have
    \[
    (p^{\vG})^{-1}(C_{\mathfrak{s}})
    \cong
    \left(
        \Bl_{\widetilde{G}''_N}\cdots
        \Bl_{G''_1}
        \bigl(M^{|\vG_0|}\times\C^m\bigr)
    \right)
    \times
    (\C^{*})^{|\vG_1\setminus\{e_1,\dots,e_m\}|}.
    \]
    Thus it suffices to prove that the canonical map
    \[
    p:
    \Bl_{\widetilde{G}''_N}\cdots
    \Bl_{G''_1}
    \bigl(M^{|\vG_0|}\times\C^m\bigr)
    \longrightarrow
    \C^m
    \]
    is semistable.

    By \Cref{presemistable}, for each $1\leq i\leq m$, the divisor
    \[
    p^{*}\{\gamma_i=0\}
    \]
    is a simple normal crossing divisor, and the pullbacks corresponding
    to distinct $i$ have no common irreducible components. Since
    $\gamma_1,\dots,\gamma_m$ form a coordinate system on $\C^m$,
    the map $p$, and hence $p^{\vG}$, is semistable.

    The same argument applies to $p_{1}^{\vG}$. Indeed, let
    $U\subset M$ be a coordinate chart with coordinates
    $z_1,\dots,z_{\dim M}$, and let
    \[
    H_j:=\{z_j=0\}\subset U.
    \]
    The collection
    \[
    \begin{aligned}
        &\left\{
            G_{S\subseteq S'}
            \cap
            \bigl(
                U\times M^{|\vG_0|-1}\times\C^{|\vG_1|}
            \bigr)
            :
            S\subseteq S'
        \right\}
        \\
        &\qquad\cup
        \left\{
            H_j\times M^{|\vG_0|-1}\times\C^{|\vG_1|}
            :
            1\leq j\leq\dim M
        \right\}
    \end{aligned}
    \]
    is a building set. Adjoining the coordinate hypersurfaces does not
    change the associated wonderful compactification, and their total
    transforms coincide with their dominant transforms. Applying the
    preceding argument to the coordinate hypersurfaces
    $\{z_j=0\}$ and $\{\gamma_i=0\}$ shows that
    \[
    p_{1}^{\vG}:
    (p_{1}^{\vG})^{-1}(U\times C_{\mathfrak{s}})
    \longrightarrow
    U\times C_{\mathfrak{s}}
    \]
    is semistable. Therefore, $p_{1}^{\vG}$ is semistable.
\end{proof}

\begin{cor}\label{technical corollary}
    Let $p\in\bigl(M^{|\vG_0|}\times \C^{|\vG_1|}\bigr)_{\vG}$ be a critical point of
    \[
p_{1}^{\vG}:\bigl(M^{|\vG_0|}\times \C^{|\vG_1|}\bigr)_{\vG}\longrightarrow M\times\cpsch,
\]
then there exists some $S\subset\vG_{1}$, which has at least one edge that is not a self-loop, such that $p$ is in the intersection of two different irreducible components of $(p_{1}^{\vG})^{*}(M\times D_{S})$.
\end{cor}
\begin{proof}
    By \Cref{presemistable}, their exists some $S\subset\vG_{1}$, such that $p$ is in the intersection of two different irreducible components of $(p_{1}^{\vG})^{*}(M\times D_{S})$. By \Cref{pull back divisor decomposition}, we have
    \[
    (p_{1}^{\vG})^{*}(M\times D_{S})=(p^{\vG})^{*}( D_{S})=\sum_{G_{S\subset S'}\in \mathcal{G}^{\vG}}\widetilde{G}_{S\subset S'}.
    \]
    If all the elements in $S$ are self-loops, we have \[
    G_{S\subset S'}=M^{|\vG_{0}|}\times D_{S}
    \]
    for any $S'\supset S$, so there is only one irreducible component. This contradicts with our choice of $S$.
\end{proof}

\subsubsection{Smooth extension theorems}

Now, we can prove \Cref{thm:schwinger smooth extension0}.

\begin{lem}\label{vanishing at critical points}
    Let $S\subseteq \vG_1$, and let $\widetilde{G}'$ be an irreducible
    component of
    \[
    (p_{1}^{\vG})^{*}(M\times D_S)
    \]
    distinct from $\widetilde{G}_{S\subseteq S}$. Then there exists
    $N\in\mathbb{N}$ such that
    \begin{equation}\label{integrand1}
        \left(\prod_{e\in\vG_1}\tau_e^N\right)
        \left(\prod_{e\in\vG_1}\sq_e\right)^*
        \widetilde{W}(\vG,\Phi)
    \end{equation}
    extends smoothly to
    \[
    \bigl(M^{|\vG_0|}\times\C^{|\vG_1|}\bigr)_{\vG}
    \]
    and is flat along $\widetilde{G}'$. Moreover, \eqref{integrand1} is flat
    along $\Crit(p_{1}^{\vG})$.

    Here,
    \[
    G_{S\subseteq S}:=\triangle_S\times D_S,
    \]
    and $\widetilde{G}_{S\subseteq S}$ denotes its dominant transform.
\end{lem}

\begin{proof}
    By \Cref{smooth extension of integrand1}, there exists $N\in\mathbb{N}$, such that \eqref{integrand1} is smooth over $\bigl(M^{|\vG_0|}\times \C^{|\vG_1|}\bigr)_{\vG}$. By our construction, $\widetilde{G}'=\widetilde{\triangle_{S'}\times D_{S}}$ for some $S'\supset S$, and $\triangle_{S'}\supsetneq\triangle_{S}$. Let 
    \[p\in \widetilde{G}'\setminus\widetilde{G}_{S\subset S},\]
    we have 
    \[
    p_{\mathcal{G}_{\vG}}(p)\in (\triangle_{S'}\setminus\triangle_{S})\times D_{S}\subset M^{|\vG_0|}\times\cpsch,
    \]
    where 
    \[
    p_{\mathcal{G}_{\vG}}:\bigl(M^{|\vG_0|}\times \C^{|\vG_1|}\bigr)_{\vG}\longrightarrow M^{|\vG_0|}\times\cpsch
    \]
    is the canonical proper map. We write 
    \[
    p_{\mathcal{G}_{\vG}}(p)=(\mathbf{z}_{1},\mathbf{z}_{2},\dots,\mathbf{z}_{|\vG_{1}|},\tau)\in M^{|\vG_0|}\times\cpsch,
    \]
    then there exists $e\in S'\setminus S$, such that $
    \mathbf{z}_{s(e)}\neq \mathbf{z}_{t(e)}.
    $
    This shows that 
    \[
    p_{e}(p)\in (M\times M\setminus \triangle)\times {0}\subset \mathrm{Bl}_{\triangle\times\{0\}}(M\times M\times \mathbb{C}),
    \]
    where $p_{e}$ is defined in \Cref{smooth extension of integrand}. By \Cref{propogator is meromorphic}, $p_{e}^{*}(\tau_{e}^{N}\mathrm{sq}^{*}P_{t})$ is flat at $p$, so \eqref{integrand1} is also flat at $p$. Since $\widetilde{G}'\setminus\widetilde{G}_{S\subset S}$ is dense in $\widetilde{G}'$, \eqref{integrand1} is flat along $\widetilde{G}'$.

    Finally, by \Cref{technical corollary}, for any critical point $p$, there exists $S\subset\vG_{1}$, such that such that $p$ is in the intersection of two different irreducible components of 
    \[
    (p_{1}^{\vG})^{*}(M\times D_{S})=\sum_{G_{S\subset S'}\in \mathcal{G}^{\vG}}\widetilde{G}_{S\subset S'}.
    \]
    However, \eqref{integrand1} is flat along any components of $(p_{1}^{\vG})^{*}(M\times D_{S})$ that is not $\widetilde{G}_{S\subset S}$, so \eqref{integrand1} is flat at any critical points.
\end{proof}
	Let
\[
p_{1}^{\vG,\mathbb{R}}:M^{|\vG_0|}\times(0,\infty)^{|\vG_{1}|}\longrightarrow M\times(0,\infty)^{|\vG_{1}|}
\]
be the projection onto the first $M$-factor and the $(0,\infty)^{|\vG_{1}|}$-factor.

\begin{thm}\label{smooth extension to schwinger2}
    For any decorated directed graph $\vG$ and 
		any $\Phi\in\Omega^\bullet(M^{|\vG_0|})$, the fiber integral
		\[
		 p_{1*}^{\vG,\mathbb{R}}\big(\widetilde W(\vG,\Phi)\big)
		\]
        can be extended to a smooth differential form on $M\times\resch$.
		 Moreover,
		the assignment
		\[
		\Phi\in\Omega^\bullet(M^{|\vG_0|})
		\longmapsto
		p_{1*}^{\vG,\mathbb{R}}\big(\widetilde W(\vG,\Phi)\big)\in\Omega^\bullet\big(M\times\resch\big)
		\]
		defines a continuous linear map between topological vector spaces.
\end{thm}
   \begin{proof}
        By \Cref{vanishing at critical points}, there exists $N\in\mathbb{N}$ such that \begin{equation}\Big(\prod_{e\in\vG_1}\tau_e^N\Big)
\Big(\prod_{e\in\vG_1}\sq_e\Big)^*\widetilde W(\vG,\Phi)
     \end{equation}
    is smooth over $\bigl(M^{|\vG_0|}\times \C^{|\vG_1|}\bigr)_{\vG}$ and flat along the critical points of $p_{1}^{\vG}$. Therefore, by \Cref{pushforward is smooth2},
    \[
    \Big(\prod_{e\in\vG_1}\tau_e^N\Big)
\Big(\prod_{e\in\vG_1}\sq_e\Big)^*p_{1*}^{\vG,\mathbb{R}}\big(\widetilde W(\vG,\Phi)\big)=p_{1*}^{\vG}\Big(\prod_{e\in\vG_1}\tau_e^N\Big)
\Big(\prod_{e\in\vG_1}\sq_e\Big)^*\big(\widetilde W(\vG,\Phi)\big)
    \]
    is a smooth differential form over $M\times\cpsch$. By \Cref{descent of differential forms}, 
    \[
    p_{1*}^{\vG,\mathbb{R}}\big(\widetilde W(\vG,\Phi)\big)
    \]
    can be extended to a smooth form on $M\times \resch$. To prove the continuity, we observe that 
    \begin{equation}\label{strange map}
    \Big(\prod_{e\in\vG_1}\tau_e^N\Big)
\Big(\prod_{e\in\vG_1}\sq_e\Big)^*:\Omega^*\big(M\times\resch\big)\longrightarrow\Omega^\bullet\big(M\times\cpsch\big)
    \end{equation}
    is a continuous map with closed image, so by Banach–Schauder theorem for Fréchet spaces, the inverse map is continuous. We observe that $p_{1*}^{\vG,\mathbb{R}}$ is the composition of the map in \eqref{strange map}, the map $p_{1*}^{\vG}$, and the inverse of the map in \eqref{strange map} onto its image. Hence the assignment
		\[
		\Phi\in\Omega^\bullet(M^{|\vG_0|})
		\longmapsto
		p_{1*}^{\vG,\mathbb{R}}\big(\widetilde W(\vG,\Phi)\big)\in\Omega^\bullet\big(M\times\resch\big)
		\] is continuous.
    \end{proof}
	
	\begin{cor}\label{smooth extension to schwinger3}
	    For any decorated directed graph $\vG$ and 
		any $\Phi\in\Omega^\bullet(M^{|\vG_0|})$, the fiber integral
		\[
		 \widehat W(\vG,\Phi)=\int_{M^{|\vG_{0}|}} \widetilde W(\vG,\Phi)
		\]
        can be extended to a smooth differential form on $\resch$.
		 Moreover,
		the assignment
		\[
		\Phi\in\Omega^\bullet(M^{|\vG_0|})
		\longmapsto
		\widehat W(\vG,\Phi)\in\Omega^\bullet\big(\resch\big)
		\]
		defines a continuous linear map between topological vector spaces.
	\end{cor}
    \begin{proof}
        By \Cref{smooth extension to schwinger2}, the composition of $p_{1*}^{\vG,\mathbb{R}}$ with the fiber integral along \[
        M\times \resch\longrightarrow \resch
        \] is continuous.
    \end{proof}

	\subsection{Zeta-regularized heat kernel formulation}\label{subsec: heat kerenl zeta regularization}
In this subsection, we introduce the zeta-regularized heat kernel formulation.

    First, we can slightly strengthen the result of
\Cref{interability of Feynman graph integrand} as follows.
	\begin{prop}\label{prop:W-hat-large-time-decay}
		For every decorated directed graph $\vG$, every $\Phi\in\Omega^\bullet(M^{|\vG_0|})$, there exist
		constants $\delta,C>0,r\in\Z_{>0}$ such
		that on $(0,\infty)^{|\vG_1|}$, there exists $f_{\Phi}\in C^\infty\big((0,\infty)^{|\vG_1|}\big)\cap L^1\big((0,\infty)^{|\vG_1|}\big)$ satisfying \begin{enumerate}[(1)]
			\item \[\|f_\Phi\|_{L^1\big((0,\infty)^{\vG_1}\big)}\leq C;\]
			\item \[
			\big|
			\widehat W(\vG,\Phi)_\top(\bft)
			\big|
			\le
			C|{f_{\Phi}(\bft)}|\exp\Big(-\delta\sum_{e\in\vG_1}t_e\Big),
			\quad
			\bft=(t_e)_{e\in\vG_1}.
			\]
		\end{enumerate}
		Here, for $w\in\Omega^\bullet((0,\infty)^{|\vG_1|})$, $w_\top\in C^\infty((0,\infty)^{|\vG_1|})$ denotes the coefficient of top degree components of $w$.

	\end{prop}

	\begin{proof}
   On $(0,\infty)^{|\vG_1|}$, write
\[
        \phi(\bft):=\widehat W_{\top}(\vG,\Phi).
\]
Fix $\delta\in(0,c)$, where $c$ is the smallest constant appearing in the
 estimate \eqref{exponential decay}, and set
\[
        f_{\Phi}(\bft)
        :=
        \phi(\bft)\exp\Big(\delta\sum_{e\in\vG_1}t_e\Big).
\]
We now prove that $f_{\Phi}\in L^1\big((0,\infty)^{|\vG_1|}\big)$.

As in the proof of \Cref{interability of Feynman graph integrand}, consider the
decomposition
\[
        (0,\infty)^{|\vG_1|}
        =
        \bigcup_{S\subseteq\vG_1}
        [1,\infty)^{|S|}
        \times
        (0,1]^{|\vG_1\setminus S|}.
\]
Since $t_e\le1$ for $e\in\vG_1\setminus S$, the exponential decay
estimate \eqref{exponential decay}, together with Tonelli's theorem, gives
\[
\begin{aligned}
        &\int_{[1,\infty)^{|S|}
        \times(0,1]^{|\vG_1\setminus S|}}
        |f_{\Phi}(\bft)|\,d\bft    \le
        C\int_{[1,\infty)^{| S|}}
        e^{\delta\sum_{e\in S}t_e}
        e^{-c\sum_{e\in S}t_e}\,d\bft'
        <\infty .
\end{aligned}
\]
Summing over all subgraphs $\vG'\subseteq\vG$, we obtain
$f_{\Phi}\in L^1$.
	\end{proof}
	
	We define the zeta-regularized heat-kernel Feynman graph integral by
	\begin{equation}\label{def:zeta heat kernel graph integral}
		W^{\mathrm{HK}}(\vG,\Phi,\bfs)
		:=
		\int_{\resch}
		\prod_{e\in\vG_1}\frac{t_e^{s_e}}{\Gamma(s_e+1)}
		\widehat W(\vG,\Phi),\bfs\in \cH^{|\vG|_1}.
	\end{equation}
	Here $t_e$ is regarded as a smooth function on $\resch$, given by the
	composition of the blow-up map
	$
	\resch\to(0,+\infty)^{|\vG_1|}
	$
	with the projection to the component related to $e$.

	\begin{prop}\label{prop:zeta heat kernel holomorphic}
		The function
		\[
		W^{\mathrm{HK}}(\vG,\Phi,\bfs)
		\]
		is holomorphic on $\cH^{|\vG_1|}$.
		In particular,
		\[W^{\mathrm{HK}}(\vG,\Phi,\mathbf{0})=
		\int_{\resch}
		\widehat W(\vG,\Phi).\]
	\end{prop}
	
	\begin{proof}
		By \Cref{smooth extension to schwinger3}, the configuration-space integral
		$\widehat W(\vG,\Phi)$ is smooth on the compactified Schwinger space.
		Since $\Re(s_e)\ge0$ for all $e$, and by
		\Cref{prop:W-hat-large-time-decay} , it follows that the resulting integrand is integrable.
		Moreover, differentiation with respect to the parameters $s_e$ only
		introduces powers of $\ln t_e$, which remain integrable. Therefore the integral depends
		holomorphically on $\bfs\in\cH^{|\vG_1|}$.
	\end{proof}

	\section{Zeta regularized Cauchy principal values}
	\label{sec:multi-zeta-cpv}

	This section proves a zeta regularization theorem for
	Feynman graph integrals in the sense of generalized Cauchy principal value.  The construction has three ingredients.  First, we
	introduce a singularity class which is stable under coordinate changes
	compatible with normal crossing divisors and which keeps track of the
	multi-zeta parameters.  Second, we show that the zeta-regularized propagator
	belongs to this class after blowing up the diagonals.  Third, we prove Cauchy principal value statements for this singularity class.

	\subsection{Statement of the main theorem}
	
	Let
	\[
	\cH:=\{s\in\C:\Re(s)\ge0\}
	\]
	be the closed right half-plane.  For $\Re(s)\gg0$, define the
	zeta-regularized propagator by
	\begin{equation}\label{def zeta propagator}
		P^s
		:=
		\frac{1}{\Gamma(s+1)}
		\int_0^{+\infty}
		t^s\,dt\wedge(\bar\partial_{\tE}^*\otimes\operatorname{id})H_t .
	\end{equation}
	The singular behavior of $P^s$ along the diagonal will be studied in
	\Cref{subsection propagators}.
	
	Let $\vG$ be a decorated directed graph as in \Cref{defn:directed-graph}.  For
	each edge $e\in\vG_1$, let $s_e$ be a complex parameter and write
	\[
	\bfs=(s_e)_{e\in\vG_1}.
	\]
	Replacing the propagator in Schwinger space $P_{t_{e}}$ assigned to the edge $e$ by
	$
	\frac{1}{\Gamma(s_e+1)}t_{e}^{s_e}P_{t_{e}},
	$
	and using the same contraction rule as in the definition of
	$\check{W}(\vG,\Phi)$ in \eqref{eq: check W}, for each
	$\Phi\in\Omega^\bullet(M^{|\vG_0|})$ we obtain a form on
	$\Conf_{\vG_0}(M)$, denoted by
	\[
	\check{W}(\vG,\Phi,\bfs).
	\]
	When $\Re(s_e)\gg0$ for every edge $e$, this form is absolutely
	integrable over $\Conf_{\vG_0}(M)$, and the resulting integral is
	holomorphic in $\bfs$ by dominated convergence.
	
	The main result of this section is that:
	
	\begin{thm}\label{thm:zeta graph integral holomorphic continuation}
		Let $(M,E,\omega)$ be a triple as in \Cref{Propagators in holomorphic field theories}, and
		let $\vG$ be a decorated directed graph. Then there exists a holomorphic function
		\[
		W^{\mathrm{CPV}}(\vG,\Phi,-):\cH^{|\vG_1|}\longrightarrow \mathbb C
		\]
		with the following property: for every $\bfs\in \cH^{|\vG_1|}$ such that
		$\Re(s_e)\gg 0, \forall e\in\vG_1$, one has
		\[
		W^{\mathrm{CPV}}(\vG,\Phi,\bfs)
		=
		\int_{\Conf_{\vG_0}(M)}
		\check{W}(\vG,\Phi,\bfs).
		\]
		Equivalently, the multi-zeta regularized Cauchy principal value renormalization of the Feynman graph integral admits an
		analytic continuation on $\cH^{|\vG_1|}$, in the sense of the
		generalized Cauchy principal value constructed below.
		
		In particular, $W^{\mathrm{CPV}}(\vG,\Phi,{0})=W^{\mathrm{CPV}}(\vG,\Phi)$.
	\end{thm}
	
	The rest of this section is devoted to the proof of
	\Cref{thm:zeta graph integral holomorphic continuation}. 
	
	\subsection{Generalized divisorial type singularities}
	We introduce a singularity class, called generalized divisor-type
singularities, which is stable under coordinate changes compatible with normal
crossing divisors and keeps track of the multi-zeta parameters. The motivation
for this class comes from \Cref{lem: integrate good part}.
	
	We first introduce the function $\ln(u,s)$, which we call the generalized
	logarithm function, where $s\in\C$ and $u>0$, defined by
	\[
	\ln(u,s):=\begin{cases}
		\frac{u^s-1}{s},\text{ if } s\neq0\\
		\ln(u), \text{ if }s=0.
	\end{cases}
	\]
	
	\begin{rem}
		We use the term generalized logarithm because
		\[
		\lim_{s\to 0}\frac{u^s-1}{s}
		=\ln(u).
		\]
	\end{rem}
	
For any
	\[
	\bfz=(z_1,\ldots,z_{m+n})\in\C^{m+n},\quad
	\bfi=(i_1,\ldots,i_m)\in\Z^m,
	\]
	we set
	\[
	\bfz^\bfi=z_1^{i_1}\cdots z_m^{i_m}.
	\]
	Let
	\[
	\bfL=(L_1,\ldots,L_m):\C^\ell\longrightarrow\C^m
	\]
	be a linear map whose coefficient matrix has entries in $\Z_{\ge0}$ for some $\ell\in\Z_{>0}$. For
	\[
	\bfj=(j_1,\ldots,j_m)\in\Z_{\ge0}^m,
	\quad
	\bfs=(s_1,\ldots,s_\ell)\in\C^\ell,
	\]
	we define
	\[
	\ln^{\bfL}(\bfz,\bfj,\bfs)
	:=
	\Big(\ln(|z_1|^2,L_1(\bfs))\Big)^{j_1}
	\cdots
	\Big(\ln(|z_m|^2,L_m(\bfs))\Big)^{j_m},
	\]
    and\[r^{\bfL}(\bfz,\bfj,\bfs)=|z_1|^{2j_1L_1(\bfs)}\cdots|z_m|^{2j_mL_m(\bfs)}\]
	When $\ell=m$ and $\bfL={\mathrm{id}}$, we omit the superscript and simply write
	\[
\ln(\bfz,\bfj,\bfs):=\ln^{{\mathrm{id}}}(\bfz,\bfj,\bfs)\quad\text{and}\quad r(\bfz,\bfj,\bfs):=r^{\id}(\bfz,\bfj,\bfs).
	\]
    When $\bfL=0$, we write
    \[\ln(\bfz,\bfj):=\ln^{0}(\bfz,\bfj,\bfs).\]
	Recall that $\cH=\{s\in\C:\Re(s)\ge0\}$, and set
    \be\label{eq: defn of check H}\ccH:=\cH\setminus\cup_{j\in\Z_{>0}}\{z\in\C: |z-j|\leq 1/4\}.\ee
	
	\begin{defn}\label{defn of extended divisor}
		Given a differential form
		\[
		\alpha\in\Omega^{\bullet,\bullet}\big((\C^*)^m\times\C^n\big)
		\]
		depending on parameters $\bfs\in\cH^\ell$, for some
		$\ell\in\Z_{>0}$, we say that $\alpha$ has generalized
		\textbf{divisorial type singularities} along the principal divisor
		\[
		z_1z_2\cdots z_m=0
		\]
if it is a finite linear combination of terms
\[
        \bfz^{-\bfi}\,g(\bfz,\bfs),
        \quad
        \bfi\in\Z^m,
\]
where $g$ satisfies: \begin{enumerate}[(1)]
    \item On $\ccH^\ell$, the function $g$ is a finite linear combination of terms
\[
        \ln^\bfL(\bfz,\bfk,\bfs)r^{\bfL'}(\bfz,\bfl,\bfs)\,\gamma(\bfz,\bfs),
        \quad
        \bfk,\bfl\in\Z_{\ge0}^m, 
\]
where
$
        \gamma(\cdot,\bfs)\in\Omega^{\bullet,\bullet}(\C^{m+n})
$
 and holomorphic in $\bfs$ on
$
        \ccH^\ell.
$
Here
	\[
	\bfL,\bfL':\C^\ell\longrightarrow\C^m
	\]
	are some linear maps whose coefficient matrix has entries in $\Z_{\ge0}$.
\item Let
\[
        \bfj=(j_1,\ldots,j_\ell)\in\Z_{\ge0}^\ell\setminus\{0\} \subset\cH^\ell.
\]
On 
$
   \big\{\bfs\in\cH^\ell:     |\bfs-\bfj|<\frac12\big\},
$
the function $g$ is a finite linear combination of terms
\[
        \ln^\bfL(\bfz,\bfk,\bfs-\bfj) r^{\bfL'}(\bfz,\bfl,\bfs)\,\beta(\bfz,\bfs),
        \quad
        \bfk,\bfl\in\Z_{\ge0}^m,
\]
where
$
        \beta(\cdot,\bfs)\in\Omega^{\bullet,\bullet}(\C^{m+n})
$
and holomorphic in $\bfs$ on
$
        \big\{\bfs\in\cH^\ell:     |\bfs-\bfj|<\frac12\big\}.$
        Here
	\[
	\bfL,\bfL':\C^\ell\longrightarrow\C^m
	\]
	are some linear maps whose coefficient matrix has entries in $\Z_{\ge0}$.
\end{enumerate}

	\end{defn}
	\begin{rem}\label{rmk: independence of compatible coordinates}
		The preceding definition is independent of the choice of coordinates
		compatible with the normal crossing divisor (see \Cref{defn: compatible coordinate change}). Indeed, under such a coordinate
		change $F:U\to V,\bfz\mapsto \bfw$, after shrinking $U$ if necessary, one has
		\[
		w_a\circ F
		=
		u_a(\bfz)\,z_1^{q_{a1}}\cdots z_m^{q_{am}},
		\quad a=1,\ldots,m,
		\]
		where $u_a$ is nowhere vanishing and $q_{ab}\in\Z_{\ge0}$. 
		Also noting
		\[
		\ln(xy,s)=x^s\ln(y,s)+\ln(x,s).
		\]
	\end{rem}

	\subsection{Zeta-regularized propagators and their singularities}
	\label{subsection propagators}
	
	We now study the singularities of the zeta-regularized propagator
	$P^s$ defined in \eqref{def zeta propagator}.  The large-time part of this
	integral is smooth.  More precisely, if $\eta\in C_c^{\infty}([0,1))$ satisfies
	$\eta(t)=1$ for $t$ sufficiently small, then
	\begin{equation}\label{smooth large time part}
		\frac{1}{\Gamma(s+1)}
		\int_0^{+\infty}
		\big(1-\eta(t)\big)t^s\,dt\wedge(\bar\partial_{\tE}^*\otimes\operatorname{id})H_t
	\end{equation}
	defines a smooth section on $M\times M$, depending holomorphically on
	$s\in\cH$.  Hence the singularities of
	$P^s$ are determined by its short-time part
	\begin{equation}\label{short time zeta propagator}
		\frac{1}{\Gamma(s+1)}
		\int_0^{+\infty}
		\eta(t)t^s\,dt\wedge(\bar\partial_{\tE}^*\otimes\operatorname{id})H_t.
	\end{equation}
	\def\cR{{\mathcal{R}}}
We first prove the following fundamental result in analysis:
\begin{lem}\label{lem: integrate good part}
Fix $m\in\mathbb Z_{\geq0}$.
For $r>0$ and $s\in\cH$, define
\[
        J(r,s)
        =
        \frac{r^{m}}{\Gamma(s+1)}
        \int_0^{+\infty}
        e^{-\frac{r}{t}}t^{s-m-1}c(t)\,{dt},
\]
where $c\in C_c^\infty([0,1))$.

Then the following holds.

\begin{enumerate}[(1)]
\item \label{lem45 item1}For every $\ell\in\mathbb Z_{>0}$, let
\[
        B_\ell
        =
        \bigl\{
                s\in\cH:
                |s-\ell|<\tfrac12
        \bigr\}.
\]
There exist functions $u_\ell(r,s)$, $v_\ell(r,s)$ and $w_\ell(r,s)$, smooth in
$r$ on $[0,\infty)$ and holomorphic in $s$ on 
$B_\ell$, such that
\[
        J(r,s)
        =
        \ln(r,s-\ell)u_\ell(r,s)
        +
       r^s v_\ell(r,s)+w_\ell(r,s).
\]

\item\label{lem 45 item 2}
There exist functions $a(r,s)$, $b(r,s)$ and $c(r,s)$, smooth in $r$ on $[0,\infty)$ and holomorphic in $s$ on $\ccH$, such that
\[
        J(r,s)
        =
        \ln(r,s) a(r,s)+r^sb(r,s)+c(r,s).
\]
\end{enumerate}
\end{lem}

\begin{proof}
Let
\[
        \widehat c(z)
        =
        \int_0^\infty
        t^{z-1} c(t)\,dt
\]
be the Mellin transform of $c$. It extends meromorphically to $\C$, with
simple poles at $z=0,-1,-2,\ldots$. See also \eqref{eq: expression of hat c}.

Fix $d>1+\Re(s)$. By the inverse Mellin transform, we have
\[
        t^s c(t)
        =
        \frac{1}{2\pi\sqrt{-1}}
        \int_{L_d}
        \widehat c(z)t^{s-z}\,dz.
\]Here\[L_{d}
        :=
        \bigl\{
                z\in\C:\Re(z)=d
        \bigr\}.\]
Consequently, for $r>0$, using
$
        \int_0^\infty
        e^{-\frac{r}{t}}t^{s-m-z-1}\,dt
        =
        r^{s-m-z}\Gamma(m+z-s),
$
\be\label{eq: computation of J r s}
\begin{split}
        J(r,s)
        &=
        \frac{r^{m}}{2\pi\sqrt{-1}\,\Gamma(s+1)}
        \int_{L_d}
        \widehat c(z)
        \Big(
                \int_0^\infty
                e^{-\frac{r}{t}}t^{s-m-z-1}\,dt
        \Big)dz        =
        \frac{1}{2\pi\sqrt{-1}}
        \int_{L_d}
        h_s(z)r^{s-z}\,dz,
\end{split}
\ee
where
\[
        h_s(z)
        =
        \frac{\widehat c(z)\Gamma(m+z-s)}
        {\Gamma(s+1)}.
\]
The poles of $h_s(z)$ are
\[
        P_{m,s}:=
        \bigl\{
        s-m-k:
        k\in\mathbb Z_{\geq0}
        \bigr\}
        \cup
        \bigl\{
        -j:j\in\Z_{\geq0}
        \bigr\}.
\]
They are simple unless
\[
        -j=s-m-k.
\]
Thus, for $s\in\cH$, two poles can coincide only when
$s\in\mathbb Z_{\geq0}$.

Repeated integration by parts gives
\be\label{eq: expression of hat c}
        \widehat c(z)
        =
        \frac{(-1)^N}
        {z(z+1)\cdots(z+N-1)}
        \int_0^\infty
        t^{z+N-1}c^{(N)}(t)\,dt .
\ee
Together with Stirling's formula, it follows that, if
$|\Im(z)|>|\Im(s)|+1$, then, for every $N\in\Z_{>0}$, there exists
$C_N>0$ such that
\be\label{eq: estimate of h s}
        |h_s(z)|
        \leq
        C_N |\Im(z)|^{-N}.
\ee
Fix any $A>0$ such that the vertical line
\[
        L_{-A}
        =
        \bigl\{
                z\in\C:\Re(z)=-A
        \bigr\}
\]
does not contain any pole of $h_s$. By \eqref{eq: estimate of h s} and the
residue theorem, we have
\be\label{eq: move mellin contour}
\begin{split}
        J(r,s)
        =
        \frac{1}{2\pi\sqrt{-1}}
        \int_{L_{-A}}
        h_s(z)r^{s-z}\,dz
        +
        \sum_{z_0\in P_{m,s}\cap\{z:\Re(z)>-A\}}
        \Res_{z=z_0}
        \big(
                h_s(z)r^{s-z}
        \big).
\end{split}
\ee
On $L_{-A}$, we have
$
        |r^{s-z}|
        =
        r^{\Re(s)+A}.
$
Thus \eqref{eq: estimate of h s} implies
\[
        \Big|
        \int_{L_{-A}}
        h_s(z)r^{s-z}\,dz
        \Big|
        =
        O\big(r^{\Re(s)+A}\big),
        \quad
        0<r\leq1.
\]
Since $A$ is arbitrary, assuming that $s\notin\Z_{\geq0}$, we obtain the
asymptotic expansion:
\be\label{eq: asym of J s r}
        J(r,s)
        \sim
        \sum_{j=0}^\infty
        A_j(s)r^{s+j}
        +
        \sum_{k=0}^\infty
        B_k(s)r^{m+k},\quad r\to 0^+
\ee
where
\[
        A_j(s)
        =
        \frac{c^{(j)}(0)\Gamma(m-s-j)}
        {j!\Gamma(s+1)}
\quad\text{and}\quad
        B_k(s)
        =
        \frac{(-1)^k}
        {k!\,\Gamma(s+1)}
        \widehat c(s-m-k).
\]
The coefficients $A_j$ and $B_k$ are
holomorphic on $\{s\in\cH:s\notin \Z\}$.

We first consider the case where $s$ lies in $\ccH$. The poles
\[
        z=-j
        \quad\text{and}\quad
        z=s-m-k
\]
coincide at $s=0$ precisely when
$
        j=m+k.
$

For $0\leq j<m$, the function $A_j(s)$ is
holomorphic and bounded on $\ccH$.

We next consider the colliding terms given by $j=m+k,\, k\in \Z_{\geq0}$. By Stirling's formula
and \eqref{eq: expression of hat c}, there exist bounded holomorphic functions
$a_k,b_k:\ccH\to\C$ such that
\[
        A_{m+k}(s)
        =
        -\frac{\gamma_k}{s}+a_k(s),
        \quad
        B_k(s)
        =
        \frac{\gamma_k}{s}+b_k(s),
\]
where
\[
        \gamma_k
        =
        \frac{(-1)^k c^{(m+k)}(0)}
        {k!(m+k)!}.
\]
Thus, we obtain an asymptotic expansion of the form
\be\label{eq: asym of J s r0}
        J(r,s)
        \sim
        \ln(r,s)
        \sum_{q=0}^{+\infty}\gamma_q r^q
        +
        \sum_{q=0}^{+\infty}
        \bigl(v_q(s)r^q+u_q(s)r^{s+q}\bigr),
\ee
where
\be\label{eq: uniformly bounded of u v 0 q}
        \text{$u_q,v_q$ are holomorphic and bounded on $\ccH$.}
\ee

By the Borel lemma and \eqref{eq: uniformly bounded of u v 0 q}, there exist
$\tilde\a$ and $\tilde\b$, smooth and compactly supported in $r$ on
$[0,1)$ and holomorphic in $s$ on $\ccH$, and
$\tilde{\gamma}\in C^\infty_c([0,1))$, such that
\[
        \tilde\gamma(r)
        \sim
        \sum_{q=0}^\infty \gamma_q r^q,
        \quad
        \tilde\a(r,s)
        \sim
        \sum_{q=0}^\infty u_q(s)r^q,
        \quad
        \tilde\b(r,s)
        \sim
        \sum_{q=0}^\infty v_q(s)r^q .
\]

For any $r>0$, $s\in\ccH$, and $q\in\Z_{\geq0}$, we have
\[
        \frac{\p}{\p \bar{s}}
        \frac{\p^q}{\p r^q}J(r,s)
        =
        \frac{\p^q}{\p r^q}
        \frac{\p}{\p \bar{s}}J(r,s)
        =
        0.
\]

By \eqref{eq: computation of J r s}, for any $r>0$, $s\in\ccH$,
$p=0$ or $1$, and $q\in\Z_{\geq0}$, we have
\[
        \frac{\p^p}{\p s^p}
        \frac{\p^q}{\p r^q}J(r,s)
        =
        \frac{\p^q}{\p r^q}
        \frac{\p^p}{\p s^p}J(r,s)
        =
        \frac{1}{2\pi\sqrt{-1}}
        \int_{L_d}
        \frac{\p^p}{\p s^p}
        \Bigg[
        \prod_{\ell=0}^{q-1}
        \big(s-z-\ell\big)
        h_s(z)r^{s-z-q}
        \Bigg]\,dz.
\]
Since the residue of
$
        \prod_{\ell=0}^{q-1}
        \big(s-z-\ell\big)h_s(z)
$
is determined by the residue of $h_s(z)$, using the estimate
\eqref{eq: estimate of h s} and repeating the argument above, we see that the
termwise partial derivative of the right-hand side of
\eqref{eq: asym of J s r0} gives the asymptotic behavior of
\[
        \frac{\p^p}{\p s^p}
        \frac{\p^q}{\p r^q}J(r,s)
        =
        \frac{\p^q}{\p r^q}
        \frac{\p^p}{\p s^p}J(r,s)
\]
near $r=0$. As a result,
\[
        \gamma(r,s)
        :=
        J(r,s)
        -
        \ln(r,s)\tilde{\gamma}(r)
        -
        \tilde{\a}(r,s)
        -
        r^s\tilde{\b}(r,s)
\]
is smooth in $r$ on $[0,+\infty)$ (in fact, flat at $r=0$) and holomorphic in
$s$ on $\ccH$.

Setting
\[
        a=\tilde{\gamma},\quad b=\tilde{\b},\quad c=\tilde{\a}+\gamma,
\]
finishes the proof of item \eqref{lem 45 item 2}. The proof of item \eqref{lem45 item1} is similar.
\end{proof}

Now we have:
\begin{prop}\label{thm41}
Let $U\subset \C^n$ be open and $C\in S_{\sm}(U\times U)$. Let $\eta\in C_C^{\infty}([0,1))$. Then there
exist polynomials
\[
        \alpha_k(\tilde{\bfy},d\tilde{\bfy}),
    \tilde\alpha_k(\tilde{\bfy},d\tilde{\bfy}),\tilde{\tilde\alpha}_k(\tilde{\bfy},d\tilde{\bfy}),\quad
        k=0,\ldots,l,
\]
and smooth differential forms
\[
        \beta_k,\tilde\beta_k, \tilde{\tilde\beta}_k
        \in \Gamma(\widetilde E\boxtimes\widetilde E),
        \quad
        k=0,\ldots,l,
\]
depending holomorphically on $s\in\ccH$, such that, for $s\in\ccH$,
\begin{equation}\label{thm41eq0}
\begin{aligned}
        &\frac{1}{\Gamma(s+1)}
        \int_0^{+\infty}
        t^s\eta(t)e^{-\frac{\rho^2}{2t}}\,C
                      =
        \sum_{k=0}^{l}
        \alpha_k(\tilde{\bfy},d\tilde{\bfy})\,\beta_k
+\tilde\alpha_k(\tilde{\bfy},d\tilde{\bfy})\,\tilde\beta_k
        \ln(\rho^2,s)+\tilde{\tilde\alpha}_k(\tilde{\bfy},d\tilde{\bfy})\,\tilde{\tilde\beta}_k\rho^{2s},
\end{aligned}
\end{equation}
where
\[
        \tilde{\bfy}
        =
        (\tilde y_1,\ldots,\tilde y_n),
        \quad
        \tilde y_i
        =
        \frac{2(\bar z_i-\bar w_i)}
        {\rho^2(\bfz,\bar{\bfz},\bfw,\bar{\bfw})}.
\]
Here the integral is taken over the $t$-variable. More precisely, if
\[
        w=\alpha+\beta dt,
        \quad
        \alpha,\beta\in
        C^\infty\big((0,+\infty);
        \Gamma(\widetilde E\boxtimes\widetilde E)\big),
\]
then
\[
        \int_0^{+\infty} w
        :=
        \int_0^{+\infty}\beta\, dt .
\]
An analogous expression also holds on any disk of radius $1/2$ centered at a
positive integer $j$ with
$\ln(\rho^2,s)$ in \eqref{thm41eq0} replaced respectively by
$\ln(\rho^2,s-j)$.

\end{prop}

\begin{proof}
 It suffices to consider the case where
\begin{equation}\label{thm41eq1}
        t^s\eta C_1
        =
        a(\bfy,d\bfy)t^s b_1
\quad\text{or}\quad
        t^s\eta C_2
        =
        a(\bfy,d\bfy)t^s b_2\,dt,
\end{equation}
where $a$ is a monomial of degree $m\geq0$, and
\[
        b_1,b_2\in
        C_c^\infty\big([0,+\infty),
        \Gamma(\widetilde E\boxtimes\widetilde E)\big).
\]
Put
$
        u=\frac{\rho^2}{2t}.
$
We have,  for some polynomials $a_1$ and $a_2$,
\[
        a(\bfy,d\bfy)
        =
        a\big(u\tilde{\bfy},d(u\tilde{\bfy})\big)
        =
        a_1(\tilde{\bfy},d\tilde{\bfy})u^m
        +
        a_2(\tilde{\bfy},d\tilde{\bfy})u^{m-1}\,\Big(\frac{\rho d\rho-u dt}{t}\Big),
        \quad m\geq 0 .
\]
Thus,
\[\ba
        &\quad a(\bfy,d\bfy)dt
        =
        a_1(\tilde{\bfy},d\tilde{\bfy})u^{m}dt
        +
        a_2(\tilde{\bfy},d\tilde{\bfy})(\rho\, d\rho)\,u^{m-1}\,\frac{dt}{t},
        \quad m\geq 0 .
\ea\]
Thus, by \Cref{lem: integrate good part} 
\[\int a(\bfy, d\bfy)t^s b_1=a_2(\tilde{\bfy},d\tilde{\bfy})\rho^{2m}2^{-m}\int_0^\infty e^{-\frac{\rho^2}{2t}} t^{s-m-1}b_1 dt\]
and
\[\ba\int a(\bfy, d\bfy)t^s b_2 dt&=a_1(\tilde{\bfy},d\tilde{\bfy})\rho^{2m}2^{-m}\int_0^\infty e^{-\frac{\rho^2}{2t}} t^{s-m}b_2 dt,\\
&+a_2(\tilde{\bfy},d\tilde{\bfy})(\rho d\rho)\rho^{2m-2}2^{-m+1}\int_0^\infty e^{-\frac{\rho^2}{2t}} t^{s-m-2}b_2 dt,\\
&=a_1(\tilde{\bfy},d\tilde{\bfy})\rho^{2m}2^{-m}\int_0^\infty e^{-\frac{\rho^2}{2t}} t^{s-m-1}(tb_2) dt,\\
&+a_2(\tilde{\bfy},d\tilde{\bfy})(\rho d\rho)\rho^{2(m-1)}2^{-m+1}\int_0^\infty e^{-\frac{\rho^2}{2t}} t^{s-(m-1)-1}b_2 dt,
\ea\]
are of the form \eqref{thm41eq0}.
(Also note that, when $m=0$, we have $a_2=0.$)

\end{proof}

    \def\Bl{\mathrm{Bl}}
	\def\ta{{\tilde{\a}}}
	\def\tb{{\tilde{\b}}}
	\def\ty{{\tilde{y}}}
	
	Let $U\subset\C^n$ be an open set and $\triangle:=\{(\bfz,\bfz):\bfz\in U\}$ denotes the diagonal of $U\times U$. The blow up of $U\times U$ is denoted by $\Bl_{\triangle}(U \times U)$, and we have a canonical map $p:\Bl_{\triangle}(U \times U)\rightarrow U\times U$, such that $p^{-1}(\triangle)$ is the exceptional divisor on $\Bl_{\triangle}(U \times U)$.
	
	The following proposition is the main result of this subsection:
	\begin{prop}\label{lem44}
		Let $C\in S_{\sm}(U\times U)$. Set
		\[
		\mathcal P_C(s)
		:=
		\frac{1}{\Gamma(s+1)}
		\int_0^{+\infty} t^s\eta(t)e^{-\frac{\rho^2}{2t}}C .
		\]
		Then the following holds. 
		
		\begin{enumerate}[(1)]
			\item For $s\in\cH$, $\mathcal P_C(s)$ and all its holomorphic
			derivatives have generalized divisorial type singularities along
			$p^{-1}(\triangle)$, in the sense of \Cref{defn of extended divisor}.

			\item The pullback
			\[
			\triangle^*\mathcal P_C(s)
			:=
			\frac{1}{\Gamma(s+1)}
			\lim_{\epsilon\to0}
			\triangle^*
			\left(
			\int_{\epsilon}^{\infty}
			t^s\eta(t)e^{-\frac{\rho^2}{2t}}C
			\right)
			\]
			is a smooth bundle-valued differential form on $U$, depending
			holomorphically on $s\in\cH$. 
		\end{enumerate}
		Here
		``holomorphic derivatives'' means derivatives with respect to the $U\times U$ component, not with respect to the
		zeta parameter $\bfs$.
	\end{prop}
	
	\begin{proof}
		The proof is the zeta-regularized analogue of
		\cite[Proposition 3.17]{WYFeynman}. We work only on $\ccH$; the proofs of the corresponding statements for the balls of radius $1/2$ centered at positive integers are similar. By \Cref{thm41}, for $s\in\ccH$ we have
		a local expansion of the form
		\[
		\mathcal P_C(s)
		=
		 \sum_{k=0}^{l}
        \alpha_k(\tilde{\bfy},d\tilde{\bfy})\,\beta_k
+\tilde\alpha_k(\tilde{\bfy},d\tilde{\bfy})\,\tilde\beta_k
        \ln(\rho^2,s)+\tilde{\tilde\alpha}_k(\tilde{\bfy},d\tilde{\bfy})\,\tilde{\tilde\beta}_k\rho^{2s}.
		\]

		It remains only to check that $\alpha_k,\tilde{\alpha}_k$ appearing in this expansion
		have the asserted singularity type after pulling back to
		$\Bl_{\triangle}(U\times U)$. This is local on the blow-up. Hence, as in
		\cite[Proposition 3.17]{WYFeynman}, we may work near $(0,0)\in\C^n\times\C^n$.
		On the chart of $\Bl_{\triangle}(\C^n\times\C^n)$ where
		$z_1-w_1$ is the exceptional coordinate, we use coordinates
		\[
		z_1-w_1,\quad
		\lambda_2,\ldots,\lambda_n,\quad
		w_1,\ldots,w_n,
		\quad
		z_i-w_i=(z_1-w_1)\lambda_i,\quad i\ge2.
		\]
		By the proof of \cite[Lemma 2.25]{WYFeynman}, after shrinking the coordinate
		neighborhood if necessary,
		\[
		\rho^2
		=
		|z_1-w_1|^2\,F
		\]
		on this chart, where $F$ is a positive smooth function. Therefore
		\[
		\rho^{2s}
		=
		|z_1-w_1|^{2s}F^s,
		\]
		with $F^s$ smooth and holomorphic in $s$. Moreover,
		\[
		\ln(\rho^2,s)
		=
		F^s\ln(|z_1-w_1|^2,s)+\ln(F,s).
		\]
		The second term is smooth, while the first one is a generalized logarithmic
		factor.
		
		On the same chart, each component of $\tilde{\bfy}$ has the form
		\[
		\tilde y_i
		=
		\frac{\bar z_i-\bar w_i}{\rho^2}
		=
		\frac{1}{z_1-w_1}\,F_i,
		\]
		where $F_i$ is smooth. Hence $\tilde y_i$ has divisorial type singularities
		along $p^{-1}(\triangle)$. The same is true for $d\tilde y_i$, since
		\[
		d\tilde y_i
		=
		-\frac{d(z_1-w_1)}{(z_1-w_1)^2}F_i
		+
		\frac{1}{z_1-w_1}dF_i .
		\]
		It follows that every term in the above expansion has generalized divisorial
		type singularities along $p^{-1}(\triangle)$. 
		
		For holomorphic derivatives, we use the same observation as in
		\cite[Proposition 3.17]{WYFeynman}: if $C\in S_{\sm}(U\times U)$, then
		\[
		e^{\frac{\rho^2}{2t}}
		\partial_{z_i}
		\big(e^{-\frac{\rho^2}{2t}}C\big),
		\quad
		e^{\frac{\rho^2}{2t}}
		\partial_{w_i}
		\big(e^{-\frac{\rho^2}{2t}}C\big)
		\]
		have again regular expressions. 
		
		Finally, on the diagonal we have $\rho=0$, and the pullback of every factor
		$(\bar z_i-\bar w_i)/t$ and $d((\bar z_i-\bar w_i)/t)$ is zero. Thus
		$\triangle^*C$ is a smooth $t$-dependent form extending to $t=0$. Since
		$\Re(s)\ge0$, the integral
		\[
		\frac{1}{\Gamma(s+1)}
		\int_0^{+\infty} t^s\eta(t)\triangle^*C
		\]
		converges and depends holomorphically on $s$. The same argument applies after
		taking holomorphic derivatives. This proves the proposition.
	\end{proof}

	\subsection{Local generalized Cauchy principal values}
	
	We will prove 
	\begin{thm}\label{thm-a2}
		Given the following data:
		\begin{enumerate}[(a)]
			\item   A compactly supported $(m+n,m+n)$ differential form $\alpha$ on $\mathbb{C}^{m+n}$ with generalized divisorial type singularities  along $z_1 \cdots z_m=0$ in the sense of \Cref{defn of extended divisor} above.
			\item A holomorphic function $z^{\bfj}$, where $\bfj \in\mathbb{Z}_{>0}^m$.
			\item A nowhere vanishing smooth function $f$ on a neighborhood of support of $\a$.
			
		\end{enumerate}
		Then the following limit exists for any $\bfs\in \cH^\ell$:
		\be\label{thm-a-eq}\Psi(\bfs):=\lim_{\delta\to0}\int_{|f\bfz^\bfj|>\delta}\a .\ee
		Moreover, the limit is independent of the choice of $\bfj$ and smooth function $f$.  We will also denote the limit as
		\[\dashint_{\C^{m+n}}\a\]
		Lastly, $\Psi$ can be extended to be a holomorphic function  on $\cH^\ell$.
	\end{thm}

   For simplicity, from now on we work only in the case
\[
        m=\ell
        \quad\text{and}\quad
        \bfL=\bfL'=\id .
\]
The general case follows by essentially the same argument.
    We also work only on $\ccH^m$; the proofs of the corresponding statements for the balls of radius $1/2$ centered at points in $\Z_{\geq0}^m\setminus\{0\}$ are similar.

	Let $$B:=\{\bfz\in \C^{m+n}:\sup_{1\leq i\leq m+n}|z_i|\leq1\}.$$
	For each $\bfj\in \Z_{>0}^m$, let $$S^\bfj_\delta:=\big\{\bfz\in B:|\bfz^\bfj|=\delta\big\}\quad\text{ and } \quad B_\delta^\bfj:=\big\{\bfz\in B:|\bfz^\bfj|\in(\delta,1)\big\}.$$
	\def\hi{{\hat{i}}}

	Similar to \cite[Lemma B.4]{WYFeynman}(or \cite[Lemma 6.4]{herrera1971residues}), we have
	\begin{lem}\label{lema3}
		Let $\phi\in L^{+\infty}(\C^{m+n})$, then for any $\bfi,\bfk,\bfl\in \Z_{\geq0}^m,\bfj\in\Z_{>0}^m$, any $\bfs\in \ccH^m$:
		\be\ba\label{lema3eq1}
		\lim_{\delta\to0}\int_{S^\bfj_\delta} \phi(\bfz,\bar\bfz)\ln(\bfz,\bfk,\bfs)r(\bfz,\bfl,\bfs) &dz_i\wedge d\bfz_{\hi}\wedge d\bar\bfz_{\hi}=0,\\
		\lim_{\delta\to0}\int_{S^\bfj_\delta} \phi(\bfz,\bar\bfz) \ln(\bfz,\bfk,\bfs)r(\bfz,\bfl,\bfs)&d\bar{z}_i\wedge d\bfz_{\hi}\wedge d\bar\bfz_{\hi}=0.
		\ea\ee
		Here $d\bfz_{\hi}:=dz_1\wedge\cdots \wedge dz_{i-1}\wedge dz_{i+1}\cdots \wedge dz_{m+n}$ and $d\bar{\bfz}_{\hi}:=\overline{d\bfz_{\hi}},\, i\in\{1,2,\cdots,m\}.$ 
	\end{lem} 
	\begin{proof}
		We may as well assume that \be\label{eq: phi support in B}
        \text{$\phi$ is supported on $B.$}
        \ee
		For $\mathbf{j} = (j_1, \dots, j_m)$, set $|\mathbf{j}| := j_1 + \dots + j_m$.

		Note that if $|z|\le 1$ and $\Re(s)\ge 0$, then, by the mean value theorem
		for convex domains (viewing $|z|^{2s}$ as a function of $s$),
		\be\label{bounded by log}
		\Big|\frac{|z|^{2s}-1}{s}\Big|
		\le
		\big|\ln(|z|^2)\big|.
		\ee

Consider the function
\[
        \phi_{\bfs}(\bfz,\bar\bfz)
        :=
        \phi(\bfz,\bar\bfz)r(\bfz,\bfl,\bfs)\frac{\ln(\bfz,\bfk,\bfs)}{\ln(\bfz,\bfk)}
\]
Then 
      \[\phi(\bfz,\bar\bfz)\bfz^{-\bfi}\ln(\bfz,\bfk,\bfs)r(\bfz,\bfl,\bfs)=\phi_\bfs(\bfz,\bar\bfz)\bfz^{-\bfi}\ln(\bfz,\bfk)\]

        By \eqref{eq: phi support in B} and \eqref{bounded by log}, we have
        \[\phi_{\bfs}(\bfz,\bar\bfz)\in L^\infty(\C^{m+n}).\] The lemma then follows from \cite[Lemma B.4]{WYFeynman}.
		
	\end{proof}
	The argument in \cite[Lemma~B.5]{WYFeynman} applies verbatim, with
$\ln(\bfz,\bfk)$ replaced by $\ln(\bfz,\bfk,\bfs)$, giving the lemma
below. The main idea is to use symmetry to prove the vanishing of the relevant
integrals. A key observation is that $\ln(\bfz,\bfk,\bfs)$ and $r(\bfz,\bfl,\bfs)$ are
invariant under the $S^1$-action
$
        (e^{i\theta},\bfz)\mapsto e^{i\theta}\bfz .
$
	
	\begin{lem}\label{lema4}
		Let $g \in C^{\infty}\big(\C^{m+n}\big)$ be independent of $z_l,\, l\in\{1,2.\cdots,m\}$, and let
		$\bfi,\bfk,\bfl \in \Z_{\geq0}^m$, $\bfj\in\Z_{>0}^m$, $\bfs\in \ccH^m$. Then,
		for all $a,b \in \Z_{\geq0}$ such that $a+b<i_l$, for each $\delta>0$,
		the following holds:
		\begin{enumerate}[(1)]
			\item
			\[
			\int_{S_\delta^\bfj}
			z_l^a \bar{z}_l^b \bfz^{-\bfi}\ln(\bfz,\bfk,\bfs)r(\bfz,\bfl,\bfs)\, g\,
			d\bar{z}_l \wedge d\bfz_{\hat{l}} \wedge d\bar{\bfz}_{\hat{l}}
			=0;
			\]
			\item
			if $b\geq1$, then
			\[
			\int_{S_\delta^\bfj}
			z_l^a \bar{z}_l^b \bfz^{-\bfi}\ln(\bfz,\bfk,\bfs)r(\bfz,\bfl,\bfs)\, g\,
			dz_l \wedge d\bfz_{\hat{l}} \wedge d\bar{\bfz}_{\hat{l}}
			=0.
			\]
		\end{enumerate}
			Here $d\bfz_{\hat{l}}:=dz_1\wedge\cdots \wedge dz_{l-1}\wedge dz_{l+1}\cdots \wedge dz_{m+n}$ and $d\bar{\bfz}_{\hat{l}}:=\overline{d\bfz_{\hat{l}}}.$ 
	\end{lem}

	We have
	\begin{lem}[Lemma B.6 in \cite{WYFeynman}]\label{lem-a5}
		Let $k\in C^{\infty}(\C^{m+n})$, and $\bfi\in\Z_{\geq0}^m$. There exists a decomposition
		$$
		\begin{aligned}
			& k(\bfz,\bar{\bfz})=\sum_{a+b<i_l} z_{l}^a \bar{z}_{l}^b g_{a, b}^l(\bfz_{\hat{l}}, \bar{\bfz}_{\hat{l}})+K^l(\bfz,\bar{\bfz}), \\
			& K^l(\bfz,\bar{\bfz})=\sum_{\bfk+\bfj=\bfi}\bfz^\bfj \cdot \bar{\bfz}^\bfk \cdot K^l_{\bfj, \bfk}(\bfz,\bar{\bfz})
		\end{aligned}
		$$
		such that $g_{a, b}^l$ and $K_{\bfj, \bfk}^l$ are smooth functions.
	\end{lem}

	Proceeding as in \cite[Proposition 6.5]{herrera1971residues}, using \Cref{lema3}-\Cref{lem-a5}, we have:
	\begin{prop}\label{prop-a5}
		For any $\bfi,\bfk,\bfl \in \Z_{\geq0}^m,\bfj\in\Z_{>0}^m,\bfs\in \ccH^m$ and $g \in C^{\infty}\left(\mathbb{C}^{m+n}\right)$,
		we have
		\be\label{prop-a-eq3}
		\lim _{\delta \rightarrow 0} \int_{S_\delta^\bfj} \bfz^{-\bfi} \ln(\bfz,\bfk,\bfs)r(\bfz,\bfl,\bfs)g d \bar{z}_1 \wedge d \bfz_{\hat{1}} \wedge d \bar{\bfz}_{\hat{1}}=0
		\ee
		and
		\be\label{prop-a-eq4}
		\lim _{\delta \rightarrow 0} \int_{S_\delta^\bfj} \bfz^{-\bfi} \ln(\bfz,\bfk,\bfs)\bar{z}_1 r(\bfz,\bfl,\bfs)g d {z}_1 \wedge d \bfz_{\hat{1}} \wedge d \bar{\bfz}_{\hat{1}}=0.
		\ee

	\end{prop} 
	
	We have
	\begin{cor}\label{cor-a5}
		For any $\bfi,\bfk,\bfl \in \Z_{\geq 0}^m$, $\bfj \in \Z_{>0}^m,\bfs\in\ccH^m$, and $g \in C_C^{\infty}(\C^{m+n})$, we have
		\be\label{cor-a-eq3}
		\lim_{\delta \to 0} 
		\int_{|\bfz^\bfj|=\delta} 
		\bfz^{-\bfi}\,\ln(\bfz,\bfk,\bfs)r(\bfz,\bfl,\bfs)\, g\;
		d\bar z_1 \wedge d\bfz_{\hat 1} \wedge d\bar{\bfz}_{\hat 1}
		=0,
		\ee
		and
		\be\label{cor-a-eq4}
		\lim_{\delta \to 0} 
		\int_{|\bfz^\bfj|=\delta} 
		\bfz^{-\bfi}\,\ln(\bfz,\bfk,\bfs)\,\bar z_1 r(\bfz,\bfl,\bfs)g\;
		d z_1 \wedge d\bfz_{\hat 1} \wedge d\bar{\bfz}_{\hat 1}
		=0.
		\ee
	\end{cor}
	
	\begin{proof}
		We may assume $g \in C_c^{\infty}(B)$.  
		Indeed, if not, there exists $N>1$ such that
		\[
		\mathrm{supp}(g) \subset
		B_N:=\{\bfz\in\C^{m+n}:\sup_{1\le j\le m+n}|z_j|<N\}.
		\]
		In this case, one verifies that \Cref{prop-a5} continues to hold when the
		set $B$ is replaced by $B_N$.
		
		Now assume $g \in C_c^{\infty}(B)$. Then
		\[
		\int_{|\bfz^\bfj|=\delta}
		\bfz^{-\bfi}\,\ln(\bfz,\bfk,\bfs)r(\bfz,\bfl,\bfs)\, g\;
		d\bar z_1 \wedge d\bfz_{\hat 1} \wedge d\bar{\bfz}_{\hat 1}
		=
		\int_{S_\delta}
		\bfz^{-\bfi}\,\ln(\bfz,\bfk,\bfs)r(\bfz,\bfl,\bfs)\, g\;
		d\bar z_1 \wedge d\bfz_{\hat 1} \wedge d\bar{\bfz}_{\hat 1},
		\]
		so \eqref{cor-a-eq3} follows directly.  
		The argument for \eqref{cor-a-eq4} is analogous.
	\end{proof}
	
	Proceeding as in \cite[Proposition~B.9]{WYFeynman}, it follows from
	\Cref{cor-a5} that
	
	\begin{prop}\label{prop-a7}
		Let $f$ be a real nowhere–vanishing smooth function on $\C^{m+n}$, and let
		$\bfj \in \Z_{>0}^m$, $\bfi,\bfk,\bfl \in \Z_{\ge0}^m$, $\bfs\in\ccH^m$.
		Then, for any $g \in C_c^{\infty}(\C^{m+n})$, we have
		\be\label{prop-a6-eq2}
		\lim_{\delta \to 0}
		\int_{|f\bfz^\bfj|=\delta}
		\bfz^{-\bfi}\,\ln(\bfz,\bfk,\bfs)r(\bfz,\bfl,\bfs)\, g\;
		d\bar z_1 \wedge d\bfz_{\hat 1} \wedge d\bar{\bfz}_{\hat 1}
		=0 .
		\ee
	\end{prop}
	
	\def\bfl{{\mathbf{l}}}
	It is straightforward to check that \Cref{thm-a2}, in the case where
$\bfL=\id$ and $\bfs\in\ccH^m$, follows directly from the result below:
	\begin{prop}\label{prop-a8}
		Let $\bfi,\bfk,\bfl \in \Z_{\ge0}^m$, $\bfs\in \ccH^m$ and $\bfj \in \Z_{>0}^m$.
		Let $f \in C^{\infty}(\C^{m+n})$ be a real nowhere–vanishing function.
		Then the limit
		\be\label{prop-a8-eq0}
		\Psi(\bfs):=\lim_{\delta \to 0}
		\int_{|f\bfz^\bfj|>\delta}
		\bfz^{-\bfi}\,\ln(\bfz,\bfk,\bfs)r(\bfz,\bfl,\bfs)\,\psi g\;
		d\bfz \wedge d\bar{\bfz}
		\ee
		exists and is independent of the choice of $f$ and $\bfj$, for every
		$g \in C_c^{\infty}(\C^{m+n})$ and every
		$\psi\in C^{\infty}(\cH^m\times \C^{m+n})$ which is holomorphic in  $\ccH^m$.
		
		Moreover, $\Psi$ can be extended to a holomorphic function on $\ccH^m$.
	\end{prop}
	
	\begin{proof}
  We call the integrand
$
        \bfz^{-\bfi}\,\ln(\bfz,\bfk,\bfs)r(\bfz,\bfl,\bfs)\,\psi g
$
a $(\bfi,\bfk)$-type (generalized) divisor-type singularity.

		If $0\le i_l\le1$ for $1\le l\le m$, then for $\bfs\in\ccH^m$
		\[
		\bfz^{-\bfi}r(\bfz,\bfl,\bfs)\,\ln(\bfz,\bfk,\bfs)\,\psi g
		\]
		is Lebesgue integrable on $\C^{m+n}$, so the limit
		\eqref{prop-a8-eq0} exists and is independent of $f$ and $\bfj$.
		Moreover, the derivatives of the integrand with respect to $\bfs$ are also
		Lebesgue integrable. Thus $\Phi$ is holomorphic in $\ccH^m$.
		
		We introduce a lexicographic order on $\Z_{\ge0}\times\Z_{\ge0}$ by
		\[
		(p_1,p_2)\le(q_1,q_2)
		\quad\Longleftrightarrow\quad
		\big(p_1<q_1\big)
		\ \text{or}\ 
		\big(p_1=q_1 \text{ and } p_2\le q_2\big).
		\]
		
		Assume inductively that the proposition holds for all pairs
		$(|\bfi|,|\bfk|)<(p_1,p_2)$.
		We prove it for $(|\bfi|,|\bfk|)=(p_1,p_2)$.
		
		The case $0\le i_l\le1$ is already handled above, so we may assume
		$i_1\ge2$.
		
		Consider
		\[
		b=
		\frac{1}{1-i_1+l_1s_1}\,
		r(\bfz,\bfl,\bfs)\bfz^{-\bfi} z_1\,
		d\bar z_1\wedge d\bfz_{\hat1}\wedge d\bar{\bfz}_{\hat1}.
		\]
		
		We compute:
		\[
		\begin{aligned}
			&\quad
			\int_{|f\bfz^\bfj|>\delta}
			r(\bfz,\bfl,\bfs)\bfz^{-\bfi}\ln(\bfz,\bfk,\bfs)\,\psi g\;
			d\bfz\wedge d\bar{\bfz}
			\\
			&=
			\int_{|f\bfz^\bfj|>\delta}
			d\Big(\ln(\bfz,\bfk,\bfs)r(\bfz,\bfl,\bfs)\,\psi g\, b\Big)
			-
			\int_{|f\bfz^\bfj|>\delta}
			(\partial_{z_1}\ln(\bfz,\bfk,\bfs))\,r(\bfz,\bfl,\bfs)\psi g\;dz_1\wedge b
			\\&-
			\int_{|f\bfz^\bfj|>\delta}
			\ln(\bfz,\bfk,\bfs)\,\partial_{z_1}(r(\bfz,\bfl,\bfs)\psi g)\;dz_1\wedge b
			\\
			&=
			\int_{|f\bfz^\bfj|=\delta}
			r(\bfz,\bfl,\bfs)\ln(\bfz,\bfk,\bfs)\,\psi g\, b
			-
			\int_{|f\bfz^\bfj|>\delta}
			(\partial_{z_1}\ln(\bfz,\bfk,\bfs))\,r(\bfz,\bfl,\bfs)\psi g\;dz_1\wedge b
			\\&-
			\int_{|f\bfz^\bfj|>\delta}
			\ln(\bfz,\bfk,\bfs)\,\partial_{z_1}(r(\bfz,\bfl,\bfs)\psi g)\;dz_1\wedge b
			\\
			&=: I_1(\delta,\bfs)+I_2(\delta,\bfs)+I_3(\delta,\bfs).
		\end{aligned}
		\]
		By \Cref{prop-a7},
		\[
		\lim_{\delta\to0}I_1(\delta,\bfs)=0.
		\]
        
		For $I_2(\delta,\bfs)$, the integrand is of $(\bfi,\bfk')$-type  with $|\bfk'|\leq|\bfk|-1$
		(since $\p_{z}\big(\frac{|z|^{2s}-1}{s}\big)^k
		=k\big(\frac{|z|^{2s}-1}{s}\big)^{k-1}|z|^{2s}z^{-1}$).
		By the induction hypothesis,
		$\lim_{\delta\to0}I_2(\delta,\bfs)$ exists, is independent of $f$
		and $\bfj$, and is holomorphic in $\bfs\in\ccH^m$.
		
		For $I_3(\delta,\bfs)$, the integrand is of $(\bfi',\bfk)$–type with $|\bfi'|\leq|\bfi|-1$.
		Again by the induction hypothesis,
		$\lim_{\delta\to0}I_3(\delta,\bfs)$ exists, is independent of $f$
		and $\bfj$, and is holomorphic in $\bfs\in\ccH^m$.
		
		Since $\Z_{\ge0}\times\Z_{\ge0}$ is well ordered by the above
		lexicographic order, the claim follows by transfinite induction.
	\end{proof}

	\subsection{Global generalized Cauchy principal values}
	
	We now pass from the local Cauchy principal value theorem to the global form.

	\begin{defn}\label{defn:smooth defining function global}
		Let $X$ be a complex manifold with a simple normal crossing divisor $D$.
		A \emph{smooth defining function} of $D$ is a smooth function
		\[
		h:X\longrightarrow \R_{\ge0}
		\]
		such that $h^{-1}(0)=D$, and for each point $p\in D$, there exists a
		coordinate neighborhood $U$ of $p$, compatible with $D$, a positive
		smooth function $f\in C^{\infty}(U)$, and a holomorphic function $g$ on $U$, such that
		\[
		h|_U=f\,g\bar g,
		\quad
		g^{-1}(0)=D\cap U.
		\]
	\end{defn}
	
	\begin{defn}\label{defn: compatible coordinate change}
		Let $X$ be a complex manifold with a simple normal crossing divisor $D$. A coordinate chart \textbf{compatible with} $D$ is an open subset $U\subset X$ together with a biholomorphic map \[\varphi:U\rightarrow V\subset\mathbb{C}^{m+n},\] where $V$ is an open subset of $\mathbb{C}^{m+n}$, and $m,n$ are non-negative integers, such that $\varphi(D\cap U)=(\{z_1\cdots z_m=0\}\times \mathbb{C}^{n})\cap V$. We say that a property \textbf{holds locally}, if it holds in any coordinate chart compatible with $D$.
	\end{defn}

	\begin{defn}\label{divisorial type singularities}
		 Let $X$ be a closed complex manifold and let $D\subset X$ be a simple
	normal crossing divisor.  We say that a differential form
	$
	\alpha(\bfs)\in \Omega^\bullet(X\setminus D)
	$ parametrized by $\bfs\in\cH^\ell$ for some $\ell\in\Z_{>0}$
	has \emph{global generalized divisorial type singularities along $D$}, if this property holds in
	every coordinate chart compatible with $D$, in the sense of
	\Cref{defn of extended divisor}.
	\end{defn}
    
    \begin{defn}\label{defn:global generalized cpv}
		Let $X$ be a closed complex manifold with a simple normal crossing divisor
		$D$, and let $\alpha(\bfs)$ be a differential form on
		$X\setminus D$ with global generalized divisorial type singularities along
		$D$.  Choose a smooth defining function $h$ of $D$.  We define the
		generalized Cauchy principal value of $\alpha(\bfs)$ by
		\[
		\dashint_X \alpha(\bfs)
		:=
		\lim_{\epsilon\to0}
		\int_{\{h>\epsilon\}}\alpha(\bfs),
		\]
		provided the limit exists.
	\end{defn}
	
	\begin{prop}[Global generalized Cauchy principal value]
		\label{prop:global generalized cpv}
		Let $X$ be a closed complex manifold and let $D\subset X$ be a simple
		normal crossing divisor.  Suppose that $\alpha(\bfs)$ is a top-degree
		differential form on $X\setminus D$ with global generalized divisorial type
		singularities along $D$, depending holomorphically on $\bfs\in\cH^\ell$.
		Then, for any smooth defining function $h$ of $D$, the limit
		\[
		\lim_{\epsilon\to0}
		\int_{\{h>\epsilon\}}\alpha(\bfs)
		\]
		exists and defines a holomorphic function on $\cH^\ell$.
		Moreover, this holomorphic function is independent of the choice of the
		smooth defining function $h$.
	\end{prop}
	
	\begin{proof}
		Choose a finite coordinate cover $\{U_i\}$ of $X$ compatible with the
		simple normal crossing divisor $D$, and choose a partition of unity
		$\{\rho_i\}$ subordinate to this cover. By assumption, each
		$\rho_i\alpha(\bfs)$ has compact support in $U_i$ and admits a local
		generalized divisorial expression of the type considered in \Cref{thm-a2}.
		Therefore
		\[
		\dashint_{U_i}\rho_i\alpha(\bfs)
		\]
		is well-defined and depends holomorphically on $\bfs\in\cH^\ell$. Since the cover
		is finite, the sum
		\[
		\sum_i
		\dashint_{U_i}\rho_i\alpha(\bfs)
		\]
		is again holomorphic on $\cH^\ell$.
		
		The proof that the limit  is independent of $h$, the open cover, and the partition of unity, follows exactly the same argument as in \cite[Definition B.15]{WYFeynman}.
	\end{proof}

	\subsection{Proof of the main theorem}
	
	\begin{proof}[Proof of \Cref{thm:zeta graph integral holomorphic continuation}]\def\cU{{\mathcal{U}}}

		\Cref{lem44} shows that, the kernel
		$P^{s_e}$ has generalized divisorial type singularities along the diagonal.
		Therefore, after
		pulling back to the Fulton-MacPherson compactification
		$\widetilde{\Conf}_{\vG_0}(M)$, the form
		\[
		\check{W}(\vG,\Phi,\bfs),\quad\bfs\in\cH^{|\vG_1|}
		\]
		has generalized divisorial type singularities along
		$
		\widetilde{\Conf}_{\vG_0}(M)\setminus \Conf_{\vG_0}(M).
		$
		 Indeed, this follows from
		\Cref{lem44}, the stability of generalized divisorial type singularities under
		holomorphic pullback, and the fact that the operators $I_v$ associated with the vertices are holomorphic differential operators. This is exactly the same argument as in
		the proof of \cite[Theorem 2.21]{WYFeynman}, with $P$ replaced by $P^{s_e}$.
		
		By  \Cref{prop:global generalized cpv}, the limit
		\[
		\dashint_{\widetilde{\Conf}_{\vG_0}(M)}
		\check{W}(\vG,\Phi,\bfs).
		\]
		is well-defined and depends holomorphically on $\bfs\in\cH^{|\vG_1|}$.

		Finally, suppose that $\Re(s_e)\gg0$ for every $e\in\vG_1$. In this
		region $\check{W}(\vG,\Phi,\bfs)$ is  Lebesgue integrable  over $\Conf_{\vG_0}(M)$. Hence
		\[
		W^{\mathrm{CPV}}(\vG,\Phi,\bfs)
		=
		\int_{\Conf_{\vG_0}(M)}
		\check{W}(\vG,\Phi,\bfs).
		\]
	\end{proof}

	\section{A Fubini-type theorem for graph integrals}
	\label{sec:zeta fubini theorem}
	
	In \cite{WYFeynman}, the Feynman graph integral was defined by the Cauchy
	principal value over the Fulton-MacPherson compactification of the
	configuration space. This may be viewed as first integrating the
	Schwinger-time variables to form the propagators, and then integrating over
	the configuration space.
	
	By contrast, the heat-kernel formulation discussed in \Cref{sec: heat formulation} keeps
	the Schwinger parameters. One first integrates over the
	configuration variables, and obtains a smooth form on the compactified
	Schwinger space.
	
	In this section we compare these two orders of integration by means of
	zeta-regularization, and prove the following Fubini-type theorem.

	\begin{thm}[Fubini-ype theorem for Feynman graph integral]
		\label{thm:zeta fubini theorem}
		For every decorated directed graph $\vG$,  the holomorphic function $W^{\mathrm{ZF}}(\vG,\Phi,\bfs)$, initially defined by
\eqref{eq: Psi ZF} for $\Re(s_e)\gg0$, admits an analytic continuation
to the interior of $\mathbb{H}^{|\vG_1|}$ that extends continuously to
$\mathbb{H}^{|\vG_1|}$. Moreover, 
		\[
		W^{\mathrm{ZF}}(\vG,\Phi,\bfs)=W^{\mathrm{HK}}(\vG,\Phi,\bfs)
		=
		W^{\mathrm{CPV}}(\vG,\Phi,\bfs),
		\quad
		\bfs\in\cH^{|\vG_1|}.
		\]
		In particular, setting $\bfs=0$, we have
		\[
		W^{\mathrm{ZF}}(\vG,\Phi)=W^{\mathrm{HK}}(\vG,\Phi)
		=
		W^{\mathrm{CPV}}(\vG,\Phi).
		\]
	\end{thm}

	\begin{proof}
		$W^{\mathrm{HK}}(\vG,\Phi,\bfs)
		$ and
		$W^{\mathrm{CPV}}(\vG,\Phi,\bfs) $ are holomorphic on $\cH^{|\vG_1|}$: the heat-kernel side by
		\Cref{prop:zeta heat kernel holomorphic}, and the Cauchy principal value side
		by \Cref{thm:zeta graph integral holomorphic continuation}. When
		$\Re(s_e)\gg0$ for all $e\in\vG_1$, all integrals are
		absolutely convergent. Thus Fubini's theorem gives
		\[
		W^{\mathrm{ZF}}(\vG,\Phi,\bfs)=W^{\mathrm{HK}}(\vG,\Phi,\bfs)
		=
		W^{\mathrm{CPV}}(\vG,\Phi,\bfs); \quad\Re(s_e)\gg0, \forall e.
		\]
This in particular shows that $W^{\mathrm{ZF}}(\vG,\Phi,\bfs)$ admits an analytic continuation
to the interior of $\cH^{|\vG_1|}$ that extends continuously to
$\cH^{|\vG_1|}.$ 
        The theorem then follows from the uniqueness of analytic continuation. 
	\end{proof}

	\section{Holomorphicity of Feynman graph integrals.}\label{holomorphic gauge anomaly}
	
	\subsection{Holomorphicity and integrals over boundaries of compactified
		Schwinger spaces.}
	Throughout the section,  the notation $\bar\partial$ is used according to the context. Let $\vG$ be a holomorphic decorated directed graph.
	By \Cref{thm:graph integral current}, $W(\vG,\cdot)$ is a well-defined
	current on $M^{|\vG_0|}$. It therefore makes sense to study its
	holomorphicity in the sense of currents. More precisely, we will study whether
	\be\label{holomorphicity}
	\bar\partial W(\vG,\Phi):=W(\vG,\dbs\Phi)=0,\quad\forall \Phi\in\Omega^\bullet(M^{|\vG_0|}),
	\ee
	where $\dbs$ means that the Koszul signs are dictated.
	
	Unfortunately,  \eqref{holomorphicity}  is not true in general even if the graph is holomorphic decorated.
	The failure of $\left( \ref{holomorphicity} \right)$ is called (gauge)
	anomalies by physicists. We will show these anomalies can be
	computed by integrals over boundaries of compactified Schwinger spaces.
	
	For $0<\epsilon<L\leq+\infty$, define
	\[
	W_\ep^{L}(\vG,\Phi)
	:
	=\int_{(\epsilon,L)^{|\vG_1|}}
	\int_{M^{|\vG_0|}}
	\widetilde W(\vG,\Phi)=\int_{(\epsilon,L)^{|\vG_1|}}
	\widehat W(\vG,\Phi)\]
	and define
	\[
	W_0^{L}(\vG,\Phi)
	:
	=\lim_{\ep\to0}\int_{\widetilde{(\epsilon,L)^{|\vG_1|}}}
	\widehat W(\vG,\Phi)\]
	Then $W_0^{+\infty}(\vG,\Phi)=W(\vG,\Phi).$
	
	\begin{prop}\label{prop: anomaly IBP}
		\label{differential transfer}Given a holomorphic decorated directed graph ${\vec{\Gamma}}$, and $\Phi \in \Omega^{*} (M^{|
			{\vec{\Gamma}}_0 |})$, we have the following:
		
		\begin{enumerate}[(1)]
			\item
			For $0<L<\infty$,\begin{equation}
				(\bar{\partial} W_0^L) ({\vec{\Gamma}}, -) = (- 1)^{| {\vec{\Gamma}}_1 |}
				\int_{\partial \widetilde{[0, L]^{| {\vec{\Gamma}}_1 |}}}
				\widehat W(\vG,-) .
				\label{QME}
			\end{equation}
			\item
			\[ (\bar{\partial} W_0^{+ \infty}) ({\vec{\Gamma}}, -) = \lim_{L\to\infty}(- 1)^{| {\vec{\Gamma}}_1 |} \int_{\partial \widetilde{[0, L]^{|
						{\vec{\Gamma}}_1 |}}}  \widehat W(\vG,-). \]
		\end{enumerate}
		
	\end{prop}
	
	\begin{proof}\def\id{{\mathrm{id}}}

		We prove the second inequality only.  
		On the open interval $(0,+\infty)$, the Schwinger propagator is smooth and satisfies
		\be\label{eq:derivative dt Pt}
		\Big(d^{(0,\infty)}+\bar\partial\otimes\id+\id\otimes\bar\partial\Big)P_t=0.
		\ee
		Therefore, by Stokes' formula and \eqref{eq:derivative dt Pt}, we have
		\be\label{eq: 53}
		\widehat W(\vG,\dbs\Phi)
		=
		d^{(0,\infty)^{|\vG_1|}}\widehat W(\vG,\Phi).
		\ee
		Here $d^{(0,\infty)^{|\vG_1|}}$ denotes the de Rham differential on
		$(0,\infty)^{|\vG_1|}$.
		
		By \Cref{smooth extension to schwinger3}, both
		$\widehat W(\vG,\Phi)$ and $\widehat W(\vG,\dbs\Phi)$ extend smoothly to the
		compactified Schwinger space. Thus the identity \eqref{eq: 53} extends to
		$\resch$:
		\be\label{eq: 54}
		\widehat W(\vG,\dbs\Phi)
		=
		d^{\resch}\widehat W(\vG,\Phi).
		\ee
		Here $d^{\resch}$ denotes the de Rham differential on $\resch$.
		
		Applying Stokes' theorem and using the large-time exponential decay of the
		Schwinger graph integrand (\Cref{prop:W-hat-large-time-decay}), we obtain
		\[
		\begin{aligned}
			&\quad(\bar\partial W)(\vG,\Phi)
			=(-1)^{|\vG_1|}
			\int_{\resch}
			\widehat W(\vG,\bp^{\mathrm{sign}}\Phi)=\lim_{L\to\infty}(-1)^{|\vG_1|}
			\int_{\widetilde{[0,L]^{|\vG_1|}}}
			\widehat W(\vG,\bp^{\mathrm{sign}}\Phi)\\&=\lim_{L\to\infty}
			(-1)^{|\vG_1|}
			\int_{\widetilde{[0,L]^{|\vG_1|}}}
			d^{\resch}\widehat W(\vG,\Phi)=\lim_{L\to\infty}
			(-1)^{|\vG_1|}
			\int_{\partial\widetilde{[0,L]^{|\vG_1|}}}
			\widehat W(\vG,\Phi).
		\end{aligned}
		\]
		
	\end{proof}
	
	Let's describe the boundaries of ompactified Schwinger space $\resch$ in detail.
	
	Given ${\vec{\Gamma}}'$ is a subgraph of ${\vec{\Gamma}}$, then both ${\vec{\Gamma}}'_1 = \{ e'_1,
	\ldots e'_{| {\vec{\Gamma}}'_1 |} \}$ and ${\vec{\Gamma}}_1 = \{ e_1, \ldots, e_{| {\vec{\Gamma}}_1 |}
	\}$ are ordered sets. Assume ${\vec{\Gamma}}_1 \backslash {\vec{\Gamma}}_1' = \big\{ e_{i_1},
	\ldots, e_{i_{| {\vec{\Gamma}}_1 | - | {\vec{\Gamma}}_1' |}} \big\}$, then there exists a
	unique permutation \[
	\big\{ e_{i_1}, \ldots, e_{i_{| {\vec{\Gamma}}_1 | - | {\vec{\Gamma}}_1'
			|}}, e'_1, \ldots e'_{| {\vec{\Gamma}}'_1 |} \big\} \rightarrow \{ e_1, \ldots, e_{|
		{\vec{\Gamma}}_1 |} \}.\]
	We denote the sign of this permutation by $$(- 1)^{\sigma
		({\vec{\Gamma}}', {\vec{\Gamma}} / {\vec{\Gamma}}')},$$ where $${\vec{\Gamma}} / {\vec{\Gamma}}'$$  is the graph obtained by contracting $\vG'$ to one vertex.
	
	The boundaries $\partial \widetilde{[0,L]^{|\vG_1|}}$ has the following
	decomposition:
	\[ \partial \widetilde{[0, L]^{| {\vec{\Gamma}}_1 |}} = \Big( - \partial_0
	\widetilde{[0, L]^{| {\vec{\Gamma}}_1 |}} \Big) \cup \partial_L \widetilde{[0,
		L]^{| {\vec{\Gamma}}_1 |}}, \]
	where $\partial_0 \widetilde{[0, +\infty)^{| {\vec{\Gamma}}_1 |}}$ (resp. $\partial_L
	\widetilde{[0, L]^{| {\vec{\Gamma}}_1 |}}$) describe the boundary components near the
	origin (resp. away from the origin).
To describe the boundary in greater detail, we introduce the following notation.

\def\forein{1}
\if\forein0
For $e\in\vG_1$, let $\mathfrak{s}_e=(\vG_1,e)$, and set
\[
C_{\vG}^{\mathbb R}
:=
\bigcup_{e\in\vG_1}
\Big(\prod_{e'\in\vG_1}\sq_{e'}\Big)
\circ\kappa_{\mathfrak{s}_e}
\big(C_{\mathfrak{s}_e}\big)
\subset\resch.
\]
Recall that $\mathfrak{s}_e$ is a marked nested sequence of length $1$,
as introduced in \Cref{defn: nested sequence}, and that
$\kappa_{\mathfrak{s}_e}$ is the lifted map described in
\Cref{prop: kappa s}.\fi

Using the coordinates introduced in \Cref{chart on real} or
\cite[\S~3.1]{wang2025feynman}, we find that the proper smooth function
\[
        (0,\infty)^{|\vG_1|}
        \longrightarrow
        (0,\infty),
        \qquad
        (t_e)_{e\in\vG_1}
        \longmapsto
        \big(\sum_{e\in\vG_1}t_e^2\big)^{1/2}
\]
admits a smooth extension to $\resch$, which is also proper.
We denote this extended function by $\varrho_{\vG_1}$. Moreover, $\varrho^{-1}_{\vG_1}(0)$ is a manifold with corner of dimension
$|\vG_1|-1$.

    Then we have decomposition
	\[ \left\{\begin{array}{l}
		\partial_0 \widetilde{[0, L]^{| {\vec{\Gamma}}_1 |}} = \bigcup_{{\vec{\Gamma}}' \subseteq
			{\vec{\Gamma}}} (- 1)^{\sigma ({\vec{\Gamma}}', {\vec{\Gamma}} / {\vec{\Gamma}}')} \varrho_{\vG_1'}^{-1}(0) \times \widetilde{[0, + L]^{|
				{\vec{\Gamma}}_1 \backslash {\vec{\Gamma}}_1' |}}\\
		\partial_L \widetilde{[0, L]^{| {\vec{\Gamma}}_1 |}} = \bigcup_{e \in {\vec{\Gamma}}_1}
		(- 1)^{| e |} \{ L \} \times  \widetilde{[0, L]^{| \vG_1 \setminus \{e\} |}}
	\end{array}\right. . \]
	{By the decomposition above and 
\Cref{prop:W-hat-large-time-decay}, together with the decomposition
\[
        P_t=-dt\wedge(\bar\partial^*\otimes1)H_t+H_t
\]
and the convergence $H_t\to\mathcal H$ as $t\to\infty$, we obtain the
following proposition.

\begin{prop}\label{prop:large-time-boundary}
For each $e\in\vG_1$, there exist finitely many differential forms
$H_{e,k}$, obtained from holomorphic derivatives of the harmonic projection
kernel $\mathcal H$, such that
\be\label{eq: Large L has no contribution}
        \lim_{L\to\infty}
        \int_{\partial_L\widetilde{[0,L]^{|\vG_1|}}}
        \widehat W(\vG,\Phi)
        =
        \sum_{e\in\vG_1}
        \sum_k
        W\bigl(\vG\setminus\{e\},\,H_{e,k}\wedge\Phi\bigr).
\ee
Here $\vG\setminus\{e\}$ denotes the graph obtained from $\vG$ by removing
the edge $e$ while keeping all vertices. For each pair $(e,k)$, the graph
$\vG\setminus\{e\}$ is decorated with possibly different holomorphic
Lagrangian densities.
\end{prop}
	
	\begin{rem}
		If the reader is familar with Batalin-Vilkovisky
		formalism{ \cite{batalin1981gauge,batalin1983generalized,schwarz1993geometry,costellorenormalization,alexandrov1997geometry}},
		it should obvious that   \eqref{QME}  can be used to
		compute the failure of quantum master equation:
		\[ \Big( Q I + \frac{1}{2} \{ I, I \} + \hbar \Delta_{\tmop{BV}} I \Big)
		e^{\frac{1}{\hbar} I} = O e^{\frac{1}{\hbar} I} . \]
		More specifically, the term $(\bar{\partial} W_0^L) ({\vec{\Gamma}}, -)$
		corresponds to $(Q I) e^{\frac{1}{\hbar} I}$, the term
		\[ \int_{\partial_L \widetilde{[0, L]^{| {\vec{\Gamma}}_1 |}}}
		\int_{M^{| {\vec{\Gamma}}_0 |}} \widetilde{W} ({\vec{\Gamma}}, -) \]
		corresponds to $\left( \frac{1}{2} \{ I, I \} + \hbar \Delta_{\tmop{BV}} I
		\right) e^{\frac{1}{\hbar} I}$, and the term
		\[ \int_{\partial_0 \widetilde{[0, L]^{| {\vec{\Gamma}}_1 |}}}
		\int_{M^{| {\vec{\Gamma}}_0 |}} \widetilde{W} ({\vec{\Gamma}}, -) \]
		corresponds to the anomaly $O e^{\frac{1}{\hbar} I}$.
	\end{rem}

	\subsection{Anomaly integrals}\label{Laman graph integral}

	By \eqref{eq: Large L has no contribution}, it remains to study the integrals
	over
	$
	\partial_0 \widetilde{[0,L]^{|\vG_1|}}
	$
	in detail. Given a subgraph ${\vec{\Gamma}}' \subseteq {\vec{\Gamma}}$,
	we consider
	\[ \int_{\varrho_{\vG_1'}^{-1}(0) \times
		\widetilde{[0, + L]^{| {\vec{\Gamma}}_1 \backslash {\vec{\Gamma}}_1' |}}}
	\int_{M^{| {\vec{\Gamma}}_0 |}} \widetilde{W} ({\vec{\Gamma}}, -). \]
	 We first
	concentrate on the case ${\vec{\Gamma}}' = {\vec{\Gamma}}$.  In the
	following, we set:
	\[ O_{{\vec{\Gamma}}} (\Phi) := \int_{\varrho_{\vG_1}^{-1}(0)} \int_{M^{|
			{\vec{\Gamma}}_0 |}} \widetilde{W} ({\vec{\Gamma}}, \Phi), \text{\quad for }   \Phi
	\in \Omega^{*}_c (M^{| {\vec{\Gamma}}_0 |}) . \]
If we take
\[
        E=\Lambda^\bullet T^*_{1,0}M .
\]
Then
\[
        \Gamma(\widetilde E)
        =
        \Omega^\bullet(M).
\]
For every $k\ge1$, let $ \pr_i:M^k\to M,  i=1,\ldots,k, $ be the projection onto the $i$-th factor. There is a natural embedding
\[
        \Omega^\bullet(M)^{\otimes k}
        \hookrightarrow
        \Omega^\bullet(M^k),
        \quad
        \alpha_1\otimes\cdots\otimes\alpha_k
        \longmapsto
        \pr_1^*\alpha_1\wedge\cdots\wedge \pr_k^*\alpha_k .
\]
Therefore a Dolbeault holomorphic Lagrangian density of degree $k$ canonically
extends to an operator
\[\Omega^\bullet(M^{k})\longrightarrow \Gamma(K_M\otimes \Lambda^\bullet T^*_{0,1}M)\hookrightarrow
        \Omega^\bullet(M).\]   More generally, if $l\ge k$, then a degree-$k$ Dolbeault holomorphic Lagrangian density canonically extends to an operator \[ \Omega^\bullet(M^l) \longrightarrow \Omega^\bullet(M^{l-k+1}) \] once one specifies the $k$ components of $M^l$ on which the density acts. We shall use these extensions in the following theorems.

	\begin{thm}\label{thm:anomaly operator kahler}
		
		Let ${\vec{\Gamma}}$ be a decorated directed graph (not necessary connected), then
		there exists a degree-$|\vG_0|$ Dolbeault holomorphic Lagrangian density 
		\[
		D_{\vG}:\Omega^\bullet(M^{|\vG_0|})\longrightarrow
        \Omega^\bullet(M)
		\]
	 such that
		\[
		O_{\vG}(\Phi)
		=
		\int_M
		D_{\vG}\Phi.
		\]
	\end{thm}

	It follows from \Cref{thm:anomaly operator kahler} that
	\begin{prop}\label{prop:boundary reduction kahler}
		Let $\vG$ be a connected decorated directed graph, and let
$\vG'\subseteq\vG$ be a subgraph.  Then there exists a
degree-$|\vG'_0|$ Dolbeault holomorphic Lagrangian density
\[
        D_{\vG'}:
        \Omega^\bullet(M^{|\vG_0|})
        \longrightarrow
        \Omega^\bullet(M^{|(\vG/\vG')_0|}),
\]
acting on the factors indexed by $\vG'_0$, such that
\begin{equation}\label{eq:boundary reduction kahler}
        \int_{\varrho_{\vG_1'}^{-1}(0)
        \times
        \widetilde{[0,+\infty)^{|\vG_1\setminus\vG'_1|}}}
        \widehat W(\vG,\Phi)
        =
   W\big(\vG/\vG',D_{\vG'}\Phi\big).
\end{equation}

	\end{prop}
	
	\begin{proof}
		The proof follows exactly as in \cite[Proposition 11]{wang2025feynman}, using
		\Cref{thm:anomaly operator kahler} in place of
		\cite[Theorem 6]{wang2025feynman}. 
	\end{proof}
	Thus, together with \eqref{eq: Large L has no contribution}, we immediately have
	\begin{cor}\label{cor:kahler anomaly quotient graphs}
		Let $\vG$ be a connected decorated directed graph. Then
		\[
		(\bar\partial W)(\vG,\Phi)
		= (-1)^{|\vG_1|}\Big(\sum_{e\in\vG_1}
        \sum_k
        W\bigl(\vG\setminus\{e\},\,H_{e,k}\wedge\Phi\bigr)+
\sum_{\vG'\subseteq\vG}
		(-1)^{\sigma(\vG',\vG/\vG')}
		W\big(\vG/\vG',D_{\vG'}\Phi\big)\Big).
		\]
        Here $\vG/\vG'$ is the graph obtained by contracting $\vG'$ to one vertex.
	\end{cor}
	
	\subsection{Proof of \Cref{thm:anomaly operator kahler}}
	First, we recall:
		\begin{thm}[Theorem 2.3.5 in \cite{hormander1983analysis}]\label{thm: finite jet}
			Let $Z\subset M$ be a smooth embedded submanifold, with $M$ compact.  Let
			$W\in \mathcal D'(M)$ satisfy
			\[
			\operatorname{supp}W\subset Z .
			\]
			Then there exists $N\geq 0$ such that $W(f)$ depends only on the normal
			$N$-jet of $f$ along $Z$.  More precisely, in a tubular coordinate chart
			\[
			(\bfx,\bfz)\in \mathbb R^k\times \mathbb R^\ell,
			\quad
			Z=\{\bfz=0\},
			\]
			we write, for $\bfi=(i_1,\ldots,i_\ell)\in\mathbb N^\ell$,
			\[
			\partial_\bfz^{\bfi}
			=
			\frac{\partial^{|\bfi|}}
			{\partial z_1^{i_1}\cdots \partial z_\ell^{i_\ell}} .
			\]
			If
			$
			\partial_\bfz^{\bfi} f(x,0)=0,
			 \forall |\bfi|\leq N
			$
			in every such chart, then
			\[
			W(f)=0 .
			\]
		\end{thm}

 To apply \Cref{thm: finite jet} in the proof of
\Cref{thm:anomaly operator kahler}, we establish
\Cref{prop:anomaly-current} and
\Cref{prop: anomaly support near diagonal} below.
    \begin{prop}\label{prop:anomaly-current}
Let $\vG$ be a decorated directed graph. The boundary anomaly
\[
        O_{\vG}(\Phi)
        =
        \int_{\varrho_{\vG_1}^{-1}(0)}
        \widehat W(\vG,\Phi),
        \quad
        \Phi\in\Omega^\bullet(M^{|\vG_0|}),
\]
defines a current on $M^{|\vG_0|}$. 
\end{prop}

\begin{proof}
This follows from the continuity of the map
$
        \Phi\longmapsto \widehat{W}(\vG,\Phi),
$
established in \Cref{smooth extension to schwinger3}, together with the
compactness of $\varrho^{-1}_{\vG_1}(0)$.

\end{proof}

The proposition below reduces the proof of
\Cref{thm:anomaly operator kahler} to the connected-graph case.
\begin{prop}\label{prop: anomaly vanishes on disconnected graph}
	If $\vG$ is disconnected, then $O_{\vG}(-)=0.$
\end{prop}
	\begin{proof}
			Indeed, suppose that
		\[
		\vG=\vG'\sqcup\vG''
		\]
		is a disjoint union. For decomposable test forms
		\[
		\Phi=\Phi'\wedge\Phi'',
		\quad
		\Phi'\in\Omega^\bullet(M^{|\vG'_0|}),\quad
		\Phi''\in\Omega^\bullet(M^{|\vG''_0|}),
		\]
		the graph form splits, up to the Koszul sign, as
		\[
		\widetilde W(\vG,\Phi)
		=
		\widetilde W(\vG',\Phi')\wedge
		\widetilde W(\vG'',\Phi'').
		\]
		Such forms span a dense subspace of
		$\Omega^\bullet(M^{|\vG_0|})$, so it is enough to consider them.
		\def\bfu{\mathbf{u}}
		
		In fact, the contribution of a disconnected graph to the boundary face
		$\varrho^{-1}_{\vG_1}(0)$ is $0$. The boundary face $\varrho^{-1}_{\vG_1}(0)$ corresponds to the simultaneous
		collapse of all Schwinger parameters $t_e$, $e\in\vG_1$, with one common
		scale. Under the natural map
		\[
		\mathrm{p}:
		\widetilde{(0,\infty)^{|\vG_1|}}
		\longrightarrow
		\widetilde{(0,\infty)^{|\vG'_1|}}
		\times
		\widetilde{(0,\infty)^{|\vG''_1|}},
		\]
		this face is mapped into the product boundary; more precisely,
		\[
		{\mathrm{p}}(\varrho_{\vG_1}^{-1}(0))
		\subset
		\varrho_{\vG_1'}^{-1}(0)\times \varrho_{\vG_1''}^{-1}(0)
		.
		\]
Let
\[
        I_{\vG_1}: \varrho^{-1}_{\vG_1}(0)\longrightarrow\resch
\]
be the inclusion map. Then by the definition of $\widehat{W}$, one has
\[
        I_{\vG_1}^*\widehat W(\vG,\Phi)
        =
        \pm\,
        {\mathrm{p}}^*
        \Big(
                I_{\vG_1'}^*\widehat W(\vG',\Phi')
                \wedge
                I_{\vG_1''}^*\widehat W(\vG'',\Phi'')
        \Big).
\]
		However,
		\[
		\dim \varrho_{\vG_1}^{-1}(0)
		=
		\dim \varrho_{\vG_1'}^{-1}(0)
		+
		\dim \varrho_{\vG_1''}^{-1}(0)+1.
		\]
		Therefore its
		integral over $\varrho^{-1}_{\vG_1}(0)$ vanishes.
	\end{proof}

\begin{prop}\label{prop: anomaly support near diagonal}
	If $\vG$ is connected, then $O_{\vG}(-)$ is support on the diagonal of the small diagonal $\triangle_{\vG}:=\{(p_1,\cdots,p_{\vG_0})\in M^{|\vG_0|}:p_1=\cdots p_{\vG_0}\}$.
\end{prop}
\begin{proof}
			Since $\vG$ is connected,  by the triangle inequality,
	if $\operatorname{supp}\Phi$ is disjoint from the small diagonal, then by
	continuity of $\sum_{e\in\vG_1}\rho_e^2$ and compactness of $\operatorname{supp}\Phi$ there exists $c>0$ such that on $\operatorname{supp}\Phi$,
	\[
	\sum_{e\in\vG_1}\rho_e^2\ge c.
	\]
	As a result, on $\{\varrho_{\vG_1}=\ep\}$, for some constant $c'>0$ \[
	\sum_{e\in\vG_1}\frac{\rho^2_e}{t_e}\geq \frac{\sum_{e\in\vG_1}\rho_e^2}{\max_{e'\in\vG_1}{t_{e'}}}\geq \frac{c'}{\ep}
	\]
	Hence, on $\{\varrho=\ep\}$, the Gaussian factors in the propagators give
	\[
	\big|\widetilde W(\vG,\Phi)\big|
	\le
	C e^{-c''/\ep}
	\]
	for some constant $c''>0$. Therefore the restriction of the boundary integrand to
	$\varrho_{\vG_1}^{-1}(0)$ vanishes, and $O_{\vG}(\Phi)=0$.	
\end{proof}

  By \Cref{thm: Propogator has regular expression}, the Feynman graph weight
$\widetilde W(\vG,\Phi)$ is locally of the following form:
		\begin{defn}[Admissible regular expression]\label{defn:admissible-regular-expression}
Let $\vG$ be a connected directed graph. Let $(U,\bfz)$ be a holomorphic
coordinate chart of $M$, and let $n=\dim_\C(M)$. We write
\[
        \bfZ=(\bfz_1,\bfz_2,\ldots,\bfz_{|\vG_0|})
        \in U^{|\vG_0|}.
\]
For each edge $e\in\vG_1$, set
\[
        \bfy_e
        :=
        \frac{\bar\bfz_{h(e)}-\bar\bfz_{t(e)}}{t_e}.
\]
An admissible regular expression on $U$ is a differential form on
$(0,\infty)^{|\vG_1|}\times M^{|\vG_0|}$ which can be locally expressed as a
finite linear combination of terms of the form
\[
        \exp\Big(
        -\sum_{e\in\vG_1}
        \frac{\rho^2(\bfz_{h(e)},\bfz_{t(e)})}{2t_e}
        \Big)
        \prod_{e\in\vG_1}
        \bfy_e^{\bfk_e}d\bfy_e^{\bfl_e}
        \wedge
        A(\bft,\bfZ,\bar\bfZ),
\]
where
\[
        \bfk_e\in\mathbb Z_{\ge0}^n,
        \qquad
        \bfl_e\in\{0,1\}^n,
        \qquad
        \bft=(t_e)_{e\in\vG_1},
\]
and $A\in\Omega^\bullet\big((0,\infty)^{|\vG_1|}\times U^{|\vG_0|}\big)$ is
smooth up to $\resch\times U^{|\vG_0|}$. We denote the space of admissible
regular expressions on $U$ by
\[
        \mathcal{RE}_{\vG}^{\mathrm{adm}}(U).
\]
Here
\[
        \bfy_e^{\bfk_e}
        :=
        \prod_{i=1}^n y_{e,i}^{k_{e,i}},
        \qquad
        d\bfy_e^{\bfl_e}
        :=
        \bigwedge_{\ell_{e,i}=1}dy_{e,i},
\]
with the wedge product ordered increasingly in $i$.
\end{defn}

	The result below is a local version of
\Cref{smooth extension to schwinger2}, and follows from a similar argument. It will be needed in the proof of
\Cref{thm:anomaly operator kahler}.\begin{thm}\label{thm:regular-expression-pushforward}
			Let
			$
			\Theta\in \mathcal{RE}_{\vG}^{\mathrm{adm}}(U),
			$
			then the fiber integral
			\[
			\int_{U^{|\vG_0|-1}}
			\Theta
			\]
			extends smoothly to  $\resch\times U$.
		\end{thm}

We introduce two subspaces of $\mathcal{RE}_{\vG}^{\mathrm{adm}}(U)$ below;
their relation is described in
\Cref{lem:anti-holomorphic-ideal-gains-time}.
    \begin{defn}\label{defn:normal-and-schwinger-ideals}
		Let $\vG$ be a connected directed graph. We work in the coordinates introduced in
\Cref{defn:admissible-regular-expression}. We denote by
\[
        \overline{\mathfrak m}_{\Delta}
        \mathcal{RE}_{\vG}^{\mathrm{adm}}(U)
\]
the class of finite sums of expressions of the form
\[
        \chi\cdot\,
        (\bar\bfz_j-\bar\bfz_1)^{\bfk}\,\Theta,
        \quad
        2\le j\le |\vG_0|,
        \quad
        \Theta\in\mathcal{RE}_{\vG}^{\mathrm{adm}}(U),
\]
where $\bfk=(k_1,\cdots, k_n)\in\mathbb Z_{\ge0}^n$ with $|\bfk|:=k_1+\cdots +k_n>0$, and
$\chi\in C_c^\infty(U^{|\vG_0|})$.

Similarly, we denote by
\[
        \mathfrak t_{\vG}
        \mathcal{RE}_{\vG}^{\mathrm{adm}}(U)
\]
the class of finite sums of expressions of the form
\[
        t_e\,\Theta_e,
        \quad
        e\in\vG_1,
        \quad
        \Theta_e\in\mathcal{RE}_{\vG}^{\mathrm{adm}}(U).
\]
	\end{defn}
		
	\begin{lem}\label{lem:anti-holomorphic-ideal-gains-time}
		If $\vG$  is connected, one has
		\[
		\overline{\mathfrak m}_{\Delta}
		\mathcal{RE}_{\vG}^{\mathrm{adm}}(U)
		\subset
		\mathfrak t_{\vG}
		\mathcal{RE}_{\vG}^{\mathrm{adm}}(U).
		\]
	\end{lem}

		\begin{proof}
			Since $\vG$ is connected, for each $j=2,\ldots,|\vG_0|$ we may choose a
			path $P_j$ in the underlying unoriented graph from $1$ to $j$. Along this
			path, we have the telescoping identity
			\[
			\bar\bfz_j-\bar\bfz_1
			=
			\sum_{e\in P_j}
			\varepsilon_{j,e}
			\big(\bar\bfz_{h(e)}-\bar\bfz_{t(e)}\big),
			\]
			where $\varepsilon_{j,e}\in\{1,-1\}$ depends only on whether the orientation
			of $e$ agrees with the chosen path direction.
			
			Hence
			\[
			\chi\cdot(\bar\bfz_j-\bar\bfz_1)^{\bfk}\Theta
			=
			\chi\cdot\big(\sum_{e\in P_j}
			\varepsilon_{j,e}\,t_e\bfy_e\big)^{\bfk}\Theta \in \mathfrak t_{\vG}
		\mathcal{RE}_{\vG}^{\mathrm{adm}}(U).
			\]	
		\end{proof}

	\def\bfJ{{\mathbf{J}}}

	\begin{proof}[Proof of \Cref{thm:anomaly operator kahler}]
By \Cref{prop: anomaly vanishes on disconnected graph}, we may assume that $\vG$ is a connected decorated directed graph.
Let $(U,\bfz)$ be a coordinate chart of $M$. We use it to introduce local
coordinates on $M^{|\vG_0|}$ near the small diagonal as follows. We keep
$\bfz_1$ as the coordinate on the first copy and set
\[
        \bfx_k:=\bfz_k-\bfz_1,
        \quad 2\le k\le |\vG_0|.
\]
Then the small diagonal is given locally by
\[
        \bfx_2=\cdots=\bfx_{|\vG_0|}=0 .
\]

By using a partition of unity, we may assume that the projection of
$\operatorname{supp}(\Phi)$ onto the first copy is contained in a compact
subset \[C\subset U.\]
By \Cref{prop:anomaly-current} and
\Cref{prop: anomaly support near diagonal}, $O_{\vG}$ is a current supported
on the small diagonal.  Thus by \Cref{thm: finite jet}, there exists $N\in\mathbb N$ such that
$O_{\vG}(\Phi)$ depends only on the $N$-jet of $\Phi$ in the normal
directions to the small diagonal.
Let $\eta:\mathbb R\to\mathbb R_+$ be a bump function such that
$\eta\equiv 1$ near $0$ and $\operatorname{supp}\eta$ is sufficiently
small. For the purpose of computing $O_{\vG}(\Phi)$, we may replace
$\Phi$ by the localized finite normal Taylor polynomial 
\[
        \eta\Big(\sum_{i=2}^{|\vG_0|}|\bfx_i|^2\Big)
        \sum_{\{\bfI,\bfJ\in \Z_{\geq0}^{n(|\vG_0|-1)}:|\bfJ|+|\bfI|\le N\}}
        \bfX^{\bfJ}\bar\bfX^{\bfI}\,
        \Phi_{\bfJ,\bfI}(\bfz_1,\bar\bfz_1),
\]
where $\bfX:=(\bfx_2,\cdots,\bfx_{|\vG_0|})$.
Here $\Phi_{\bfJ,\bfI}$ has compact support inside $C\subset U,$ so that this replacement has the same $N$-jet as $\Phi$. As long as the support of $\eta$ is small enough, the expressions above also have compact support inside $U^{|\vG_0|}.$

		By \Cref{thm: Propogator has regular expression}, each Taylor monomial above gives terms of the
		form
		\[
		\Theta\wedge
		\eta \bfX^\bfJ\bar\bfX^\bfI\,\Phi_{\bfJ,\bfI}(\bfz_1,\bar\bfz_1),
		\quad
		\Theta\in\mathcal{RE}_{\vG}^{\mathrm{adm}}(U).
		\]
		
		Suppose first that $|\bfI|>0$. Then the factor
		$\bar\bfX^\bfI$ contains at least one anti-holomorphic normal coordinate.
		By \Cref{lem:anti-holomorphic-ideal-gains-time}, it is a finite sum of terms
		\[
		t_e\,\Theta_e,
		\quad
		\Theta_e\in\mathcal{RE}_{\vG}^{\mathrm{adm}}(U).
		\]
	Applying \Cref{thm:regular-expression-pushforward}, each
		\[
		\int_{U^{|\vG_0|-1}}\eta\Theta_e
		\]
		extends smoothly to $\resch\times U$. Therefore
		\[
		\int_{U^{|\vG_0|-1}} t_e\,\eta\Theta_e
		=
		t_e
		\int_{U^{|\vG_0|-1}}\eta\Theta_e
		\]
		vanishes after restriction to the boundary face where all
		$t_e$, $e\in\vG_1$, collapse. Hence every monomial with
		$
		|\bfI|>0
		$
		makes no contribution to $O_{\vG}(\Phi)$. Thus only the terms with $\bfI=0$ can contribute. Therefore the
		local contribution to $O_{\vG}(\Phi)$ is a finite linear combination of
		\[
		\Phi_{\bfJ,0}(\bfz_1,\bar\bfz_1)
		=
		\left.
		\partial_{\bfz_2}^{\bfj_2}\cdots
		\partial_{\bfz_{|\vG_0|}}^{\bfj_{|\vG_0|}}
		\Phi
		\right|_{\bfx_2=\cdots=\bfx_{|\vG_0|}=0},
		\quad
		|\bfJ|=|(\bfj_2,\cdots,\bfj_{|\vG_0|})|\le N.
		\]
		The coefficients of this finite linear combination are obtained by applying
		\Cref{thm:regular-expression-pushforward} to the corresponding admissible
		regular expressions with the monomial $\eta\bfX^\bfJ$, then integrate over $\varrho_{\vG_1}^{-1}(0)$. Hence these coefficients
		depend smoothly on  $\bfz_1,\bar\bfz_1$. This proves the theorem.
	
	\end{proof}

		\bibliography{lib}
		\bibliographystyle{plain}
	\end{document}